\documentclass[11pt]{article}
\usepackage[margin=1in]{geometry}
\usepackage[T1]{fontenc}
\usepackage{setspace}
\usepackage{lmodern,microtype,amsmath,amssymb,amsthm,mathtools}
\usepackage{booktabs,enumitem,xcolor,graphicx,needspace,placeins}
\usepackage{tikz,pgfplots}
\pgfplotsset{compat=1.18}
\usepgfplotslibrary{fillbetween,groupplots}
\usetikzlibrary{arrows.meta,calc,decorations.pathreplacing,positioning,patterns,backgrounds}
\usepackage{algorithm,algpseudocode}
\usepackage[authoryear,round]{natbib}
\usepackage{xurl}
\usepackage{etoc}
\usepackage{bbm}
\usepackage[colorlinks=true,linkcolor=blue!45!black,citecolor=blue!45!black,urlcolor=blue!45!black]{hyperref}
\setlist{nosep,leftmargin=1.6em}
\algrenewcommand{\algorithmiccomment}[1]{\hfill\textcolor{blue!65!black}{$\triangleright$ #1}}
\newtheorem{theorem}{Theorem}[section]
\newtheorem{lemma}[theorem]{Lemma}
\newtheorem{proposition}[theorem]{Proposition}

\newtheorem{definition}[theorem]{Definition}

\theoremstyle{definition}\newtheorem{example}[theorem]{Example}
\theoremstyle{remark}
\newcommand{\STARS}{\textsf{STARS}}

\newcommand{\E}{\mathbb E}\newcommand{\Pp}{\mathbb P}
\newcommand{\one}{\mathbf 1}\newcommand{\Reg}{\operatorname{Reg}}
\newcommand{\clip}{\operatorname{clip}}\newcommand{\Var}{\operatorname{Var}}
\newcommand{\epsT}{\varepsilon_T}
\allowdisplaybreaks[2]
\title{{\sf Trading with the \STARS{}: Algorithm Design \& Spectrum of Fundamental Limits for Trading with Storage}}
\author{Jerry Anunrojwong\thanks{Yale University. Email: {\tt jerry.anunrojwong@yale.edu}} \and Akshit Kumar\thanks{University of Toronto. Email: {\tt akshit.kumar@utoronto.ca}} \and Rachitesh Kumar\thanks{Carnegie Mellon University. Email: {\tt rachitesh@cmu.edu}}}
\date{}
\begin{document}
\maketitle
\begin{abstract}
\normalsize
    We study an online trading problem where a trader, given a sequence of i.i.d. prices drawn from a known distribution $F$ on $[0,1]$, must make irrevocable buy, sell, or hold decisions subject to storage constraints. We investigate achievable algorithmic performance measured in terms of regret, the difference between the expected profit of the hindsight optimal policy that knows the entire price sequence and an online algorithm. We analyze finite atomic and continuous distributions characterized by their local behavior around the median which we capture using a parameter $\beta$. The parameter $\beta$ quantifies how the mass of prices accumulates around the distribution median. We identify a new driver of algorithmic performance, demonstrating that median gaps coupled with an initial inventory level of zero can force regret scaling of $\Omega(T^{(\beta + 1)/(2\beta+4)})$ --- establishing a novel spectrum of fundamental limits on algorithmic performance. We then study \STARS{}, short for {\sf S}torage {\sf T}rading by {\sf A}veraging {\sf R}epeatedly across multiple {\sf S}cenarios, which simulates possible future price scenarios to approximate the value-to-go function and make buy/sell/hold decisions. We show that \STARS{} obtain near-optimal algorithmic performance (upto poly-logarithmic factors) across a broad range of distributions.
    In particular, it achieves $O(\log T)$ regret for finite atomic prices, $\widetilde O(T^{\beta/(2\beta+2)})$ for continuous distributions without a median gap and $\widetilde O(T^{(\beta+1)/(2\beta+4)})$ for continuous distributions with a median gap for $\beta \geq 0$.

\end{abstract}
\etocdepthtag.toc{main}
\section{Introduction}
\label{sec:introduction}

Online resource allocation problems arise in revenue management, online matching, e-commerce fulfillment, and online marketplaces. In the canonical formulation, a decision maker begins with an exogenously given stock of resources and decides how to allocate them as requests arrive. Accepting a request consumes a resource; rejecting it preserves that resource for subsequent requests. A central question is how much the decision maker loses, relative to a hindsight optimizer, by making these decisions without knowing future requests; see \cite{balseiro2024survey}.

Many operational settings depart from this one-way allocation paradigm. A decision maker must procure resources \emph{endogenously} and decide how to allocate them to exploit future selling opportunities. Inventory then becomes central. Buying increases inventory, selling decreases it, and both decisions must respect capacity limits. A warehouse manager chooses when to acquire, store, and sell a commodity. A battery operator makes the analogous charging, idling, and discharging decisions. Such procurement-and-allocation decisions also arise for intermediaries whose inventory consists of contractual service capacity.

A concrete, albeit stylized, motivation is freight brokerage. A broker brings together shippers requiring transportation and carriers supplying it. For concreteness, consider a single-lane broker who can buy one unit of carrier service, sell one unit of that service to a shipper, or hold at each time step. The broker begins with no carrier capacity reserved. A purchase secures one unit of contractual carrier service, a sale allocates that unit to shipper demand, and holding keeps the commitment unallocated. The intermediary limits its unallocated commitments to $B$ units, which can be interpreted as an internal risk or working-capital limit. We deliberately abstract from the details of freight brokerage to focus on the core operational question: how does one make online trading decisions when resources need to be {\it endogenously} procured and allocated?

We study a prototypical version of this problem. Over a horizon of $T$ time steps, a trader observes i.i.d. prices drawn from a known distribution and can buy one unit, sell one unit, or hold at each time step. Inventory lies between zero and a capacity $B$, and the terminal inventory must equal the initial inventory $y_0$. Performance is measured by regret relative to the best feasible trading sequence with full knowledge of the prices.

The closest classical analogue is the multisecretary problem (see \cite{kleinberg2005multiple,arlotto2019uniformly}), where a decision maker observes a sequence of $n$ candidate valuations and makes irrevocable accept/reject decisions, subject to accepting at most $k$ candidates. The connection to trading with storage is particularly close at the level of the hindsight problem. 
%Suppose that the initial inventory is a linear distance from both boundaries, $0$ and $B$, and that capacity is proportional to the horizon, for example, $B=T/2$ 
Suppose that both the initial inventory $y_0$ and the unused storage capacity $B-y_0$ are proportional to the horizon $T$. If we temporarily ignore the pathwise inventory constraints, while retaining the requirement that terminal inventory equal initial inventory, an optimal hindsight trading strategy buys at the $\lfloor T/2\rfloor$ lowest prices, sells at the $\lfloor T/2\rfloor$ highest prices, and holds once if $T$ is odd. Thus, the empirical median acts as a hindsight threshold, with ties broken at random. {Selecting the sale periods is the same as selecting the accepted candidates} in a multisecretary problem with $n=T$ candidates and $k=\lfloor T/2\rfloor$ acceptances when valuations are known in advance. Under i.i.d. prices, the temporal ordering of these selected trades produces inventory fluctuations of order $\sqrt{T}$. {Because both buffers are proportional to $T$, they dominate the $\sqrt{T}$-scale inventory fluctuations. The relaxed hindsight strategy is therefore feasible with high probability, connecting the two hindsight problems.}
%Linear slack from both boundaries therefore makes this trading sequence feasible with high probability, connecting the two hindsight problems.

However, when the initial inventory level is zero, the problems differ substantially. 
%At $y_0=0$, the static threshold solution 
{The relaxed hindsight strategy} may prescribe a sale before the first purchase{, making it infeasible}. The hindsight optimizer must now account for when buying and selling opportunities arrive, {rather than simply selecting the lowest and highest prices.}
%and its thresholds depend on both the inventory and the remaining horizon. 
Unlike the multisecretary problem, the optimum can no longer be characterized by a single static threshold. The connection therefore fails already at the hindsight level and is even less transparent for online policies, which act without knowledge of future prices. This is the central difficulty introduced by storage constraints.

These observations motivate three questions. \textit{(i) How do storage constraints fundamentally change achievable performance relative to the classical dynamic resource allocation paradigm?} \textit{(ii) What are the fundamental drivers of the attainable regret scaling?} \textit{(iii) What simple algorithmic principles achieve these regret scalings?}

\subsection{Main contributions}
\label{subsec:contributions}

\begin{enumerate}[label={$(\roman*)$},leftmargin=*,itemsep=0.7em]
    \item \textit{Novel spectrum of fundamental limits}: We identify the \textit{initial inventory} and the \textit{local geometry of the price distribution} as two fundamental drivers of algorithmic performance. For regular continuous distributions (Definition~\ref{def:continuous-family}), the parameter $\beta\geq0$ describes how rapidly probability mass accumulates near the median, and $\Delta$ denotes the median gap. When $y_0\in[\varepsilon B,(1-\varepsilon)B]$ and $B=\rho T$, for fixed $\varepsilon,\rho>0$, the initial inventory provides linear slack from both boundaries. Proposition~\ref{prop:interior-transfer} shows that a feasible implementation of the \textsf{CwG} principle of \cite{besbes2025} then achieves regret $\widetilde O(T^{\beta/(2\beta+2)})$. Its inventory avoids the boundaries with high probability, recovering the spectrum of regret exponents from the multisecretary problem.

    A new spectrum emerges when the system starts at the boundary, i.e., $y_0 = 0$. Theorem~\ref{lb:thm-positive-gap} shows that, for every $\beta \geq 0$ and every fixed regular continuous distribution with a positive median gap, every feasible online policy starting empty incurs $\Omega(T^{(\beta+1)/(2\beta+4)})$ regret. %Theorem~\ref{lb:thm-positive-gap} shows, for every $\beta\geq0$, a bounded-density distribution with a positive median gap, for which every feasible online policy, including the optimal online policy, incurs $\Omega(T^{(\beta+1)/(2\beta+4)})$ regret. 
    Since $\beta\geq0$, the new exponent is at least $1/4$. The contrast is particularly striking when the density is positive at the median-gap edges ($\beta = 0$). For the same price distribution and storage capacity, starting with both inventory and unused capacity proportional to $T$ permits polylogarithmic regret, whereas starting empty forces every feasible online policy to incur regret of at least order $T^{1/4}$.

    This distinction also sheds a new perspective on \textit{degeneracy} in resource allocation. 
    %Consider the fluid relaxation without intermediate inventory constraints. 
    Consider the fluid relaxation that ignores the intermediate inventory constraints and requires purchases and sales to balance only in expectation.
    For a continuous distribution with a median gap, every price in the gap is an optimal shadow price for the balance constraint.
    The fluid problem therefore has {\it multiple} optimal dual prices, a form of degeneracy studied in the resource allocation literature (see \cite{jasin2012re,balseiro2024survey}). For the multisecretary problem, \citet{besbes2025} show that gaps do not change the achievable polynomial regret exponent, although logarithmic factors may differ. In contrast, for regular continuous distributions, a positive median gap forces a strictly larger polynomial regret exponent when trading starts from zero inventory. 
    Figure~\ref{fig:fundamental-drivers} summarizes this distinction.

    \begin{figure}[t]
    \centering
    \begingroup
    \definecolor{twsgray}{HTML}{5D6770}
    \definecolor{twsdensityred}{HTML}{C62828}

    % Native PGFPlots panels: the formulas, axis dimensions, and final width
    % match the former standalone density plots, without external PDF assets.
    \pgfplotsset{
        twsdensityaxis/.style={
            width=3.0cm,
            height=2.8cm,
            axis lines=left,
            axis line style={black,->,line width=0.55pt},
            xmin=0,
            xmax=1.06,
            ymin=0,
            xtick=\empty,
            ytick=\empty,
            xlabel={$p$},
            ylabel={$f(p)$},
            xlabel style={font=\scriptsize,at={(axis description cs:0.96,-0.08)},anchor=north},
            ylabel style={font=\scriptsize,at={(axis description cs:-0.10,0.50)},anchor=south},
            clip=false,
        },
        twsdensitycurve/.style={twsdensityred,line width=1.05pt,no marks},
    }
    \newcommand{\twsdensityplot}[2]{%
        \resizebox{1.72cm}{!}{%
            \begin{tikzpicture}
                \begin{axis}[twsdensityaxis,ymax=#1]
                    #2
                \end{axis}
            \end{tikzpicture}%
        }%
    }
    \newcommand{\twsgapplots}[2]{%
        \begin{scope}[shift={(#1,#2)}]
            \node[anchor=south west,inner sep=0] at (0,0)
                {\twsdensityplot{1.70}{%
                    \addplot[twsdensitycurve,domain=0:0.4,samples=2] {1.25};
                    \addplot[twsdensitycurve,domain=0.6:1,samples=2] {1.25};
                }};
            \node[anchor=south west,inner sep=0] at (1.92,0)
                {\twsdensityplot{3.00}{%
                    \addplot[twsdensitycurve,domain=0:0.4,samples=51] {6.25*(0.4-x)};
                    \addplot[twsdensitycurve,domain=0.6:1,samples=51] {6.25*(x-0.6)};
                }};
            \node[anchor=south west,inner sep=0] at (3.84,0)
                {\twsdensityplot{4.50}{%
                    \addplot[twsdensitycurve,domain=0:0.4,samples=51] {23.4375*(0.4-x)^2};
                    \addplot[twsdensitycurve,domain=0.6:1,samples=51] {23.4375*(x-0.6)^2};
                }};
            \node[font=\scriptsize,anchor=south] at (0.86,1.30) {$\beta=0$};
            \node[font=\scriptsize,anchor=south] at (2.78,1.30) {$\beta=1$};
            \node[font=\scriptsize,anchor=south] at (4.70,1.30) {$\beta=2$};
        \end{scope}
    }

    \newcommand{\twsnogapplots}[2]{%
        \begin{scope}[shift={(#1,#2)}]
            \node[anchor=south west,inner sep=0] at (0,0)
                {\twsdensityplot{1.35}{%
                    \addplot[twsdensitycurve,domain=0:1,samples=2] {1};
                }};
            \node[anchor=south west,inner sep=0] at (1.92,0)
                {\twsdensityplot{2.40}{%
                    \addplot[twsdensitycurve,domain=0:1,samples=101] {4*abs(x-0.5)};
                }};
            \node[anchor=south west,inner sep=0] at (3.84,0)
                {\twsdensityplot{3.60}{%
                    \addplot[twsdensitycurve,domain=0:1,samples=101] {12*(x-0.5)^2};
                }};
            \node[font=\scriptsize,anchor=south] at (0.86,1.30) {$\beta=0$};
            \node[font=\scriptsize,anchor=south] at (2.78,1.30) {$\beta=1$};
            \node[font=\scriptsize,anchor=south] at (4.70,1.30) {$\beta=2$};
        \end{scope}
    }

    \resizebox{0.98\textwidth}{!}{%
    \begin{tikzpicture}[x=1cm,y=1cm]
        % White cells only.
        \fill[white] (2.25,0.00) rectangle (14.75,6.80);
        \draw[twsgray!70,line width=0.75pt] (2.25,0.00) rectangle (14.75,6.80);
        \draw[twsgray!70,line width=0.75pt] (8.50,0.00) -- (8.50,6.80);
        \draw[twsgray!70,line width=0.75pt] (2.25,3.40) -- (14.75,3.40);
        \draw[blue,line width=1.35pt]
            (8.50,3.40) rectangle (14.75,6.80);

        % Rotated row labels.
        \node[rotate=90,align=center,font=\small] at (1.25,5.10)
            {\textbf{Positive gap}\\ $\Delta>0$};
        \node[rotate=90,align=center,font=\small] at (1.25,1.70)
            {\textbf{No gap}\\ $\Delta=0$};

        % Three native density panels in each regime.
        \twsgapplots{2.58}{5.08}
        \twsgapplots{8.83}{5.08}
        \twsnogapplots{2.58}{1.68}
        \twsnogapplots{8.83}{1.68}

        % Left: achievable CwG-transfer bounds. Right: boundary results.
        \node[font=\scriptsize] at (3.44,4.60)
            {$\mathcal O(\log^2 T)$};
        \node[font=\scriptsize] at (5.36,4.60)
            {$\tilde{\mathcal O}(T^{\frac{1}{4}})$};
        \node[font=\scriptsize] at (7.28,4.60)
            {$\tilde{\mathcal O}(T^{\frac{1}{3}})$};

        \node[font=\scriptsize,text=blue] at (9.69,4.60)
            {$\tilde{\Theta}(T^{\frac{1}{4}})$};
        \node[font=\scriptsize,text=blue] at (11.61,4.60)
            {$\tilde{\Theta}(T^{\frac{1}{3}})$};
        \node[font=\scriptsize,text=blue] at (13.53,4.60)
            {$\tilde{\Theta}(T^{\frac{3}{8}})$};

        \node[font=\scriptsize] at (3.44,1.20)
            {$\mathcal O(\log T)$};
        \node[font=\scriptsize] at (5.36,1.20)
            {$\mathcal O(T^{\frac{1}{4}})$};
        \node[font=\scriptsize] at (7.28,1.20)
            {$\mathcal O(T^{\frac{1}{3}})$};

        \node[font=\scriptsize,text=twsgray] at (9.69,1.20)
            {$\tilde{\mathcal O}(1)$};
        \node[font=\scriptsize,text=twsgray] at (11.61,1.20)
            {$\tilde{\mathcal O}(T^{\frac{1}{4}})$};
        \node[font=\scriptsize,text=twsgray] at (13.53,1.20)
            {$\tilde{\mathcal O}(T^{\frac{1}{3}})$};

        % Interior and boundary policies are explicitly distinguished.
        \node[font=\scriptsize] at (5.38,3.82)
            {\cite{besbes2025}, all $\beta\geq0$:\quad
             $\tilde{\mathcal O}(T^{\frac{\beta}{2\beta+2}})$};
        \node[font=\scriptsize,text=blue] at (11.63,3.82)
            {Optimal regret, all $\beta\geq 0$:\quad
             $\tilde{\Theta}(T^{\frac{\beta+1}{2\beta+4}})$};
        \node[font=\scriptsize] at (5.38,0.42)
            {\cite{besbes2025}, all $\beta\geq0$:\quad
             $\tilde{\mathcal O}(T^{\frac{\beta}{2\beta+2}})$};
        \node[font=\scriptsize,text=twsgray] at (11.63,0.42)
            {\STARS{} upper bound, all $\beta\geq 0$:\quad
             $\tilde{\mathcal O}(T^{\frac{\beta}{2\beta+2}})$};

        % Starting-inventory labels below the matrix.
        \node[align=center,font=\small] at (5.38,-0.48)
            {\textbf{Linearly interior}\\[-1pt]
             $y_0\wedge(B-y_0)\ge\eta T$};
        \node[align=center,font=\small] at (11.63,-0.48)
            {\textbf{At a boundary}\\[-1pt] $y_0=0$};
    \end{tikzpicture}%
    }
    \endgroup
    \caption{Regret scaling by starting inventory and the local geometry of the price distribution.}
    \label{fig:fundamental-drivers}
\end{figure}
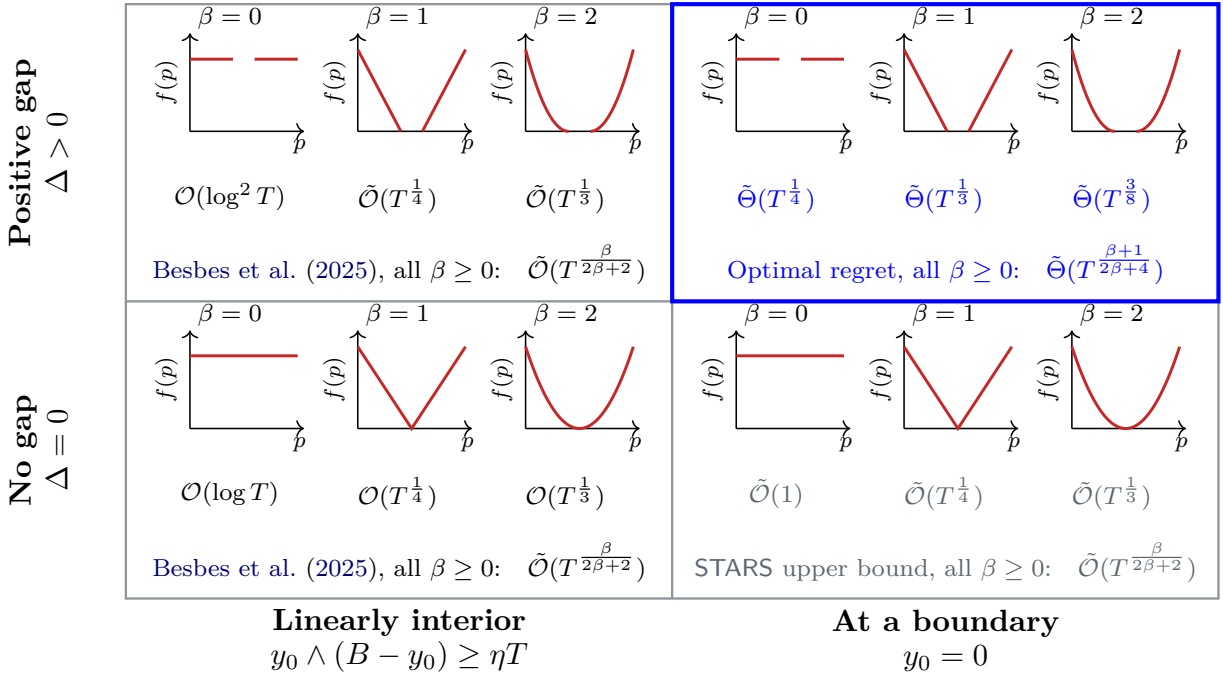

    \FloatBarrier

    \item \textit{Fragility of boundary-unaware algorithms}: We next show that algorithmic ideas inherited directly from the dynamic resource allocation literature without any consideration for the boundary constraint can be highly sub-optimal. Consider the classical certainty-equivalent heuristic which continually re-solves a fluid problem and steers inventory toward the terminal target. For uniform prices, Proposition~\ref{ce:theorem} shows that it achieves the classical logarithmic regret rate when the initial inventory is a linear distance from both storage boundaries.

    The same heuristic behaves very differently at $y_0=0$. The boundary-unaware certainty-equivalent heuristic does not value the inventory slack needed to exploit future selling opportunities. It may therefore sell when the inventory is close to zero and lose precisely the opportunities that a temporally aware hindsight policy preserves inventory for. Proposition~\ref{ce:theorem} proves that this boundary-unaware certainty-equivalent policy incurs $\Omega(\sqrt T)$ regret even for uniform prices. For the same uniform instance, \STARS{}, discussed next, has only $O(\log^2 T)$ regret by Theorem~\ref{thm:stars-upper}.

    \item \textit{A unifying algorithmic principle}: Motivated by the fragility of the boundary-unaware certainty-equivalent heuristic, we study an approximate dynamic programming approach based on simulated future prices. At each time step, the trader draws multiple future scenarios, computes the best feasible continuation for each current action in each scenario, and implements the action with the highest average total profit, including the current trade. We call this algorithm \STARS{}, short for {\sf S}torage {\sf T}rading by {\sf A}veraging {\sf R}epeatedly across multiple {\sf S}cenarios. Given a price simulator and a hindsight optimizer, the algorithm requires no distribution-specific fine-tuning. Its buy/hold/sell regions arise endogenously from the scenario continuation values.

    For expositional simplicity and motivated by practical settings, we focus on $y_0=y_T=0$. With capacity proportional to the horizon and sufficiently many scenarios at each time step, \STARS{} achieves $\mathcal{O}(\log T)$ regret for finite atomic distributions, $\widetilde{\mathcal{O}}(T^{\beta/(2\beta+2)})$ for regular continuous prices without a median gap, and $\widetilde{\mathcal{O}}(T^{(\beta+1)/(2\beta+4)})$ for regular continuous prices with a median gap (Theorem~\ref{thm:stars-upper}). 
% On the canonical positive-gap instances,
% {For every fixed distribution in the continuous positive-gap family considered here,}
For positive-gap distributions,
this guarantee matches our fundamental limit up to logarithmic factors.

% Specific distribution examples reveal how the algorithm protects inventory. For finite atomic examples with a median gap, scenario optimization generates a {\it logarithmic inventory buffer}: below a reserve level that grows logarithmically with the remaining horizon, the policy favors buying or holding near the upper gap edge, even at prices at which it would sell with more inventory. It thereby preserves stock for more valuable selling opportunities.
    Specific examples illustrate how \STARS{} protects inventory. In finite-price examples with a median gap, the policy may hold or buy at a moderately high price when inventory is low, but sell at the same price when inventory is plentiful. This behavior preserves stock for more valuable future selling opportunities. The inventory level separating these behaviors grows logarithmically with the remaining horizon.
% , creating a logarithmic inventory buffer.
For uniform prices, protection takes a softer form. When inventory is small relative to the square root of the remaining horizon and far from the upper boundary, the expected upward drift is of inverse-inventory order: the smaller the inventory, the stronger the push away from zero. This \textit{Bessel-type repulsion} weakens as inventory grows. Section~\ref{sec:proof-ideas} develops both interpretations.

    \item \textit{Key technical innovation}: We follow the compensated-coupling approach of \citet{vera2021bayesian}: regret is the sum of \textit{marginal compensations}, each measuring the loss from making hindsight take the online action at the online policy's current inventory (Proposition~\ref{prop:main-decomposition}). In the multisecretary analyses of \citet{bray2025} and \citet{besbes2025}, a crucial step is to bound expected marginal compensation uniformly over the current state by a quantity that decreases with the remaining horizon. Summing this time-dependent bound gives the regret guarantee. This simplification is unavailable here. For uniform prices, Proposition~\ref{prop:main-nonvanishing} shows that expected marginal compensation at zero inventory remains bounded below by a positive constant as the remaining horizon grows. We must therefore account jointly for \textit{how often the policy visits a state} and \textit{how much compensation it pays there}. Our analysis derives both the compensation and the inventory movement from the same hindsight marginal values, then combines them along the policy's inventory path. Section~\ref{sec:proof-ideas} explains this compensation-weighted accounting and how it yields guarantees for the different distribution classes.
\end{enumerate}
\FloatBarrier

\subsection{Related literature}
\label{subsec:related-literature}
\paragraph{\it Regret and distributional geometry in resource allocation.}
For the multisecretary problem with a fixed finite support, \cite{arlotto2019uniformly} develop a \textsf{Budget-Ratio} policy with regret bounded uniformly over the horizon and hiring budget. For regular continuous valuations, \cite{bray2025} establishes logarithmic regret by controlling hindsight shadow prices. In network revenue management, \cite{vera2021bayesian,bumpensanti2020re} obtain bounded loss by combining re-solving with thresholding. These results show how adaptivity and distributional structure improve on square-root losses; see \cite{balseiro2024survey} for a broader survey. 
\cite{besbes2025} identify a spectrum of multisecretary regret rates governed by mass accumulation around gaps. Their \textsf{CwG} principle and clustered-distribution framework motivate our focus on median geometry, and their \textsf{RAMS} algorithm provides the simulation-based idea we specialize to our problem with storage constraints. We allow inventory to be replenished as well as depleted and require feasibility at every time step.
We identify a different set of fundamental limits: a continuous distribution with a positive median gap can impose an additional polynomial penalty when starting close to the boundary. Our regret decomposition builds on the compensated-coupling framework of \cite{vera2021bayesian}, which compares online and hindsight actions from the same state and expresses regret as the sum of the resulting losses. In the multisecretary setting, \cite{bray2025} relate one-step loss to uncertainty in hindsight marginal values, while \cite{besbes2025} use marginal-compensation bounds to transfer guarantees to \textsf{RAMS}. We retain this decomposition, but must also control the inventory process generated by \STARS{}. We combine state-dependent compensation with the time spent in each state. 

\paragraph{\it Simulation-based and model-predictive policies.}
\cite{talluri1999randomized} compute network bid prices by averaging dual prices from demand-scenario linear programs. Model predictive control provides the broader approach of repeatedly optimizing a future plan and implementing its current action \citep{mayne2000,mesbah2016}. Relatedly, \cite{sinclair2023hindsight} use hindsight planning to learn policies under exogenous uncertainty. 
In a related dynamic matching context, \cite{chen2026feature} introduce \textsf{SOAR}, which samples future demand, solves a matching problem with the remaining supply, and implements the current assignment. Its expected average reward is characterized by an average of finite-sample hindsight values. \cite{kumar2026dark} replace the simulator with realized demand observations in \textsf{ROAR}, obtaining a related identity and regret guarantees without prior distributional knowledge. \STARS{} uses the same forward-looking idea to preserve inventory, but averages multiple scenarios from a known price distribution. The benchmarks also differ: those papers study per-match loss relative to the fluid problem, whereas we study total loss relative to the hindsight optimum which is a tighter benchmark. The ordering constraints on purchases and sales lead to a totally different regret argument and analytical techniques.

\paragraph{\it Prophet inequalities for trading.}
The trading-prophet literature studies the closely related problem of repeated buying and selling against a hindsight benchmark. \cite{correa2026trading} establish competitive guarantees for stochastic prices, \cite{azar2026initial} examine initial capital and transaction costs, and \cite{aggarwal2026asymmetric} quantify the effects of capacity and initial inventory with asymmetric buying and selling prices. Their principal guarantees identify the fraction of hindsight profit achievable uniformly over admissible price distributions. We build on this perspective to ask \emph{which properties of a fixed distribution determine the difficulty of online trading}. With unit trades and capacity proportional to the horizon, our analysis identifies median gaps, local mass accumulation, and initial inventory as joint drivers of regret scaling. Across finite atomic and regular continuous distributions, these features give rise to different regret scalings. In particular, a continuous median gap coupled with zero initial inventory creates a new polynomial lower-bound spectrum, while linear initial slack permits smaller exponents (Theorem~\ref{lb:thm-positive-gap} and Proposition~\ref{prop:interior-transfer}). This distribution-sensitive classification explains how price geometry and inventory conditions affect the cost of online decision making.

%Our algorithmic focus is repeated scenario optimization, motivated by access to a price simulator. The trading-prophet algorithms include threshold rules, with sample-based variants in \cite{correa2026trading}, and LP-based randomized rounding in \cite{aggarwal2026asymmetric}. \STARS{} optimizes simulated continuations from the current inventory and averages their values to choose the current action. We analyze how this procedure generates inventory- and horizon-dependent holding, obtaining distribution-sensitive regret guarantees with the same decision rule across the three families (Theorem~\ref{thm:stars-upper}).

%\paragraph{\it Storage control and closed resource networks.} \cite{secomandi2010storage} study storage with capacity and injection/withdrawal limits under Markov prices and characterize an optimal policy through two inventory targets. This provides an antecedent for the buy/hold/sell structure; our focus is its loss relative to sample-path hindsight. In closed resource networks, \cite{kanoria2024mirror} develop Mirror Backpressure without knowledge of demand rates. Their nonlinear congestion prices protect scarce inventory and yield near-optimal performance, including from unfavorable initial configurations. This is related to the boundary protection generated by \STARS{} endogeneously. %Their benchmark is an optimal online controller in a reusable-resource network, whereas ours is a balanced-endpoint hindsight trader. We connect scenario-induced holding to additive regret rates governed by local price geometry.

\section{Model and Classes of Price Distribution}
\label{sec:model}

\subsection{Feasible trading and regret}
The horizon is $T\in\mathbb N$, capacity is $B=\lceil\rho T\rceil$ for fixed $\rho\in(0,1)$, and $y_0\in\{0,\ldots,B\}$ is the initial inventory. Prices $p_1,\ldots,p_T$ are i.i.d. with known CDF $F$ supported on $[0,1]$. At time step $t$, the trader observes $p_t$ and chooses $x_t\in\{-1,0,+1\}$, corresponding to sell, hold, and buy. Profit is $-p_tx_t$, and
\begin{align}
 y_t=y_{t-1}+x_t,\qquad 0\le y_t\le B,\qquad y_T=y_0.
 \label{eq:model-dynamics}
\end{align}
Write $\tau_t=T-t+1$ for the remaining number of time steps,
including the current decision. We use $\tau$ for the remaining
length of a generic continuation problem and $y$ for its inventory;
the continuation after the decision at $t$ has $\tau_t-1$ prices.
We reserve $n$ and $k$ for the horizon and selection quota of a
multisecretary instance. From a feasible pre-decision state $y$, the actions preserving a feasible completion are exactly
\begin{align}
 \mathcal A_t(y)=\{x\in\{-1,0,1\}:0\le y+x\le B,\ |y+x-y_0|\le T-t\}.
 \label{eq:feasible-actions}
\end{align}
The constraint $|y+x-y_0|\le T-t$ ensures that the trader can still return to $y_0$ by time $T$.

Let $\zeta_t$ encode the policy's fresh randomization, independent across time and independent of all actual prices. Put $\mathcal F_t=\sigma(p_s,\zeta_s:s\le t)$ and let $\mathcal F_0$ be trivial. A feasible online policy $\pi$ chooses an $\mathcal F_t$-measurable $x_t^\pi\in\mathcal A_t(y_{t-1}^\pi)$. Denote this policy set by $\Pi(B,T,y_0)$. For a price sequence $\mathbf p$, hindsight chooses the best \emph{feasible} action sequence:
\begin{align*}
 V_1^{\mathsf{hs}}(B,T,y_0;\mathbf p)
 =\max\left\{\sum_{t=1}^T-p_tx_t:\ (x_t,y_t)_{t=1}^T\text{ satisfy \eqref{eq:model-dynamics}}\right\}.
\end{align*}
Note that the maximum exists since the feasible action set is finite and contains the all-hold sequence. The expected regret for a policy $\pi$ is defined as
\begin{align*}
 \Reg(\pi;B,T,y_0)=\E\left[V_1^{\mathsf{hs}}(B,T,y_0;\mathbf p)-\sum_{t=1}^T-p_tx_t^\pi\right].
\end{align*}
The expectation includes actual prices and policy randomization. We denote $R_T(\pi;y_0)=\sum_{t=1}^T-p_tx_t^\pi$ as the realized trading profit for policy $\pi$, and we have that
\begin{align*}
 \Reg^*(B,T;y_0)=\inf_{\pi\in\Pi(B,T,y_0)}\Reg(\pi;B,T,y_0),\qquad
 \Reg^*(B,T)=\sup_{y_0\in\{0,\ldots,B\}}\Reg^*(B,T;y_0).
\end{align*}
For real numbers $a,b$, write $a\wedge b=\min\{a,b\}$ and
$a\vee b=\max\{a,b\}$. All logarithms are natural. We state horizon-dependent rate bounds for $T\ge2$; when $T=1$, the equal initial and terminal inventories force a hold, so every feasible policy has zero regret.

The distribution $F$ is held fixed in this notation. The model and the \STARS{} algorithm (presented in Section~\ref{sec:stars}) allow any initial inventory.
% We compare interior and zero initial inventory levels in Section~\ref{lb:sec-fundamental-limits}; for expositional simplicity and motivated by practical applications, our \STARS{} upper bounds focus on $y_0=0$.
Section~\ref{lb:sec-fundamental-limits} compares regret guarantees for linearly interior and zero initial inventory. Our regret guarantees for \STARS{} in Theorem~\ref{thm:stars-upper} assume $y_0=y_T=0$.

\subsection{Three regimes of median geometry}
We consider three distribution families, motivated by the $(\beta, \varepsilon, \delta)-$clustered distributions of \cite{besbes2025}. Our definitions focus on how probability mass accumulates near the median, since this determines the value of holding inventory for future trades. The first family consists of finite atomic distributions while the other two consist of regular continuous distributions, with and without a median gap. In all three families, the distribution is fixed independently of $T$.

\begin{definition}[Median edges]\label{def:median-edges}
For $0<q<1$, define $F^{-1}(q)=\inf\{p\in[0,1]:F(p)\ge q\}$. The lower and upper median edges are $a_F^-=F^{-1}(1/2)$ and $a_F^+=\lim_{q\downarrow1/2,\ q>1/2}F^{-1}(q)$, respectively. The median gap has width $\Delta_F=a_F^+-a_F^-$.
\end{definition}
The interval $[a_F^-,a_F^+]$ is the set of medians of $F$.
% When its endpoints differ, the intervening open interval contains no probability mass.
A positive median gap separates the lower and upper halves of the distribution by a price interval containing no probability mass. This is the gap construction of \citet[Section~3.1]{besbes2025}, specialized to quantile $1/2$.
Note that the definition identifies the gap specifically at the median, allowing other gaps elsewhere in the support. We drop the subscript $F$ and abbreviate the edges by $a^-,a^+$ when $F$ is understood.

\begin{definition}[Finite atomic distributions]\label{def:atomic-family}
% For $\mu\in(0,1/2]$ and $\Delta\in[0,1]$, let $\mathcal F^{\mathrm{at}}_{\mu,\Delta}$ contain the class of distributions with finite support $0\le A_1<\cdots<A_m\le1$, positive masses $w_i$ satisfying $\sum_iw_i=1$, and median gap $\Delta_F=\Delta$, whose quantiles satisfy $F^{-1}(1/2-q)=a^-$ and $F^{-1}(1/2+q)=a^+$ for every $0<q<\mu$.
% Note that the support and masses are fixed as $T$ varies.
{For $\mu\in(0,1/2]$ and $\Delta\in[0,1]$, let $\mathcal F^{\mathrm{at}}_{\mu,\Delta}$ be the class of distributions $F$ with finite support in $[0,1]$, median gap $\Delta_F=\Delta$, and quantiles satisfying $F^{-1}(1/2-q)=a_F^-$ and $F^{-1}(1/2+q)=a_F^+$ for every $0<q<\mu$.}
% The support and probability masses are fixed independently of $T$.
\end{definition}
Every fixed finite atomic distribution belongs to such a family for some $\mu>0$. If a cumulative mass equals $1/2$, the two adjacent support prices form a positive median gap. Otherwise $1/2$ lies strictly inside one atom, and its quantile interval has positive length on both sides of $1/2$. The parameter $\mu$ specifies the local quantile windows around the median. 

\begin{example}[Finite atomic prices, with and without a median gap]\label{ex:atomic}
% Take $(A_1,A_2,A_3,A_4)=(0.2,0.4,0.6,0.8)$. Equal masses $w_i=1/4$ give $a^-=0.4$, $a^+=0.6$, and $\Delta_F=0.2$. Changing the masses to $\mathbf{w} = (0.25,0.20,0.30,0.25)$ leaves the support unchanged but places the median inside the atom at $0.6$: now $a^-=a^+=0.6$ and $\Delta_F=0$.
{Consider prices taking values $0.2,0.4,0.6,0.8$. If each price has probability $1/4$, then $a^-=0.4$, $a^+=0.6$, and $\Delta_F=0.2$. If instead these prices have probabilities $0.25,0.20,0.30,0.25$, respectively, the median lies inside the atom at $0.6$: now $a^-=a^+=0.6$ and $\Delta_F=0$.} Both distributions satisfy Definition~\ref{def:atomic-family} with $\mu=0.05$.
% Figure~\ref{fig:example-atomic} shows why gaps between support prices need not be median gaps.
{Figure~\ref{fig:example-atomic} shows that the same support can yield either a positive median gap or no median gap, depending on the probability masses.}
\end{example}
% Example figure: finite atomic four-point distributions on a fixed support.
% Left column: uniform masses (positive median gap).
% Right column: perturbed masses (median gap closes).
% Top row: probability mass functions; bottom row: quantile functions.
\begin{figure}[!htb]
    \centering
    \begingroup
    \definecolor{twsred}{HTML}{C62828}
    \definecolor{twsgold}{HTML}{D97706}
    \definecolor{twsblue}{HTML}{315A7D}
    \definecolor{twslight}{HTML}{F7E8E8}
    \begin{tikzpicture}
        \begin{groupplot}[
            group style={
                group size=2 by 2,
                horizontal sep=1.55cm,
                vertical sep=1.20cm,
            },
            width=0.42\textwidth,
            height=0.26\textwidth,
            axis lines=left,
            axis line style={black,semithick},
            tick style={black},
            tick label style={font=\scriptsize},
            label style={font=\small},
            title style={font=\scriptsize\bfseries,align=center,yshift=1mm},
            xmin=0,
            xmax=1,
            enlargelimits=false,
            clip=false,
        ]

        % ------------------------------------------------------------------
        % (1,1) PMF, uniform masses: positive median gap.
        % ------------------------------------------------------------------
        \nextgroupplot[
            title={Positive median gap: $\Delta_F=0.2$},
            xlabel={$p$},
            ylabel={probability mass},
            ymin=0,
            ymax=0.40,
            xtick={0.2,0.4,0.6,0.8},
            xticklabels={$0.2$,$0.4$,$0.6$,$0.8$},
            ytick=\empty,
        ]
            \addplot[ycomb,mark=*,mark size=1.8pt,line width=1.3pt,color=twsblue]
                coordinates {(0.2,0.25) (0.4,0.25) (0.6,0.25) (0.8,0.25)};
            \node[font=\scriptsize,anchor=south,text=twsblue]
                at (axis cs:0.2,0.262) {$0.25$};
            \node[font=\scriptsize,anchor=south,text=twsblue]
                at (axis cs:0.4,0.262) {$0.25$};
            \node[font=\scriptsize,anchor=south,text=twsblue]
                at (axis cs:0.6,0.262) {$0.25$};
            \node[font=\scriptsize,anchor=south,text=twsblue]
                at (axis cs:0.8,0.262) {$0.25$};
            \draw[<->,>=Latex,semithick,color=twsgold]
                (axis cs:0.4,0.045)
                -- node[above,font=\scriptsize,fill=white,inner sep=1pt]
                    {$\Delta_F$}
                (axis cs:0.6,0.045);

        % ------------------------------------------------------------------
        % (1,2) PMF, perturbed masses: same support, no median gap.
        % ------------------------------------------------------------------
        \nextgroupplot[
            title={No median gap: $\Delta_F=0$},
            xlabel={$p$},
            ylabel={probability mass},
            ymin=0,
            ymax=0.40,
            xtick={0.2,0.4,0.6,0.8},
            xticklabels={$0.2$,$0.4$,$0.6$,$0.8$},
            ytick=\empty,
        ]
            \addplot[ycomb,mark=*,mark size=1.8pt,line width=1.3pt,color=twsblue]
                coordinates {(0.2,0.25) (0.4,0.20) (0.6,0.30) (0.8,0.25)};
            \node[font=\scriptsize,anchor=south,text=twsblue]
                at (axis cs:0.2,0.262) {$0.25$};
            \node[font=\scriptsize,anchor=south,text=twsblue]
                at (axis cs:0.4,0.212) {$0.20$};
            \node[font=\scriptsize,anchor=south,text=twsblue]
                at (axis cs:0.6,0.312) {$0.30$};
            \node[font=\scriptsize,anchor=south,text=twsblue]
                at (axis cs:0.8,0.262) {$0.25$};

        % ------------------------------------------------------------------
        % (2,1) Quantile function, uniform masses: jump of size Delta at 1/2.
        % ------------------------------------------------------------------
        \nextgroupplot[
            xlabel={$q$},
            ylabel={$F^{-1}(q)$},
            ymin=0,
            ymax=1,
            xtick={0,0.5,1},
            xticklabels={$0$,$\tfrac12$,$1$},
            ytick={0.2,0.4,0.6,0.8},
            yticklabels={$0.2$,$0.4=\,a_F^-$,$0.6=\,a_F^+$,$0.8$},
        ]
            \path[fill=twslight]
                (axis cs:0.45,0) rectangle (axis cs:0.55,1);
            \draw[densely dashed,color=twsred!70]
                (axis cs:0.45,0) -- (axis cs:0.45,1);
            \draw[densely dashed,color=twsred!70]
                (axis cs:0.55,0) -- (axis cs:0.55,1);
            \draw[dashed,line width=1.1pt,color=black!60]
                (axis cs:0.5,0) -- (axis cs:0.5,1);
            \node[rotate=90,font=\scriptsize,text=black!60]
                at (axis cs:0.47,0.84) {median};
            \addplot[line width=1.3pt,color=twsblue] coordinates {
                (0.00,0.2) (0.25,0.2)
                (0.25,0.40) (0.50,0.40)
                (0.50,0.60) (0.75,0.60)
                (0.75,0.8) (1.00,0.8)
            };
            \addplot[only marks,mark=*,mark size=1.6pt,color=twsblue]
                coordinates {(0.5,0.40)};
            \addplot[only marks,mark=o,mark size=1.8pt,color=twsblue,fill=white]
                coordinates {(0.5,0.60)};
            \draw[<->,>=Latex,semithick,color=twsred]
                (axis cs:0.25,0.33)
                -- node[below,font=\scriptsize,fill=white,inner sep=1pt] {$0.25$}
                (axis cs:0.50,0.33);
            \draw[<->,>=Latex,semithick,color=twsgold]
                (axis cs:0.50,0.67)
                -- node[above,font=\scriptsize,fill=white,inner sep=1pt] {$0.25$}
                (axis cs:0.75,0.67);
            \draw[<->,>=Latex,semithick,color=twsgold]
                (axis cs:0.585,0.40)
                -- node[right,font=\scriptsize,fill=white,inner sep=1pt] {$\Delta_F$}
                (axis cs:0.585,0.60);
            \draw[{Latex[length=1.2mm]}-{Latex[length=1.2mm]},semithick,color=twsred]
                (axis cs:0.45,0.08) -- (axis cs:0.50,0.08);
            \node[font=\scriptsize,text=twsred,anchor=east,inner sep=1pt]
                at (axis cs:0.435,0.08) {$\mu$};

        % ------------------------------------------------------------------
        % (2,2) Quantile function, perturbed masses: the jump moves to 0.45,
        % outside the window of radius 0.05, so the median gap is zero.
        % ------------------------------------------------------------------
        \nextgroupplot[
            xlabel={$q$},
            ylabel={$F^{-1}(q)$},
            ymin=0,
            ymax=1,
            xtick={0,0.5,1},
            xticklabels={$0$,$\tfrac12$,$1$},
            ytick={0.2,0.4,0.6,0.8},
            yticklabels={$0.2$,$0.4$,$0.6=\,a_F^\pm$,$0.8$},
        ]
            \path[fill=twslight]
                (axis cs:0.45,0) rectangle (axis cs:0.55,1);
            \draw[densely dashed,color=twsred!70]
                (axis cs:0.45,0) -- (axis cs:0.45,1);
            \draw[densely dashed,color=twsred!70]
                (axis cs:0.55,0) -- (axis cs:0.55,1);
            \draw[dashed,line width=1.1pt,color=black!60]
                (axis cs:0.5,0) -- (axis cs:0.5,1);
            \node[rotate=90,font=\scriptsize,text=black!60]
                at (axis cs:0.47,0.84) {median};
            \addplot[line width=1.3pt,color=twsblue] coordinates {
                (0.00,0.2) (0.25,0.2)
                (0.25,0.40) (0.45,0.40)
                (0.45,0.60) (0.75,0.60)
                (0.75,0.8) (1.00,0.8)
            };
            \addplot[only marks,mark=*,mark size=1.6pt,color=twsblue]
                coordinates {(0.45,0.40)};
            \addplot[only marks,mark=o,mark size=1.8pt,color=twsblue,fill=white]
                coordinates {(0.45,0.60)};
            \draw[{Latex[length=1.2mm]}-{Latex[length=1.2mm]},semithick,color=twsred]
                (axis cs:0.45,0.53) -- (axis cs:0.50,0.53);
            \node[font=\scriptsize,text=twsred,anchor=east,inner sep=1pt]
                at (axis cs:0.43,0.53) {$0.05$};
            \draw[<->,>=Latex,semithick,color=twsgold]
                (axis cs:0.50,0.67)
                -- node[above,font=\scriptsize,fill=white,inner sep=1pt] {$0.25$}
                (axis cs:0.75,0.67);
            \draw[{Latex[length=1.2mm]}-{Latex[length=1.2mm]},semithick,color=twsred]
                (axis cs:0.45,0.08) -- (axis cs:0.50,0.08);
            \node[font=\scriptsize,text=twsred,anchor=east,inner sep=1pt]
                at (axis cs:0.435,0.08) {$\mu$};

        \end{groupplot}
    \end{tikzpicture}
    \endgroup
    % \caption{The finite atomic distributions
    % with support $(A_1,A_2,A_3,A_4)=(0.2,0.4,0.6,0.8)$.
    % The top row shows probability masses and the bottom row shows quantiles.}
    \caption{Finite atomic distributions on $\{0.2,0.4,0.6,0.8\}$. The top panels show probability masses, and the bottom panels show quantile functions.}
    \label{fig:example-atomic}
\end{figure}
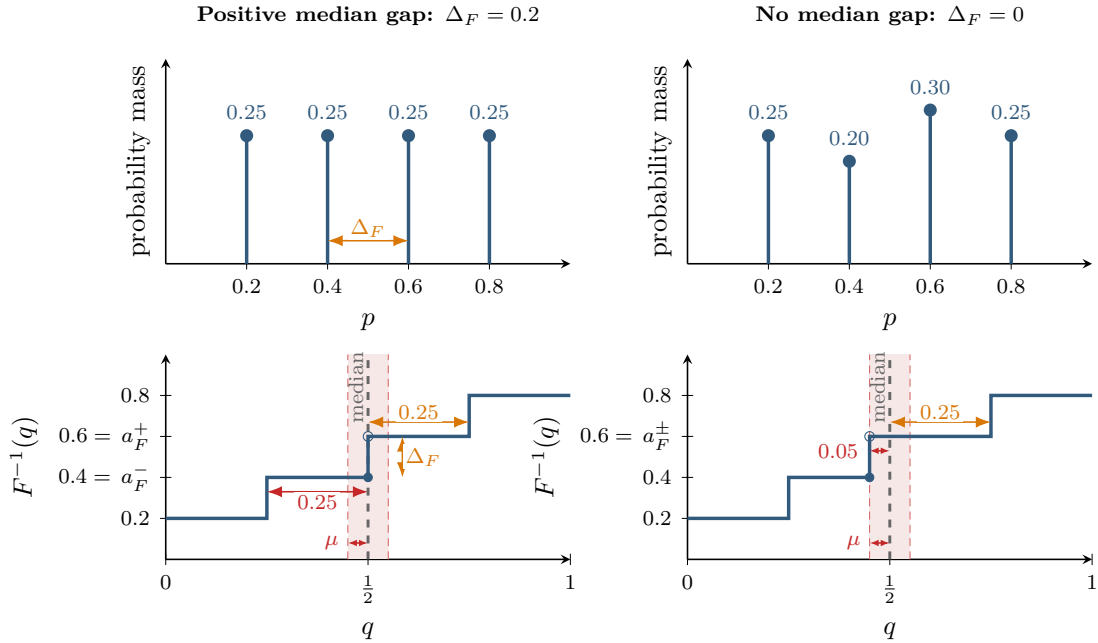

{For continuous prices, the key feature is how much probability lies close to the median edges. We use $\beta$ to describe a local power law: the probability mass within price distance $z$ below $a^-$ or above $a^+$ is asymptotically proportional to $z^{\beta+1}$. Equivalently, a probability mass $q$ next to an edge occupies a price interval of width asymptotically proportional to $q^{1/(\beta+1)}$. The case $\beta=0$ gives linear mass accumulation, while larger $\beta$ describes sparser prices near the edges.}

\begin{definition}[Regular continuous family]\label{def:continuous-family}
% For $\beta\ge0$, $\varepsilon_0\in(0,1/4]$, and $\Delta\in[0,1)$, let $\mathcal F^{\mathrm{ac}}_{\beta,\Delta,\varepsilon_0}$ contain the absolutely continuous laws on $[0,1]$ with density $0\le f\le M<\infty$, median gap $\Delta_F=\Delta$, and constants $c_-,c_+>0$ for which
{For $\beta\ge0$, $\varepsilon_0\in(0,1/4]$, and $\Delta\in[0,1)$, let $\mathcal F^{\mathrm{ac}}_{\beta,\Delta,\varepsilon_0}$ be the class of absolutely continuous distributions $F$ on $[0,1]$ with bounded density and median gap $\Delta_F=\Delta$, satisfying the following conditions. There exist constants $c_-,c_+>0$ and remainder functions $r_-,r_+$ such that, for every $0<q<\varepsilon_0$,}
% \begin{align*}
%  a^--F^{-1}(1/2-q)&=c_-q^{1/(\beta+1)}(1+r_-(q)), \qquad 0 < q < \varepsilon_0 \\
%  F^{-1}(1/2+q)-a^+&=c_+q^{1/(\beta+1)}(1+r_+(q)),
%  \qquad 0<q<\varepsilon_0.
% \end{align*}
{
\begin{align*}
 a_F^--F^{-1}(1/2-q)&=c_-q^{1/(\beta+1)}(1+r_-(q)), \\
 F^{-1}(1/2+q)-a_F^+&=c_+q^{1/(\beta+1)}(1+r_+(q)).
\end{align*}
}
The remainder functions satisfy $|r_\pm(q)|\le1/2$ and $r_\pm(q)\to0$ as $q\downarrow0$.
% When $\beta=0$, we additionally assume that $|r_\pm(q)|\le Lq^\eta$ for some finite $L\ge0$ and $\eta>0$. %All constants and error functions are fixed with $F$, independently of $T$.
{When $\beta=0$, they additionally satisfy $|r_\pm(q)|\le Lq^\eta$ for some $L\ge0$ and $\eta>0$. The density bound, constants, and remainder functions may depend on $F$.}
\end{definition}
The cases $\Delta=0$ and $\Delta>0$ give the two continuous families.
% The expansions describe two-sided mass accumulation of order $z^{\beta+1}$ at each occupied median edge. The leading constants may differ on the two sides, and behavior away from the median is unrestricted apart from the density bound. The vanishing relative remainders identify a regular leading term near each edge.
{The parameter $\varepsilon_0$ specifies a probability window around the median, and $c_-$ and $c_+$ allow different scales on the two sides. The remainder functions allow departures from an exact power law while requiring the relative error to vanish near each edge. These regularity conditions control the inventory drift in the \STARS{} analysis; the additional condition at $\beta=0$ controls the error terms in the logarithmic regret estimates. Behavior away from these median neighborhoods is unrestricted apart from the density bound.}

%The proof uses this information beyond the two-sided power envelope to determine the leading inventory drift. At $\beta=0$, the additional remainder condition controls the accumulated error in the logarithmic estimate. Atomic laws satisfy the separate local quantile conditions in Definition~\ref{def:atomic-family}.

\begin{example}[Continuous prices without a median gap]\label{ex:no-gap}
With $\Delta=0$, let $F_{\beta,0}$ have density $f_{\beta,0}(p)=(\beta+1)2^\beta|p-1/2|^\beta$ for $0\le p\le1$. These give uniform, V-shaped, and quadratic examples at $\beta=0,1,2$ respectively. Their quantiles satisfy $F_{\beta,0}^{-1}(1/2\pm q)=1/2\pm(2q)^{1/(\beta+1)}/2$ for $0<q<1/2$.
Thus they satisfy Definition~\ref{def:continuous-family} with zero remainder. Figure~\ref{fig:example-no-gap} displays the densities and their quantile functions. The shaded band isolates the local median window.
\end{example}
% Example figure: continuous distributions without a median gap,
% f_beta(p) = (beta+1) 2^beta |p-1/2|^beta, drawn for beta = 0, 1, 2.
% Left: densities in price space.  Right: quantile functions.
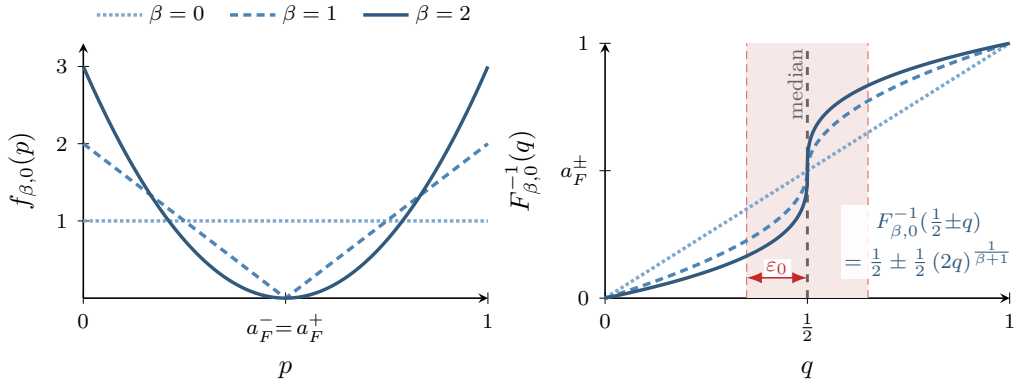
\begin{figure}[!htb]
    \centering
    \begingroup
    \definecolor{twsred}{HTML}{C62828}
    \definecolor{twsgold}{HTML}{D97706}
    \definecolor{twsblue}{HTML}{315A7D}
    \definecolor{twsbluemid}{HTML}{4E85B5}
    \definecolor{twsbluelight}{HTML}{7FA8CC}
    \definecolor{twslight}{HTML}{F7E8E8}
    \begin{tikzpicture}
        \begin{groupplot}[
            group style={
                group size=2 by 1,
                horizontal sep=1.55cm,
            },
            width=0.42\textwidth,
            height=0.30\textwidth,
            axis lines=left,
            axis line style={black,semithick},
            tick style={black},
            tick label style={font=\scriptsize},
            label style={font=\small},
            xmin=0,
            xmax=1,
            enlargelimits=false,
            clip=false,
        ]

        % ------------------------------------------------------------------
        % Left: densities f_beta for beta = 0, 1, 2.
        % ------------------------------------------------------------------
        \nextgroupplot[
            xlabel={$p$},
            ylabel={$f_{\beta,0}(p)$},
            ymin=0,
            ymax=3.30,
            xtick={0,0.5,1},
            xticklabels={$0$,$a_F^-{=}\,a_F^+$,$1$},
            ytick={1,2,3},
            legend style={
                draw=none,
                fill=none,
                font=\scriptsize,
                at={(0.5,1.03)},
                anchor=south,
                legend columns=-1,
                /tikz/every even column/.append style={column sep=7pt},
            },
        ]
            \addplot[line width=1.3pt,densely dotted,color=twsbluelight,domain=0:1,samples=2]
                {1};
            \addlegendentry{$\beta=0$}
            \addplot[line width=1.3pt,densely dashed,color=twsbluemid,domain=0:1,samples=201]
                {4*abs(x-0.5)};
            \addlegendentry{$\beta=1$}
            \addplot[line width=1.3pt,color=twsblue,domain=0:1,samples=201]
                {12*(x-0.5)^2};
            \addlegendentry{$\beta=2$}

        % ------------------------------------------------------------------
        % Right: quantile functions, drawn parametrically (exact inverses).
        % ------------------------------------------------------------------
        \nextgroupplot[
            xlabel={$q$},
            ylabel={$F_{\beta,0}^{-1}(q)$},
            ymin=0,
            ymax=1,
            xtick={0,0.5,1},
            xticklabels={$0$,$\tfrac12$,$1$},
            ytick={0,0.5,1},
            yticklabels={$0$,$a_F^\pm$,$1$},
        ]
            \path[fill=twslight]
                (axis cs:0.35,0) rectangle (axis cs:0.65,1);
            \draw[densely dashed,color=twsred!70]
                (axis cs:0.35,0) -- (axis cs:0.35,1);
            \draw[densely dashed,color=twsred!70]
                (axis cs:0.65,0) -- (axis cs:0.65,1);
            \draw[dashed,line width=1.1pt,color=black!60]
                (axis cs:0.5,0) -- (axis cs:0.5,1);
            \node[rotate=90,font=\scriptsize,text=black!60]
                at (axis cs:0.47,0.84) {median};
            \addplot[line width=1.3pt,densely dotted,color=twsbluelight,domain=0:1,samples=2]
                ({x},{x});
            \addplot[line width=1.3pt,densely dashed,color=twsbluemid,domain=0:0.5,samples=101]
                ({0.5-2*(0.5-x)^2},{x});
            \addplot[line width=1.3pt,densely dashed,color=twsbluemid,domain=0.5:1,samples=101]
                ({0.5+2*(x-0.5)^2},{x});
            \addplot[line width=1.3pt,color=twsblue,domain=0:0.5,samples=101]
                ({0.5-4*(0.5-x)^3},{x});
            \addplot[line width=1.3pt,color=twsblue,domain=0.5:1,samples=101]
                ({0.5+4*(x-0.5)^3},{x});
            \draw[<->,>=Latex,semithick,color=twsred]
                (axis cs:0.35,0.08)
                -- node[above,font=\scriptsize,fill=white,inner sep=1pt] {$\varepsilon_0$}
                (axis cs:0.50,0.08);
            \node[
                font=\scriptsize,
                text=twsblue,
                fill=white,
                fill opacity=0.88,
                text opacity=1,
                inner sep=1pt,
                align=center,
            ] at (axis cs:0.80,0.22)
                {$F_{\beta,0}^{-1}(\tfrac12{\pm}q)$\\
                 $=\tfrac12\pm\tfrac12\,(2q)^{\frac{1}{\beta+1}}$};

        \end{groupplot}
    \end{tikzpicture}
    \endgroup
    \caption{The continuous no-gap family of Example~\ref{ex:no-gap}, drawn
    for $\beta\in\{0,1,2\}$.
    The shaded band is the quantile window of radius $\varepsilon_0=0.15$
    used in Definition~\ref{def:continuous-family}.}
    \label{fig:example-no-gap}
\end{figure}

\begin{example}[Continuous prices with a positive median gap]\label{ex:positive-gap}
For $\Delta\in(0,1)$, put $a^\pm=(1\pm\Delta)/2$ and let $F_{\beta,\Delta}$ have density
\begin{align}
 f_{\beta,\Delta}(p)=\frac{(\beta+1)2^\beta}{(1-\Delta)^{\beta+1}}
 \begin{cases}(a^--p)^\beta,&0\le p<a^-,\\
     0,&a^-<p<a^+, \\
 (p-a^+)^\beta,&a^+<p\le1.\\
 \end{cases}
 \label{eq:canonical-gap-density}
\end{align}
These examples belong to Definition~\ref{def:continuous-family}, with $r_\pm=0$, $c_-=c_+=(1-\Delta)2^{-\beta/(\beta+1)}$, and $M=(\beta+1)/(1-\Delta)$. For $\beta>0$ the density vanishes at the gap edges. Equation~\eqref{eq:canonical-gap-density} also supplies the explicit instances for our lower bound discussed in Theorem \ref{lb:thm-positive-gap}. The quantiles satisfy $F_{\beta,\Delta}^{-1}(1/2\pm q)=a^\pm\pm(1-\Delta)(2q)^{1/(\beta+1)}/2$ for $0<q<1/2$.
Figure~\ref{fig:example-positive-gap} shows the resulting jump at the median for $\Delta=0.2$.
\end{example}
% Example figure: continuous distributions with a positive median gap,
% drawn with Delta = 0.2 (a_F^- = 0.4, a_F^+ = 0.6) for beta = 0, 1, 2.
% Left: densities in price space.  Right: quantile functions.
\begin{figure}[!htb]
    \centering
    \begingroup
    \definecolor{twsred}{HTML}{C62828}
    \definecolor{twsgold}{HTML}{D97706}
    \definecolor{twsblue}{HTML}{315A7D}
    \definecolor{twsbluemid}{HTML}{4E85B5}
    \definecolor{twsbluelight}{HTML}{7FA8CC}
    \definecolor{twslight}{HTML}{F7E8E8}
    \begin{tikzpicture}
        \begin{groupplot}[
            group style={
                group size=2 by 1,
                horizontal sep=1.55cm,
            },
            width=0.42\textwidth,
            height=0.30\textwidth,
            axis lines=left,
            axis line style={black,semithick},
            tick style={black},
            tick label style={font=\scriptsize},
            label style={font=\small},
            xmin=0,
            xmax=1,
            enlargelimits=false,
            clip=false,
        ]

        % ------------------------------------------------------------------
        % Left: densities f_{beta,Delta} for beta = 0, 1, 2.
        % ------------------------------------------------------------------
        \nextgroupplot[
            xlabel={$p$},
            ylabel={$f_{\beta,\Delta}(p)$},
            ymin=0,
            ymax=3.95,
            xtick={0,0.4,0.6,1},
            xticklabels={$0$,$a_F^-$,$a_F^+$,$1$},
            ytick=\empty,
            legend style={
                draw=none,
                fill=none,
                font=\scriptsize,
                at={(0.5,1.03)},
                anchor=south,
                legend columns=-1,
                /tikz/every even column/.append style={column sep=7pt},
            },
        ]
            \draw[densely dashed,semithick,color=twsgold]
                (axis cs:0.4,0) -- (axis cs:0.4,3.95);
            \draw[densely dashed,semithick,color=twsgold]
                (axis cs:0.6,0) -- (axis cs:0.6,3.95);
            \addplot[line width=1.3pt,densely dotted,color=twsbluelight,domain=0:0.4,samples=2]
                {1.25};
            \addlegendentry{$\beta=0$}
            \addplot[line width=1.3pt,densely dashed,color=twsbluemid,domain=0:0.4,samples=101]
                {6.25*(0.4-x)};
            \addlegendentry{$\beta=1$}
            \addplot[line width=1.3pt,color=twsblue,domain=0:0.4,samples=101]
                {23.4375*(0.4-x)^2};
            \addlegendentry{$\beta=2$}
            \addplot[line width=1.3pt,densely dotted,color=twsbluelight,domain=0.6:1,samples=2,forget plot]
                {1.25};
            \addplot[line width=1.3pt,densely dashed,color=twsbluemid,domain=0.6:1,samples=101,forget plot]
                {6.25*(x-0.6)};
            \addplot[line width=1.3pt,color=twsblue,domain=0.6:1,samples=101,forget plot]
                {23.4375*(x-0.6)^2};
            \draw[<->,>=Latex,semithick,color=twsgold]
                (axis cs:0.4,0.35)
                -- node[above,font=\scriptsize,fill=white,inner sep=1pt]
                    {$\Delta_F$}
                (axis cs:0.6,0.35);

        % ------------------------------------------------------------------
        % Right: quantile functions, drawn parametrically (exact inverses).
        % ------------------------------------------------------------------
        \nextgroupplot[
            xlabel={$q$},
            ylabel={$F_{\beta,\Delta}^{-1}(q)$},
            ymin=0,
            ymax=1,
            xtick={0,0.5,1},
            xticklabels={$0$,$\tfrac12$,$1$},
            ytick={0.4,0.6},
            yticklabels={$a_F^-$,$a_F^+$},
        ]
            \path[fill=twslight]
                (axis cs:0.35,0) rectangle (axis cs:0.65,1);
            \draw[densely dashed,color=twsred!70]
                (axis cs:0.35,0) -- (axis cs:0.35,1);
            \draw[densely dashed,color=twsred!70]
                (axis cs:0.65,0) -- (axis cs:0.65,1);
            \draw[dashed,line width=1.1pt,color=black!60]
                (axis cs:0.5,0) -- (axis cs:0.5,1);
            \node[rotate=90,font=\scriptsize,text=black!60]
                at (axis cs:0.47,0.84) {median};
            \addplot[line width=1.3pt,densely dotted,color=twsbluelight,domain=0:0.4,samples=2]
                ({0.5-1.25*(0.4-x)},{x});
            \addplot[line width=1.3pt,densely dotted,color=twsbluelight,domain=0.6:1,samples=2]
                ({0.5+1.25*(x-0.6)},{x});
            \addplot[line width=1.3pt,densely dashed,color=twsbluemid,domain=0:0.4,samples=101]
                ({0.5-3.125*(0.4-x)^2},{x});
            \addplot[line width=1.3pt,densely dashed,color=twsbluemid,domain=0.6:1,samples=101]
                ({0.5+3.125*(x-0.6)^2},{x});
            \addplot[line width=1.3pt,color=twsblue,domain=0:0.4,samples=101]
                ({0.5-7.8125*(0.4-x)^3},{x});
            \addplot[line width=1.3pt,color=twsblue,domain=0.6:1,samples=101]
                ({0.5+7.8125*(x-0.6)^3},{x});
            \addplot[only marks,mark=*,mark size=1.6pt,color=twsblue]
                coordinates {(0.5,0.40)};
            \addplot[only marks,mark=o,mark size=1.8pt,color=twsblue,fill=white]
                coordinates {(0.5,0.60)};
            \draw[<->,>=Latex,semithick,color=twsgold]
                (axis cs:0.58,0.40)
                -- node[right,font=\scriptsize,fill=white,inner sep=1pt] {$\Delta_F$}
                (axis cs:0.58,0.60);
            \draw[<->,>=Latex,semithick,color=twsred]
                (axis cs:0.35,0.08)
                -- node[above,font=\scriptsize,fill=white,inner sep=1pt] {$\varepsilon_0$}
                (axis cs:0.50,0.08);
            \node[
                font=\scriptsize,
                text=twsblue,
                fill=white,
                fill opacity=0.88,
                text opacity=1,
                inner sep=1pt,
                align=center,
            ] at (axis cs:0.79,0.21)
                {$F_{\beta,\Delta}^{-1}(\tfrac12{\pm}q)$\\
                 $=a_F^\pm\pm\tfrac{1-\Delta}{2}\,(2q)^{\frac{1}{\beta+1}}$};

        \end{groupplot}
    \end{tikzpicture}
    \endgroup
    \caption{The positive-gap family of Example~\ref{ex:positive-gap},
    drawn with $\Delta_F=0.2$ for $\beta\in\{0,1,2\}$.}
    \label{fig:example-positive-gap}
\end{figure}
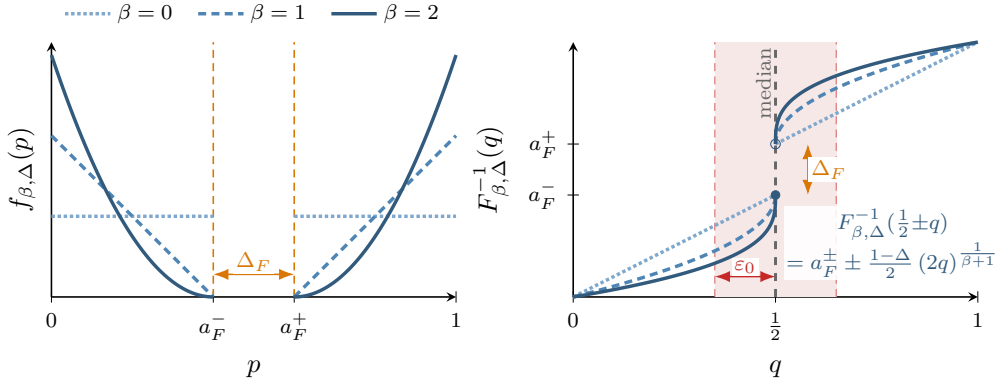

\section{Novel Spectrum of Fundamental Limits}
\label{lb:sec-fundamental-limits}

In this section, we present a novel spectrum of fundamental limits that
arises when continuous price distributions have a positive median gap and the initial inventory is zero. To isolate the effect of the initial inventory, we first
consider a setting with linear slack from both storage boundaries. A
feasible implementation of a multisecretary policy then recovers the
regret guarantees of \cite{besbes2025}. We subsequently change the initial and terminal inventories to
% BEGIN retained section overview (before full-class strengthening).
% zero. For the canonical positive-gap distributions, this change forces
% every online policy to incur a strictly larger polynomial regret exponent.
% END retained section overview.
zero. For every positive-gap distribution in the regular continuous family
of Definition~\ref{def:continuous-family}, this change forces every online
policy to incur a strictly larger polynomial regret exponent.

\subsection{Starting in the interior: Recovering previously known regret scalings}
\label{subsec:interior-transfer}

Consider first the unconstrained \emph{static median-sign rule}: buy one unit
when the price is below a fixed median, sell one unit when it is above,
and hold at equality. When there is a positive median gap, choose the
median strictly inside it. For continuous prices, its cumulative trades form
a simple symmetric random walk (SSRW). The inventory fluctuations are of order
$\sqrt T$, so linear slack from both boundaries protects this SSRW from hitting any of the boundaries
with high probability. To obtain a feasible trading policy, we must also
balance purchases and sales by the end. Interpreting a selected price as a sale
and a rejected price as a purchase, this leads to a
multisecretary type policy that adjusts its threshold to the remaining quota
and the gaps in the price distribution. 
% \ak{the median-sign rule has not been described yet, perhaps we should add a small sentence saying what really is the median-sign rule?}

We can hence use the \textsf{CwG} principle of \cite{besbes2025} with a selection
quota of $k = n/2$, where $n=2\lfloor T/2\rfloor$, and hold in the final
time step when $T$ is odd. %For finite atomic prices, use their \textsf{CwG} policy with all atom-boundary quantiles. For continuous prices, focus the rule on the median: use the \textsf{CE} quantile when $\Delta=0$; when $\Delta>0$, designate only the median gap for \textsf{CwG}'s threshold adjustment. 
The details are deferred to Appendix~\ref{app:interior-transfer}.
%The policy exhausts its quota by accepting every remaining price once the remaining quota equals the remaining horizon, and rejects once the quota is exhausted. If a proposed trade would violate a storage boundary, instead return toward $y_0$ and hold after reaching it. The remaining quota ensures that this return is feasible.
Denote the resulting feasible policy by $\pi^{\textsf{CwG}}$.

\begin{proposition}\label{prop:interior-transfer}
Fix $\eta>0$, let $T\ge2$, and assume $y_0\wedge(B-y_0)\ge\eta T$.
\begin{enumerate}[label=$(\alph*)$]
    \item Fix $\mu\in(0,1/2]$ and $\Delta\in[0,1]$. For every distribution $F\in\mathcal F^{\rm at}_{\mu,\Delta}$, there exists a constant $C=C(F,\eta)<\infty$ such that $\Reg(\pi^{\textup{\textsf{CwG}}};B,T,y_0)\le C$.
    \item Fix $\beta\ge0$, $\varepsilon_0\in(0,1/4]$, and let $\Delta=0$. For every distribution $F\in\mathcal F^{\rm ac}_{\beta,\Delta,\varepsilon_0}$, there exists a constant $C=C(F,\eta)<\infty$ such that
    \begin{align*}
        \Reg(\pi^{\textup{\textsf{CwG}}};B,T,y_0)
        \le C\log T\,\mathbbm{1}\{\beta=0\}
        +CT^{\beta/(2\beta+2)}\,\mathbbm{1}\{\beta>0\}.
    \end{align*}
    \item Fix $\beta\ge0$, $\varepsilon_0\in(0,1/4]$, and $\Delta\in(0,1)$. For every distribution $F\in\mathcal F^{\rm ac}_{\beta,\Delta,\varepsilon_0}$, there exists a constant $C=C(F,\eta)<\infty$ such that
    \begin{align*}
        \Reg(\pi^{\textup{\textsf{CwG}}};B,T,y_0)
        \le C\log^2 T\,\mathbbm{1}\{\beta=0\}
        +CT^{\beta/(2\beta+2)}(\log T)^{(\beta+2)/(2\beta+2)}
          \,\mathbbm{1}\{\beta>0\}.
    \end{align*}
\end{enumerate}
%The constants are independent of $B$, $T$, and $y_0$ subject to the stated slack condition.
\end{proposition}

The key observation is that, before enforcing the physical boundaries, the
inventory displacement has a restoring drift with at most a small
quantile perturbation. Its maximum is less than $\eta T$ except with
probability exponentially small in $T$. On this event the storage
policy implements the multisecretary decisions exactly. Moreover, its
trading profit is twice the selection reward minus the sum of all
prices and the same identity holds for the static hindsight benchmark.
Consequently, the storage regret is at
most twice the $n$-step multisecretary regret of the specified selection rule, plus
$O_{F,\eta}(Te^{-c_\eta T})+1$. Appendix~\ref{app:interior-transfer}
proves the concentration, the feasible fallback, and this comparison.

Proposition~\ref{prop:interior-transfer} establishes the achievable
rates in Figure~\ref{fig:fundamental-drivers} for this transferred
policy. In particular, a positive median gap preserves the polynomial
exponent $\beta/(2\beta+2)$ of the multisecretary guarantee. 
The proof separates the cost of choosing the wrong side of the median
from the cost of the price's distance to its median edge. The first is
controlled by the finite-type \textsf{CwG} guarantee while the second is
controlled by the local quantile bound and the inventory fluctuations.

\subsection{Starting from zero: a new lower-bound spectrum}
\label{subsec:boundary-limits}

Starting from zero, an online policy must build inventory before it can
exploit selling opportunities. Following the static median-sign rule would
prescribe sales even when the inventory level is zero.
% Therefore, preserving inventory requires additional corrections to this rule.
{A feasible policy must therefore modify this rule and may also
forgo sales to preserve inventory for later.}
% A positive median gap makes unnecessary reserve costly, while thin mass near the gap edges makes cheap opportunities to build that reserve scarce.
A positive median gap makes unnecessary reserve costly.
{Forgoing sales near the upper gap edge is less costly than
forgoing sales at higher prices, but such opportunities are scarce when
little probability mass lies near that edge.} We
show that this tradeoff imposes a polynomial regret lower bound on every
online policy.

% BEGIN retained canonical-only setup (before full-class strengthening).
% We establish the lower bound using an explicit hard distribution. Fix $\beta\geq0$ and $\Delta\in(0,1)$, and put
% $\kappa=\beta+1$, $a^\pm=(1\pm\Delta)/2$, and
% $D_{\beta,\Delta}=2^{\kappa-1}/(1-\Delta)^\kappa$.
% Let $F_{\beta,\Delta}$ be the CDF associated with the canonical density
% $f_{\beta,\Delta}$ in \eqref{eq:canonical-gap-density}.  Thus this is an
% instance of Definition~\ref{def:continuous-family} with zero remainder.
% Each occupied component has probability $1/2$, and the density is globally
% bounded by $\kappa/(1-\Delta)$. We have that
% $F_{\beta,\Delta}^{-1}(1/2\pm q)=a^\pm\pm(q/D_{\beta,\Delta})^{1/\kappa}$
% for $0<q<1/2$.
% Equivalently, the probability mass within price distance $z$ of either
% inner gap edge is exactly $D_{\beta,\Delta}z^\kappa$ for
% $0\leq z\leq(1-\Delta)/2$.
% END retained canonical-only setup.
We establish the lower bound for every fixed distribution in the regular
continuous family of Definition~\ref{def:continuous-family} with a positive
median gap. The local quantile bounds control both the scarcity of cheap
online corrections and the cost of the corrections available in hindsight.
% The two gap edges may have different leading constants, and the argument
% allows the remainder terms in that definition.

% BEGIN retained canonical Theorem 3.2 (before full-class strengthening).
% \begin{theorem}
% \label{lb:thm-positive-gap}
% Fix $\beta\ge0$, $\Delta\in(0,1)$, and $\rho\in(0,1)$.
% Assume prices are i.i.d.\ with density $f_{\beta,\Delta}$ in
% \eqref{eq:canonical-gap-density}, let $B=\lceil\rho T\rceil$, and
% set $y_0=y_T=0$. There exist constants
% $c=c(\beta,\Delta,\rho)>0$ and $T_0=T_0(\beta,\Delta,\rho)<\infty$
% such that every possibly randomized, nonanticipative feasible policy
% $\pi$ satisfies
% $\Reg(\pi;B,T,0)\ge cT^{(\beta+1)/(2\beta+4)}$
% for every $T\ge T_0$.
% \end{theorem}
% END retained canonical Theorem 3.2.
\begingroup
\begin{theorem}
\label{lb:thm-positive-gap}
Fix $\beta\ge0$, $\varepsilon_0\in(0,1/4]$, $\Delta\in(0,1)$, and
$\rho\in(0,1)$. Suppose prices are i.i.d.\ with a fixed distribution
$F\in\mathcal F^{\mathrm{ac}}_{\beta,\Delta,\varepsilon_0}$, as defined in
Definition~\ref{def:continuous-family}. Let $B=\lceil\rho T\rceil$ and
set $y_0=y_T=0$. There exist constants $c=c(F,\rho)>0$ and
$T_0=T_0(F,\rho)<\infty$ such that, for every $T\ge T_0$, every
possibly randomized, nonanticipative feasible policy $\pi$ satisfies
$\Reg(\pi;B,T,0)\ge cT^{(\beta+1)/(2\beta+4)}$.
\end{theorem}
\endgroup

% BEGIN retained canonical-only comparison (before full-class strengthening).
% Compare Theorem~\ref{lb:thm-positive-gap} with
% Proposition~\ref{prop:interior-transfer} for the same fixed
% $F_{\beta,\Delta}$ with $\Delta>0$ and the same capacity. The zero initial inventory
% lower-bound exponent exceeds the upper-bound exponent when starting in the interior by
% $(\beta+1)/(2\beta+4)-\beta/(2\beta+2)
% =1/[2(\beta+1)(\beta+2)]>0$. Thus the boundary lower bound eventually
% exceeds the achievable interior regret, including its logarithmic
% factors, for every fixed $\beta\ge0$. At $\beta=0$, the distribution is
% uniform on two separated intervals: moving from linear initial slack
% to zero initial inventory changes an achievable $O(\log^2 T)$ guarantee into
% an unavoidable $\Omega(T^{1/4})$ loss. The price distribution and its median gap
% are unchanged but with the initial inventory being zero, the decision maker must now acquire the inventory needed for
% future sales.
% END retained canonical-only comparison.
Compare Theorem~\ref{lb:thm-positive-gap} with
Proposition~\ref{prop:interior-transfer} for the same fixed
$F\in\mathcal F^{\mathrm{ac}}_{\beta,\Delta,\varepsilon_0}$ with $\Delta>0$
and the same capacity. The zero initial inventory lower-bound exponent
exceeds the upper-bound exponent when starting in the interior by
$(\beta+1)/(2\beta+4)-\beta/(2\beta+2)
=1/[2(\beta+1)(\beta+2)]>0$. Thus the boundary lower bound eventually
exceeds the achievable interior regret, including its logarithmic
factors, for every fixed $\beta\ge0$. At $\beta=0$, starting with linear
slack permits an $O(\log^2 T)$ guarantee, whereas starting empty forces
an $\Omega(T^{1/4})$ loss. The canonical distribution
$F_{0,\Delta}$, uniform on two separated intervals, is a concrete example.
The price distribution and its median gap are unchanged in this comparison,
but with the initial inventory being zero, the decision maker must now
acquire the inventory needed for future sales.

The contrast also explains how degeneracy interacts with the initial
inventory. In the balanced fluid relaxation, the intermediate storage
constraints are removed and the mean trade must be zero.
{Interpret $\lambda$ as an accounting value assigned to one
unit of inventory: buying at price $p$ gives net value $\lambda-p$,
while selling gives $p-\lambda$. Because expected net purchases are zero,
these inventory values cancel from expected profit.} For every
$\lambda\in[0,1]$ and every action rule $x(p)\in[-1,1]$ with
$\E[x(p)]=0$, we have
$\E[-px(p)]\stackrel{(a)}{=}\E[-(p-\lambda)x(p)]
\stackrel{(b)}{\le}\E|p-\lambda|$,
where $(a)$ uses the balance constraint and $(b)$ uses $|x(p)|\le1$.
The static median-sign rule attains equality for every
$\lambda\in[a^-,a^+]$. Outside this interval, the occupied prices
between $\lambda$ and the nearest median edge make the inequality
strict for that rule. Hence the optimal dual prices are precisely
$\mathop{\rm argmin}_{\lambda\in[0,1]}\E|p-\lambda|=[a^-,a^+]$.
This multiplicity of optimal shadow prices is the degeneracy created
by the median gap.

In the multisecretary problem, the decision maker has the resource
upfront and decides which arrivals should consume it. The
\textsf{CwG} policy of \cite{besbes2025} handles gaps while retaining
the same polynomial regret exponent as in the corresponding no-gap
setting, up to logarithmic factors. In trading, starting with zero inventory requires the decision maker
to acquire inventory before a sale
can be made. The same gap now makes an unnecessary inventory reserve
costly, and the continuous mass near its edges determines the
availability of inexpensive opportunities to build that reserve.
Proposition~\ref{prop:interior-transfer} and
Theorem~\ref{lb:thm-positive-gap} isolate this interaction: for the
same price distribution, linear initial slack permits the
multisecretary rates, while zero initial inventory forces the larger exponent.

% BEGIN retained uniform-example introduction (before full-class strengthening).
% To understand the lower-bound argument, first take $\beta=0$, so
% prices are uniform on the two intervals $[0,a^-]$ and $[a^+,1]$.
% END retained uniform-example introduction.
To understand the lower-bound argument, first take the canonical example
$F_{0,\Delta}$, with prices uniform on the two intervals $[0,a^-]$ and
$[a^+,1]$, where $a^\pm=(1\pm\Delta)/2$.
Call a price \emph{low} if $p_t\le a^-$ and \emph{high} if
$p_t\ge a^+$. The static median-sign rule buys at every low price and sells
at every high price. Let $S_t$ be the number of low prices minus the
number of high prices in the first $t$ steps, and let
$M_t=\max_{0\le s\le t}(-S_s)^+$. The expected deepest deficit
$\E M_T$ is of order $\sqrt T$. If the high prices outnumber the
low prices among the first $t$ arrivals, a trader starting at zero cannot make all the
sales prescribed by the static median-sign rule. Every feasible policy,
including hindsight, must therefore change at least $M_T$ units
of its prescribed high-price trades. Holding instead of selling
counts as one unit; buying instead of selling counts as two.

The necessity of these changes is shared by online and hindsight
policies. Their costs depend on \emph{which} high-price trades are
changed. Hindsight can choose prices close to $a^+$, subject to
feasibility at every time step: by the first time $S_t=-y$, at least
$y$ units of changes must already have occurred. It also chooses
which low-price purchases to forgo to restore the terminal
inventory to zero. Appendix~\ref{lb:app-positive-gap} constructs
such a feasible hindsight path whose expected correction cost is
at most $\Delta\E M_T+O(\log T)$. The first term is the minimum
gap charge; the second accounts for the distances of the chosen
prices from the gap edges and the rare adjustments needed to respect
capacity.

An online policy must choose these trades without knowing future
prices. To quantify its difficulty, apply the static median-sign rule
from zero with sales blocked whenever the inventory level is zero, ignoring the
upper boundary and terminal requirement. Its inventory is
$\overline Y_t=S_t+M_t$. This process depends only on the observed
low/high price sequence; the online
policy's inventory can differ from it. Fix $2\le L\le\sqrt T$
and $0<z\le(1-\Delta)/2$. Consider the band
$0\le\overline Y_{t-1}\le L$ and call a high price cheap to forgo
when $p_t\in[a^+,a^++z]$. For the uniform example, such a price
has probability $z/(1-\Delta)$. The static median-sign rule spends at
most order $L\sqrt T$ steps in the band in expectation, so there
are at most order $zL\sqrt T$ cheap opportunities there. This
limits how often the online policy can preserve inventory at
little additional edge cost while $\overline Y_{t-1}$ is close to zero.

The policy can also forgo a cheap high-price sale outside the band.
After such an arrival, however, $\overline Y$ has
probability at least order $L/\sqrt T$ of never returning to zero
before the horizon. On that event, the running minimum of $S$
has already reached its final value. The last visit of $\overline Y$ to zero occurred before the current arrival, and the changes
needed to cover the deficit $M_T$ had to be made by that visit.
The current high-price change therefore increases the total above
this minimum. Each excess unit
adds the fixed gap $\Delta$ to the correction cost. Thus preserving
inventory can turn out to have been unnecessary even though all
inventory is ultimately returned to zero. The online policy has
to balance this risk against forgoing higher-priced sales, which
carry a larger edge cost.

Choosing $z$ proportional to $\Delta L/\sqrt T$ balances these two
costs. The scarcity estimate then gives, for the uniform example,
$\Reg(\pi;B,T,0)\ge c_1L-c_2L^3/\sqrt T-C\log T$.
Taking $L$ to be a sufficiently small constant times $T^{1/4}$
makes the first two terms a positive multiple of $T^{1/4}$,
% BEGIN retained canonical-only proof-sketch extension (before full-class strengthening).
% while the hindsight edge cost is smaller. For general $\beta$,
% the probability of a cheap high price is
% $D_{\beta,\Delta}z^{\beta+1}$. The same argument gives
% END retained canonical-only proof-sketch extension.
while the hindsight edge cost is smaller. For any fixed distribution
$F$ in the stated class, the probability of a cheap high price is at most
$C_Fz^{\beta+1}$ for sufficiently small $z>0$. The local quantile bounds
also keep the hindsight edge cost at a smaller order than the claimed
lower bound. The same argument therefore gives
$\Reg(\pi;B,T,0)\ge
c_1L-c_2L^{\beta+3}/T^{(\beta+1)/2}-o(L)$
at the scale $L\asymp T^{(\beta+1)/(2\beta+4)}$, yielding
Theorem~\ref{lb:thm-positive-gap}. The appendix proves the
deadline construction, band-occupation estimate, and last-return
bound, and includes the physical capacity and terminal constraints.

This correction-cost calculation also connects the lower bound to
the marginal-compensation analysis used for the upper bounds.
Write $C_T(\pi)$ for the total gap and edge costs of the policy's
departures from the static median-sign rule, as defined in
Lemma~\ref{lb:lem-accounting}. That lemma gives the pathwise identity
$\sum_{t=1}^T\partial R_t
=R_T(\pi^{\mathsf{hs}};0)-R_T(\pi;0)
=C_T(\pi)-C_T(\pi^{\mathsf{hs}})$.
Subtracting the common baseline $\Delta M_T$ from both correction
costs leaves their difference unchanged. The argument above bounds
the additional gap and edge costs forced
by online decisions, and hence the same regret measured by
marginal compensation.

We illustrate the fundamental limit numerically in
Figure~\ref{fig:dp-boundary-fit}. We compute the optimal online
policy using \textsf{Dynamic Programming (DP)} for the three canonical
distributions with $\beta\in\{0,1,2\}$ and median gap $\Delta=0.2$.
We vary $T$ from $10$ to $10{,}000$, set $B=T/2$, and start and end with
zero inventory. Each point is the average hindsight profit minus the
realized \textsf{DP} profit over $200$ independent price sequences. %Thus the figure isolates the cost of not knowing future prices, even when the policy knows the distribution and optimizes its online decisions.
We fit curves of the form $aT^{(\beta+1)/(2\beta+4)}-b$, shown as dashed
lines. The coefficient $a$ and offset $b$
are estimated using weighted least squares,
with weights equal to the inverse squared Monte Carlo standard errors.
For all three distributions, the fitted curves closely track the observed
regret. This is consistent with the polynomial
spectrum in Theorem~\ref{lb:thm-positive-gap} and with the upper bounds
in Theorem~\ref{thm:stars-upper}, which also bound the regret of optimal
\textsf{DP}.

\begin{figure}[htbp]
    \centering
    % Native PGFPlots input. Source: full-20260924-v1, 200 paired paths per setting.
\begingroup
\definecolor{numDP}{HTML}{0072B2}
\definecolor{numSTARS}{HTML}{D55E00}
\definecolor{numCE}{HTML}{009E73}
\pgfplotsset{
  num series/.style={line width=0.7pt,mark size=1.5pt},
  num errors/.style={error bars/.cd,y dir=both,y explicit,
    error bar style={line width=0.45pt},
    error mark options={rotate=90,mark size=1.4pt}}
}
\begin{tikzpicture}
\begin{groupplot}[
  group style={group size=3 by 1,horizontal sep=0.65cm,ylabels at=edge left},
  width=0.275\textwidth,height=3.8cm,scale only axis,
  xmode=normal,ymode=normal,scaled ticks=false,
  xmin=0,xmax=10500,ymin=0,ymax=12.5,
  xtick={0,5000,10000},xticklabels={0,5000,10000},
  ytick={0,3,6,9,12},
  grid=major,grid style={black!10},
  axis line style={black!55},tick style={black!55},
  tick label style={font=\fontsize{8}{9.5}\selectfont},
  label style={font=\fontsize{9}{10.5}\selectfont},
  title style={font=\fontsize{9}{10.5}\selectfont},
  xlabel={Horizon $T$},ylabel={Expected regret},
  legend style={font=\fontsize{8}{9}\selectfont,draw=none,fill=white,
    at={(0.02,0.99)},anchor=north west,legend columns=1,
    inner sep=1pt,row sep=0pt,nodes={inner sep=0.5pt}},
  legend cell align=left
]
\nextgroupplot
\addplot+[num series,color=numDP,mark=*,num errors]
  table[x=T,y=mean,y error minus=errlow,y error plus=errhigh]
  {figures/numerics/algorithm_zero_beta_0_dp.dat};
\addlegendentry{Simulation}
\addplot[color=black!80,dashed,line width=0.9pt,domain=500:10000,samples=100] {0.6422131672512151*x^(1/4)-0.7939569174171421};
\addlegendentry{$0.642T^{1/4}-0.794$}
\nextgroupplot
\addplot+[num series,color=numDP,mark=*,num errors]
  table[x=T,y=mean,y error minus=errlow,y error plus=errhigh]
  {figures/numerics/algorithm_zero_beta_1_dp.dat};
\addlegendentry{Simulation}
\addplot[color=black!80,dashed,line width=0.9pt,domain=500:10000,samples=100] {0.4725346477080489*x^(1/3)-0.8438516809091756};
\addlegendentry{$0.473T^{1/3}-0.844$}
\nextgroupplot
\addplot+[num series,color=numDP,mark=*,num errors]
  table[x=T,y=mean,y error minus=errlow,y error plus=errhigh]
  {figures/numerics/algorithm_zero_beta_2_dp.dat};
\addlegendentry{Simulation}
\addplot[color=black!80,dashed,line width=0.9pt,domain=500:10000,samples=100] {0.38491666359254967*x^(3/8)-1.162518854442828};
\addlegendentry{$0.385T^{3/8}-1.163$}
\end{groupplot}
\def\numHeaderZeroTitle{(a) $\beta=0$}
\def\numHeaderOneTitle{(b) $\beta=1$}
\def\numHeaderTwoTitle{(c) $\beta=2$}
\input{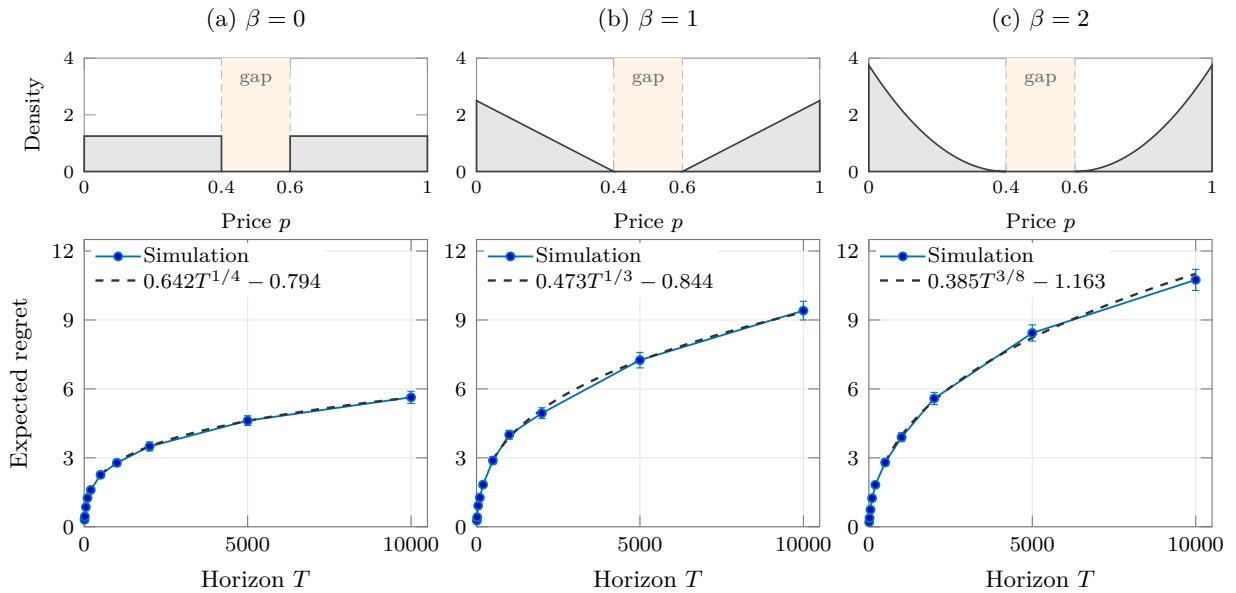}
\end{tikzpicture}
\endgroup
    \caption{Optimal online regret with zero initial inventory for the canonical positive-gap distributions, with $\Delta=0.2$ and $B=T/2$.}
    \label{fig:dp-boundary-fit}
\end{figure}

\section{Algorithmic Design Principles}
\label{sec:algorithmic-design}

Section~\ref{lb:sec-fundamental-limits} identifies how initial inventory
changes the achievable regret spectrum. We now ask how an online
policy should account for the storage boundary. We first examine
certainty equivalence for uniform prices to see where inventory
protection is needed. We then introduce \STARS{}, which generates a
holding region by averaging optimized future scenarios, state its
performance guarantees, and explain the main proof ideas.

\subsection{Fragility of the boundary-unaware certainty-equivalent heuristic}
\label{ce:section}

We begin with the simplest continuous distribution, uniform prices on
$[0,1]$. Here the fluid re-solving threshold underlying \textsf{CwG}
reduces to the certainty-equivalent threshold since there is no gap to adjust
for. The approximation balances inventory over the remaining horizon
while ignoring intermediate storage constraints. We study its natural
trading implementation, which projects each proposed trade onto the
feasible actions. The same rule has logarithmic regret with linear
initial slack, but at least square-root regret from zero initial inventory.
Since this distribution has no median gap, the latter loss isolates
the cost of boundary-unaware decisions even without the positive-gap
obstruction in Section~\ref{lb:sec-fundamental-limits}.

Write $\tau_t=T-t+1$ and $d=y-y_0$. In the fluid problem, a trading
profile $x:[0,1]\to[-1,1]$ specifies the trade at each price quantile.
The fluid problem maximizes $-\tau_t\int_0^1 F^{-1}(q)x(q)\,dq$
subject to $\tau_t\int_0^1x(q)\,dq=-d$: the remaining trades must offset
the current displacement in expectation.
Since lower prices are preferable for buying and higher prices for
selling, an optimal fluid profile buys below a quantile threshold and
sells above it. 
Consequently, the certainty-equivalent threshold is $q_t^{\textsf{CE}}(y)=1/2- (y-y_0)/(2\tau_t).$
For the uniform distribution, the price and its quantile coincide.
The policy tentatively buys if $p_t<q_t^{\textsf{CE}}(y)$ and sells if
$p_t>q_t^{\textsf{CE}}(y)$, then projects this tentative action onto 
$\mathcal A_t(y;y_0)
 =\{x\in\{-1,0,1\}:0\le y+x\le B,
                         \ |y+x-y_0|\le\tau_t-1\}.$
Projection means choosing the feasible action nearest the tentative
action. This defines $\pi^{\textsf{CE}}$ completely. The projection enforces both physical
storage feasibility and the ability to return to $y_0$ by time $T$.

\begin{proposition}
\label{ce:theorem}
Assume prices are independent and uniform on $[0,1]$, and use the
projected policy $\pi^{\textup{\textsf{CE}}}$ defined above.
\begin{enumerate}[label=$(\alph*)$]
    \item Fix $\eta>0$. There exists a constant $C=C(\eta)<\infty$
    such that, for every integer $T\ge2$ and every feasible integer $B,y_0$
    with $y_0\wedge(B-y_0)\ge\eta T$, we have
    $\Reg(\pi^{\textup{\textsf{CE}}};B,T,y_0)\le C\log T$.
    \item Fix $\rho\in(0,1)$ and let $B=\lceil\rho T\rceil$.
    There exist constants $c=c(\rho)>0$ and $T_0=T_0(\rho)<\infty$
    such that $\Reg(\pi^{\textup{\textsf{CE}}};B,T,y_0 = 0)\ge c\sqrt T$
    for every $T\ge T_0$.
\end{enumerate}
\end{proposition}

% Addressed author comment, retained for the editing record:
% \ak{the explanation can be a bit more detailed here. the first part of the proposition you can just refer to the previous section where the guarantees follows from CwG. no need to explainig that. You should try to emphasize the \sqrt{T} lower bound when starting close to zero. Basically you want to say that lets consider a policy that always buys below 1/2 and sells above 1/2 and basically this policy will hit the lower boundary around \sqrt{T} times (this follows from the SSRW hitting the zero state), now the CE will also incerase the downward drift because it will always try to come to the initial state of zero, so I will only hit the zero state more often and whenever I hit zero basically there is a cahnce that I would have been able to sell at the next price had I not sold now and made some profit which the hindsight could have avoided by swapping the decisions and this leads to regret of the order \sqrt{T}}

The result in Proposition~\ref{ce:theorem}(a) follows from the result in 
Proposition~\ref{prop:interior-transfer} for the uniform distribution, where \textsf{CwG} reduces to the certainty equivalent heuristic. We focus here on why the same principle incurs a regret of $\Omega(\sqrt{T})$ when the initial inventory is zero.

For reference, consider the static median-sign rule from
Section~\ref{lb:sec-fundamental-limits}, with sales blocked at zero:
buy below $1/2$, sell above $1/2$ when inventory is positive, and hold
when a sale is requested while the inventory level is zero. Temporarily ignore the
upper boundary and the terminal requirement. The rule's proposed trades
form the SSRW described in that section; a proposed sale at zero is
replaced by a hold. The expected number of these blocked sales is of
order $\sqrt T$, and the static median-sign rule consequently spends
order $\sqrt T$ time steps at zero inventory.

The certainty equivalent heuristic adds a downward pull to the static median-sign rule. When $y_0=0$, its threshold
is $1/2-y/(2\tau_t)$: at positive inventory, it sells at some prices
below $1/2$ at which the static median-sign rule would buy. The expected
tentative change in inventory is $-y/\tau_t$. Under the same price
sequence, its inventory stays below that of the static median-sign rule, so
its inventory level is zero whenever the static median-sign rule's inventory
level is zero. The capacity and
terminal constraints preserve this ordering. Thus \textsf{CE} spends
at least order $\sqrt T$ time steps at zero in expectation.

To see how this leads to a regret of $\Omega(\sqrt{T})$, consider an occasion when \textsf{CE}
has one unit left. Suppose the current price lies in $[1/2,3/4]$ and
the next price lies in $[7/8,1]$. \textsf{CE} sells its last unit at
the current price and then holds at the next, higher price because its
inventory is zero. Hindsight can instead hold for one time step and
sell at the higher price. This increases profit by at least $1/8$,
uses only one unit of storage, and leaves exactly the same inventory
after the two time steps. Conditional on having one unit at the start,
the two price intervals occur with probability $1/32$.

There are sufficiently many such opportunities. At zero inventory
before the final step, the next price causes a purchase with probability
$1/2$. Each departure from zero must eventually be followed by a return
from inventory one to zero, since the terminal inventory is zero.
Consequently, the expected number of time steps with one unit remaining
is also at least order $\sqrt T$. Apart from at most the final step,
each permits the two-step improvement above. These improvements do not
overlap: after a selected sale, the original \textsf{CE} inventory is
zero, so another selected pair cannot start at the next step. Their
gains can therefore be added along a single feasible hindsight path.
Appendix~\ref{ce:appendix} makes this counting argument precise and
obtains the $\Omega(\sqrt T)$ regret bound.

The example identifies the additional task in algorithm design:
preserve inventory for valuable selling opportunities before the inventory becomes zero. One can encode this protection through an
inventory-dependent holding region, whose width depends on the price
law and the remaining horizon. We next show how scenario optimization
generates such a region directly from feasible future trades.

\subsection{\textsf{STARS}: scenario optimization generates a holding region}
\label{sec:stars}

% We specialize the \textsf{RAMS} idea of \cite{besbes2025} to trading with
% storage. Each simulated continuation respects the inventory constraints,
% so its value reflects the selling opportunities that the current
% inventory makes possible. Averaging these values allows the policy to
% adapt its holding decisions without specifying a distribution-dependent
% reserve or repulsion rule.

% \subsubsection{The algorithm}
\STARS{} stands for \textsf{S}torage \textsf{T}rading by \textsf{A}veraging \textsf{R}epeatedly across multiple \textsf{S}cenarios. The algorithm uses a simulator to evaluate the future trading opportunities created by each feasible current action. At time step $t$, a candidate action changes inventory to $y_{t-1}+x$. For each simulated future price sequence, we compute the best continuation from that inventory to the target $y_0$, using a deterministic hindsight optimizer. We use the same scenarios for all candidate actions, average their profits, and execute only the maximizing current action. %The future scenario plans are discarded, and the procedure is repeated at the next time step.

To define this precisely, let $V_{t+1}^{\mathsf{hs}}(z;\mathbf p_{>t})$ be the maximum profit over time steps $t+1,\ldots,T$, starting from $z$, ending at $y_0$, and respecting capacity and unit-trade constraints. Its domain is $\mathcal Y_{t+1}=\{z\in\{0,\ldots,B\}:|z-y_0|\le T-t\}$, with terminal value $V_{T+1}^{\mathsf{hs}}(y_0;\varnothing)=0$.
The optimizer knows the prices within its scenario, so this continuation differs from an optimal nonanticipative policy. Feasible actions in \eqref{eq:feasible-actions} ensure that every continuation value evaluated below is defined.

\begin{algorithm}[htbp]
\caption{\STARS{}: Storage Trading by Averaging Repeatedly across multiple Scenarios}
\label{alg:stars}
\begin{algorithmic}[1]
\Require $T,B,y_0$, a simulator for $F$, and positive sample counts $K_t$ chosen before observing $p_t$.
\For{$t=1,\ldots,T$}
 \State Observe $p_t$ and inventory $y_{t-1}$.
 \State Form $\mathcal A_t(y_{t-1})$ using \eqref{eq:feasible-actions}.
 \Comment{Preserve a feasible completion}
 \State Draw independent $\widetilde{\mathbf p}_{>t}^{\,(i)}\sim F^{\otimes(T-t)}$, $i=1,\ldots,K_t$.
 \Comment{Fresh future scenarios}
 \For{each $x\in\mathcal A_t(y_{t-1})$}
  \State Evaluate $V_{t+1}^{\mathsf{hs}}(y_{t-1}+x;\widetilde{\mathbf p}_{>t}^{\,(i)})$ for all $i$.
  \Comment{Reuse the same scenarios}
 \EndFor
 \State Choose a maximizing feasible action
 \Statex \hspace{\algorithmicindent}$\displaystyle\begin{aligned}
 x_t\in\arg\max_{x\in\mathcal A_t(y_{t-1})}\bigg\{&-p_tx+\frac1{K_t}\sum_{i=1}^{K_t}
 V_{t+1}^{\mathsf{hs}}(y_{t-1}+x;\widetilde{\mathbf p}_{>t}^{\,(i)})\bigg\}.
 \end{aligned}$
 \State Execute $x_t$ and set $y_t=y_{t-1}+x_t$.
 \Comment{Implement only this action}
\EndFor
\end{algorithmic}
\end{algorithm}
Ties are resolved in the order hold, buy, sell, restricted to feasible maximizers. Scenarios are independent of all actual prices and previous simulations. Throughout the performance guarantees, \STARS{} means this finite-scenario algorithm, with $K_t\ge(T-t+1)^2$ chosen before observing $p_t$. We also write $\pi^K$ when the sampling schedule needs to be explicit.

% \subsubsection{What the scenario average is estimating}
An additional unit is valuable because it changes what trades can be executed in the future. For adjacent feasible inventories define its realized and expected marginal continuation values:
\begin{align*}
 \Delta V_{t+1}^{\mathsf{hs}}(y;\mathbf p_{>t})
 &=V_{t+1}^{\mathsf{hs}}(y;\mathbf p_{>t})-V_{t+1}^{\mathsf{hs}}(y-1;\mathbf p_{>t}),\\*
 \lambda_{t+1}(y)&=\E_{\mathbf p_{>t}}\left[\Delta V_{t+1}^{\mathsf{hs}}(y;\mathbf p_{>t})\right],\\*
 \widehat\lambda_{t+1}(y)&=\frac1{K_t}\sum_{i=1}^{K_t}
 \Delta V_{t+1}^{\mathsf{hs}}(y;\widetilde{\mathbf p}_{>t}^{\,(i)}).
\end{align*}
Here $\E_{\mathbf p_{>t}}$ averages a fresh i.i.d.\ future while keeping the stated inventory fixed. With holding feasible at $y$, subtracting its common continuation value from the three scenario-averaged profits gives
\begin{align*}
 \underbrace{-p_t+\widehat\lambda_{t+1}(y+1)}_{\text{buy}},\qquad
 \underbrace{0}_{\text{hold}},\qquad
 \underbrace{p_t-\widehat\lambda_{t+1}(y)}_{\text{sell}}.
\end{align*}

To understand what these sample averages estimate, let
$\pi^\circ$ be the auxiliary rule that substitutes $\lambda_{t+1}$ for
$\widehat\lambda_{t+1}$, with the same feasibility and tie conventions.
We use this reference rule in the analysis; the implemented algorithm and
the performance theorem below use finitely many scenarios.

\begin{proposition}
\label{prop:main-concavity}
Fix a terminal target $y_0$ and a decision epoch $t$.
\begin{enumerate}[label=$(\alph*)$]
    \item For every fixed future price sequence $\mathbf p_{>t}$,
    the function $V_{t+1}^{\mathsf{hs}}(\cdot;\mathbf p_{>t})$ is
    discretely concave on $\mathcal Y_{t+1}$. Its expectation and
    its empirical average are also discretely
    concave. Consequently, the realized, expected, and sample-average
    marginal values are nonincreasing in inventory.
    \item At every inventory $y$ where all three actions are feasible,
    \textup{\STARS{}} uses the rule
    \begin{align*}
     x_t=\begin{cases}
     +1,&p_t<\widehat\lambda_{t+1}(y+1),\\
     0,&\widehat\lambda_{t+1}(y+1)\le p_t\le\widehat\lambda_{t+1}(y),\\
     -1,&p_t>\widehat\lambda_{t+1}(y).
     \end{cases}
    \end{align*}
\end{enumerate}
\end{proposition}
The proposition explains why scenario optimization generates a holding region. Selling gives up the continuation value of the last stored unit, whereas buying adds the next, weakly less valuable unit. The policy holds when neither transaction compensates for the resulting change in continuation value. This hold region adapts to inventory, the remaining horizon, and the distribution without requiring a gap width, edge exponent, or target buffer as an input. % Appendix~\ref{app:basic-facts} gives a pathwise rounding proof valid for every terminal target $y_0$.

\subsection{Performance guarantees of \STARS{}}
\label{sec:stars-performance}
\begin{theorem}
\label{thm:stars-upper}
Fix $\rho\in(0,1)$, let $B=\lceil\rho T\rceil$, and assume $y_0=y_T=0$.
Run Algorithm~\ref{alg:stars} with fresh common scenarios and predictable
sample counts $K_t\ge(T-t+1)^2$.
\begin{enumerate}[label=$(\alph*)$]
    \item Fix $\mu\in(0,1/2]$ and $\Delta\in[0,1]$. For every
    distribution $F\in\mathcal F^{\mathrm{at}}_{\mu,\Delta}$,
    there exists a constant $C=C(F,\rho)<\infty$ such that
    $\Reg(\textup{\STARS};B,T,0)\le C\log T$.
    \item Fix $\beta\ge0$ and $\varepsilon_0\in(0,1/4]$, and let
    $\Delta=0$. For every distribution
    $F\in\mathcal F^{\mathrm{ac}}_{\beta,0,\varepsilon_0}$,
    there exists a constant $C=C(F,\rho)<\infty$ such that
    \begin{align*}
     \Reg(\textup{\STARS};B,T,0)
     \le C\log^2 T\,\mathbbm{1}\{\beta=0\}
       +CT^{\beta/(2\beta+2)}\,\mathbbm{1}\{\beta>0\}.
    \end{align*}
    \item Fix $\beta\ge0$, $\varepsilon_0\in(0,1/4]$, and
    $\Delta\in(0,1)$. For every distribution
    $F\in\mathcal F^{\mathrm{ac}}_{\beta,\Delta,\varepsilon_0}$,
    there exists a constant $C=C(F,\rho)<\infty$ such that
    \begin{align*}
     \Reg(\textup{\STARS};B,T,0)
     \le CT^{1/4}\log T\,\mathbbm{1}\{\beta=0\}
       +CT^{(\beta+1)/(2\beta+4)}\,\mathbbm{1}\{\beta>0\}.
    \end{align*}
\end{enumerate}
% Each bound holds for every integer $T\ge2$ and every permitted sampling
% schedule. The dependence on $F$ includes its regularity moduli in
% Definition~\ref{def:continuous-family}; the constants are independent of
% $T$ and the sampling schedule and may vary with the fixed gap $\Delta$.
\end{theorem}
% On the canonical positive-gap instances,
{For every fixed $F\in\mathcal F^{\mathrm{ac}}_{\beta,\Delta,\varepsilon_0}$
with $\Delta>0$,}
these upper bounds match the
polynomial exponent of Theorem~\ref{lb:thm-positive-gap}.
{The bounds match in order for $\beta>0$, and up to a
factor $\log T$ for $\beta=0$.}
All the guarantees apply to the same
finite-scenario decision rule, although the proof treats sampling
differently for continuous and atomic prices. For continuous prices,
Appendix~\ref{app:stars-continuous} bounds the reference rule and
Appendix~\ref{app:stars-sampling} transfers that bound to
Algorithm~\ref{alg:stars}. For atomic prices,
Appendix~\ref{app:finite-atomic} extends the compensation and inventory
estimates of Appendix~\ref{app:stars-atomic} directly to the sampled
trajectory.

For a continuous distribition with
density at most $M$, the additional regret
relative to ideal \STARS{} (which exactly computes the expected marginal value $\lambda_{t + 1}$) $\pi^\circ$ is at most
$M\E\sum_{t=1}^{T-1}K_t^{-1/2}\le M\log T$.
% At an atom, a small threshold error can change an action with positive
% probability, so this action-comparison argument is unavailable.
% Instead, the sampled conditional compensation is at most twice its
% population-threshold counterpart plus $C/K_t$, and the inventory proof
% controls the sampled path itself. The atomic guarantee therefore
% accounts for the states visited by the sampled policy directly.
For uniform prices there is a stronger special identity: at every
state the sampled inventory transition probabilities equal those of
$\pi^\circ$. For any fixed $K\ge1$, the regret is exactly
$(1+1/K)\Reg(\pi^\circ;B,T,0)$ and this fixed-$K$ exception is proved in
Appendix~\ref{app:basic-facts}.

% Replacement overview approved by the author; original Section 4.4 is retained below.
\begingroup

\subsection{Proof ideas: compensation and inventory dynamics}
\label{sec:proof-ideas}

At time step $t$, give hindsight the online policy's current inventory
$y$. Hindsight can choose whichever feasible buy, hold, or sell action
yields the highest profit. Now force its current action to be the online
policy's choice $x$, and let it optimize all subsequent trades from
$y+x$. The loss from this restriction is the \emph{marginal compensation}
\begin{align*}
 \partial R_t(y,x;\mathbf p_{\ge t})
 =V_t^{\mathsf{hs}}(y;\mathbf p_{\ge t})
 -\left(-p_tx+V_{t+1}^{\mathsf{hs}}(y+x;\mathbf p_{>t})\right)
 \ge0.
\end{align*}

\begin{proposition}
\label{prop:main-decomposition}
Fix $T,B$, an initial and terminal inventory $y_0$, and a feasible
online policy $\pi$.
\begin{enumerate}[label=$(\alph*)$]
\item For every realized price sequence and policy trajectory,
\begin{align*}
 V_1^{\mathsf{hs}}(B,T,y_0;\mathbf p)
 -\sum_{t=1}^T(-p_tx_t^\pi)
 =\sum_{t=1}^T
 \partial R_t(y_{t-1}^\pi,x_t^\pi;\mathbf p_{\ge t}).
\end{align*}
\item Consequently,
\begin{align*}
 \Reg(\pi;B,T,y_0)
 =\sum_{t=1}^T\E\left[
 \E\left[
 \partial R_t(y_{t-1}^\pi,x_t^\pi;\mathbf p_{\ge t})
 \mid\mathcal F_{t-1}
 \right]\right].
\end{align*}
\end{enumerate}
\end{proposition}

\begin{proof}
Substitute the inventory update $y_t^\pi=y_{t-1}^\pi+x_t^\pi$
into the definition of compensation and sum over time. Successive
continuation values cancel, and the terminal value
$V_{T+1}^{\mathsf{hs}}(y_0;\varnothing)$ is zero, giving part (a).
Taking expectations and conditioning each summand on the past gives
part (b). The inner expectation fixes the current inventory and
averages over the current price, actual future, and fresh policy
randomization.
\end{proof}

This is the compensated-coupling decomposition of
\citet{vera2021bayesian} and \citet{besbes2025}. We now explain how it
% is used for \STARS{}, with zero initial and terminal inventory and the
% capacity scaling of Theorem~\ref{thm:stars-upper}. The difficulty is
is used for \STARS{}, with zero initial and terminal inventory and
capacity proportional to the horizon, as in
Theorem~\ref{thm:stars-upper}.

The difficulty is
that compensation need not become small merely because many periods
remain. Even for uniform prices, every feasible online policy has
expected compensation bounded below by a positive constant whenever
its inventory is zero and at least two decisions remain
% (Appendix~\ref{app:nonvanishing-proof}). We must therefore account
(Appendix~\ref{app:nonvanishing-proof}). A worst-case bound over all
inventory levels therefore cannot give sublinear regret. We must account
for both the loss incurred at an inventory level and how often the
policy visits it.

To see how compensation arises, consider a state where all three
actions are feasible. Let $b^{\mathsf{hs}}$ and $s^{\mathsf{hs}}$ be
hindsight's buying and selling thresholds, and let $\widehat b$ and
$\widehat s$ be the corresponding scenario estimates used by \STARS{}.
These are the marginal continuation values defined in
% Section~\ref{sec:stars}. In Figure~\ref{fig:compensation-thresholds}(a),
Section~\ref{sec:stars}: $b^{\mathsf{hs}}$ is the value of adding the
next unit, whereas $s^{\mathsf{hs}}$ is the value lost by giving up
the last stored unit.

In Figure~\ref{fig:compensation-thresholds}(a),
\STARS{} buys at price $p_t$, while hindsight would hold. The extra
unit is worth only $b^{\mathsf{hs}}$ in the actual future, so compensation
is $p_t-b^{\mathsf{hs}}$. In panel (b), \STARS{} sells a unit worth
$s^{\mathsf{hs}}$ for only $p_t$, losing $s^{\mathsf{hs}}-p_t$.
Compensation thus depends on both how often decisions disagree and
how costly those disagreements are.

% The retained figure uses semantic colors, not revision colors.
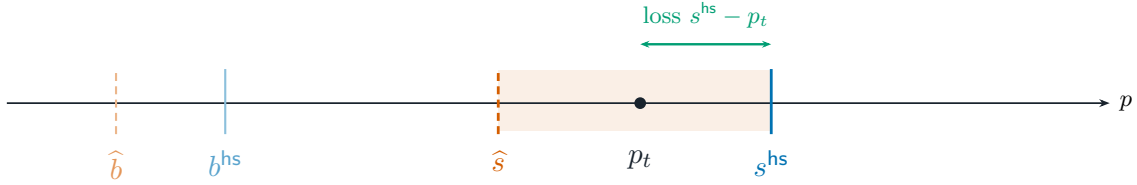
\begingroup
\color{black}
\begin{figure}[htbp]
\centering
\begingroup
\definecolor{compblue}{HTML}{0072B2}
\definecolor{comporange}{HTML}{D55E00}
\definecolor{compteal}{HTML}{009E73}
\definecolor{compink}{HTML}{19232D}
\begin{tikzpicture}
\begin{groupplot}[
    group style={group size=1 by 2,vertical sep=1.15cm},
    width=0.91\linewidth,height=2.05cm,scale only axis,
    xmin=0,xmax=1.04,ymin=0,ymax=1,
    hide axis,clip=false,
    every axis plot/.append style={no marks},
    every node/.append style={font=\small},
]
\nextgroupplot
\node[anchor=west,font=\small\bfseries] at (axis cs:0,1.30)
    {(a) Hindsight holds; \STARS{} buys};
\fill[comporange,opacity=0.10]
    (axis cs:0.30,0.27) rectangle (axis cs:0.55,0.66);
\draw[compink,-{Stealth[length=4pt]},line width=0.65pt]
    (axis cs:0,0.45) -- (axis cs:1.01,0.45);
\node[anchor=west,font=\footnotesize] at (axis cs:1.01,0.45) {$p$};
\draw[compblue,line width=1.05pt]
    (axis cs:0.30,0.25) -- (axis cs:0.30,0.67);
\node[compblue,anchor=north] at (axis cs:0.30,0.21) {$b^{\mathsf{hs}}$};
\draw[comporange,densely dashed,line width=1.05pt]
    (axis cs:0.55,0.25) -- (axis cs:0.55,0.67);
\node[comporange,anchor=north] at (axis cs:0.55,0.21) {$\widehat b$};
\draw[compblue!45,line width=0.8pt]
    (axis cs:0.70,0.25) -- (axis cs:0.70,0.67);
\node[compblue!65,anchor=north] at (axis cs:0.70,0.21) {$s^{\mathsf{hs}}$};
\draw[comporange!45,densely dashed,line width=0.8pt]
    (axis cs:0.82,0.25) -- (axis cs:0.82,0.67);
\node[comporange!65,anchor=north] at (axis cs:0.82,0.21) {$\widehat s$};
\fill[compink] (axis cs:0.43,0.45) circle (2.2pt);
\node[compink,anchor=north] at (axis cs:0.43,0.21) {$p_t$};
\draw[compteal,{Stealth[length=3.5pt]}-{Stealth[length=3.5pt]},line width=0.9pt]
    (axis cs:0.30,0.83) -- (axis cs:0.43,0.83);
\node[compteal,anchor=south,font=\footnotesize]
    at (axis cs:0.365,0.87) {loss $p_t-b^{\mathsf{hs}}$};

\nextgroupplot
\node[anchor=west,font=\small\bfseries] at (axis cs:0,1.30)
    {(b) Hindsight holds; \STARS{} sells};
\fill[comporange,opacity=0.10]
    (axis cs:0.45,0.27) rectangle (axis cs:0.70,0.66);
\draw[compink,-{Stealth[length=4pt]},line width=0.65pt]
    (axis cs:0,0.45) -- (axis cs:1.01,0.45);
\node[anchor=west,font=\footnotesize] at (axis cs:1.01,0.45) {$p$};
\draw[comporange!45,densely dashed,line width=0.8pt]
    (axis cs:0.10,0.25) -- (axis cs:0.10,0.67);
\node[comporange!65,anchor=north] at (axis cs:0.10,0.21) {$\widehat b$};
\draw[compblue!45,line width=0.8pt]
    (axis cs:0.20,0.25) -- (axis cs:0.20,0.67);
\node[compblue!65,anchor=north] at (axis cs:0.20,0.21) {$b^{\mathsf{hs}}$};
\draw[comporange,densely dashed,line width=1.05pt]
    (axis cs:0.45,0.25) -- (axis cs:0.45,0.67);
\node[comporange,anchor=north] at (axis cs:0.45,0.21) {$\widehat s$};
\draw[compblue,line width=1.05pt]
    (axis cs:0.70,0.25) -- (axis cs:0.70,0.67);
\node[compblue,anchor=north] at (axis cs:0.70,0.21) {$s^{\mathsf{hs}}$};
\fill[compink] (axis cs:0.58,0.45) circle (2.2pt);
\node[compink,anchor=north] at (axis cs:0.58,0.21) {$p_t$};
\draw[compteal,{Stealth[length=3.5pt]}-{Stealth[length=3.5pt]},line width=0.9pt]
    (axis cs:0.58,0.83) -- (axis cs:0.70,0.83);
\node[compteal,anchor=south,font=\footnotesize]
    at (axis cs:0.64,0.87) {loss $s^{\mathsf{hs}}-p_t$};
\end{groupplot}
\end{tikzpicture}

\smallskip
{\footnotesize
\textcolor{compblue}{Solid: hindsight thresholds}\qquad
\textcolor{comporange}{Dashed: \STARS{} thresholds}\qquad
\textcolor{compink}{Dot: current price}}
\endgroup
\caption{Marginal compensation at a fixed inventory with all three actions
feasible, after fixing the actual and simulated futures. The shaded interval
is where the indicated threshold disagreement changes the action. }
\label{fig:compensation-thresholds}
\end{figure}

\endgroup

The same continuation values also determine inventory movement. To
explain this connection, first consider the mean-threshold rule
$\pi^\circ$ from Section~\ref{sec:stars}, which replaces scenario
% averages by their exact expectations. This separates the effect of
% inventory from simulation error.
averages by their exact expectations. This is still an online rule:
it averages over possible futures without observing the future that
will occur. Using this rule first lets us study inventory movement
without also tracking simulation error.

% Uniform prices on $[0,1]$ provide the simplest example. Remove the
% upper capacity and let the number of future periods grow, while
% retaining the requirement to end with zero inventory. The limiting
% expected value of increasing inventory from $y-1$ to $y$, denoted by
% $\lambda_\infty(y)$, satisfies
Uniform prices on $[0,1]$ provide the simplest example. Remove the
upper capacity, fix an inventory level $y\ge1$, and let the number
of future periods grow. Each finite continuation problem still has
to end with zero inventory. Let $\lambda_\infty(y)$ denote the
limiting expected extra profit from starting with $y$ units rather
than $y-1$. For large $y$, we obtain
\begin{align*}
 \lambda_\infty(y)
 =\frac12+\frac{1}{4y}+O(y^{-2})
 \qquad\text{as }y\to\infty.
\end{align*}
An additional unit protects against high-price selling opportunities
arriving before low-price replenishments. This protection makes its
% value exceed the median, but becomes less valuable at higher
% inventories. In the limiting calculation, the buying probability is
value exceed the median $1/2$. The premium is smaller when inventory
is already plentiful, because another unit then offers less protection.

This premium also explains why the policy tends to build inventory
when little is available. For a uniform price on $[0,1]$, a buying
threshold $b$ gives buying probability $b$, while a selling threshold
$s$ gives selling probability $1-s$. In the limiting calculation,
these probabilities are
% $\lambda_\infty(y+1)$ and the selling probability is
$\lambda_\infty(y+1)$ and
% $1-\lambda_\infty(y)$. Their difference is
% $1/(2y)+O(y^{-2})$: the continuation values generate an upward inventory
% bias that weakens as inventory increases.
$1-\lambda_\infty(y)$, respectively. Since a purchase adds one unit
and a sale removes one, the expected inventory change is their
difference, $1/(2y)+O(y^{-2})$. Thus the continuation values generate
an upward inventory bias that weakens as inventory increases.
Appendix~\ref{app:stars-continuous} controls the difference between
these limiting values and the actual finite-horizon thresholds.

% The proof must then quantify the costly visits that still occur
% despite this upward bias. For uniform prices, the principal
% low-inventory contribution to compensation is controlled by
% $(y+1)^{-2}$. A visit to zero inventory receives weight one, while
% visits at larger inventories receive smaller weights.
An upward bias does not prevent the policy from returning to low
inventory. To bound regret, we must also control the cost of these
visits. For uniform prices, the compensation estimates give a
low-inventory term bounded by a constant times $(y+1)^{-2}$, in
addition to terms associated with the finite horizon and capacity.
We therefore count inventory visits with weight $(y+1)^{-2}$:
a visit to zero receives weight one, a visit to one receives weight
$1/4$, and visits to larger inventories receive progressively less
weight. These weights bound a contribution to expected compensation;
they are not the realized loss on each visit.

% Appendix~\ref{app:stars-continuous} uses the threshold estimates to
% prove that
% To bound this weighted sum, we construct a function of inventory whose
% expected one-period increase accounts for the current cost $(y+1)^{-2}$,
% up to finite-horizon and capacity errors. Summing these increases through
% the penultimate decision cancels the intermediate values along the
% inventory path. The remaining endpoint contribution is bounded because
% inventory is at most one before the final decision. The final decision
% adds at most one to the weighted sum. Appendix~\ref{app:stars-continuous} bounds
% the accumulated errors by $O(\log^2 T)$, giving
% $\E\sum_{t=1}^T(y_{t-1}^{\pi^\circ}+1)^{-2}=O(\log^2 T)$.
The proof bounds this weighted count by constructing an increasing
function $\Phi$ of inventory. Its expected one-period increase is
at least the current weight $(y+1)^{-2}$, apart from finite-horizon
and capacity error terms. This turns a cost incurred repeatedly
along the path into changes of a single function. Taking expectations
and summing through the penultimate decision, the intermediate
function values cancel, leaving only
$\E\Phi(y_{T-1}^{\pi^\circ})-\Phi(0)$. To finish with zero inventory,
the policy can have at most one unit before the final decision,
so this endpoint contribution is bounded by a constant. The final
decision itself adds at most one to the weighted count.

It remains to control the error terms in this accounting argument.
The threshold and inventory-moment estimates in
Appendix~\ref{app:stars-continuous} bound their sum in expectation by
$O(\log^2 T)$. Consequently,
\begin{align*}
 \E\sum_{t=1}^T(y_{t-1}^{\pi^\circ}+1)^{-2}=O(\log^2 T).
\end{align*}
The remaining compensation terms, including those associated with
terminal liquidation and capacity, contribute $O(\log T)$.
Together, these estimates give the uniform-price regret bound.

% A median gap changes the finite-horizon reduction in continuation
% values. Consider the canonical distribution $F_{0,\Delta}$, which is
% uniform on two intervals separated by a positive gap. With a long
% future, the expected marginal value lies above the upper gap edge
% $a^+$. With finitely many future periods, however, some price
% sequences give the additional unit a marginal value on the lower
% side of the gap. The resulting reduction in its expected value
% includes the gap width, not just price differences within the
% occupied intervals. Lower continuation values make the policy less
% willing to buy and more willing to sell.
Now introduce a median gap. Consider the canonical distribution
$F_{0,\Delta}$, which is uniform on two intervals separated by a
gap of width $\Delta$. For each fixed positive inventory, the limiting
expected marginal value lies above the upper gap edge $a^+$.
With limited time left, however, the additional unit can be much
less valuable. For some future price sequences, its marginal value
lies in the lower price interval, at or below $a^-$.

For example, take prices uniform on $[0,0.4]\cup[0.6,1]$ and suppose
only one future period remains. Starting with one unit rather than
zero adds exactly the final price to hindsight profit: the unit must
be sold, whereas with zero inventory the only feasible action is to
hold. Thus the realized marginal value is the final price itself.
It is at most $0.4$ with probability $1/2$, and its expectation is
$0.5$, below the upper gap edge $0.6$. Although no price lies inside
the gap, an expected marginal value can. With more future periods,
the marginal value reflects an entire optimal trading sequence;
we must quantify how often it still falls below the gap.

The important feature is that a marginal value on the lower side
falls short of $a^+$ by at least the full gap width $\Delta$.
In the uncapped continuation problem, the reduction from the
long-future expected marginal value includes exactly $\Delta$
times the probability of falling on that side, in
addition to the effects of price differences within the two
intervals. Lower expected marginal values lower both buying and
selling thresholds, making the mean-threshold rule less willing to
buy and more willing to sell.

% More precisely, write $\tau=T-t\ge1$ for the number of future prices.
% At low inventories away from the upper capacity, apart from a negligible
% capacity correction, the
% finite-horizon reduction is bounded by a constant multiple of $y/\tau$
% without a gap and of $\Delta y/\sqrt{\tau}+y/\tau$ with a gap.
% Comparing these corrections with the inventory premium of order $1/y$
% gives an inventory scale of order $\sqrt{\tau}$ without a gap, and
% $\tau^{1/4}$ for a fixed positive gap and sufficiently large
% $\tau$. Below a sufficiently small constant multiple of each scale,
% the estimates guarantee a strong upward inventory bias. These scales
% describe where this protection can be established; they are neither
% inventory targets nor regret bounds.
The proof quantifies this effect as follows. Write $\tau=T-t\ge1$
for the number of future prices, and let $\lambda_\tau^B(y)$ be the
expected marginal value of the $y$th unit with this finite future
and capacity $B$. At low feasible inventories away from the upper
capacity, Appendix~\ref{app:stars-continuous} bounds
$\lambda_\infty(y)-\lambda_\tau^B(y)$ by a constant multiple of
% \begin{align*}
%  \frac{\Delta y}{\sqrt{\tau}}+\frac{y}{\tau},
% \end{align*}
$\Delta y/\sqrt{\tau}+y/\tau$,
apart from a negligible capacity correction. The first term arises
when the hindsight marginal value falls below the gap; this event
has probability bounded by a constant times
$y/\sqrt{\tau}$. Without a gap, this term vanishes and only the
$y/\tau$ correction remains.

To see what this means for inventory movement, compare the reduction
with the long-future inventory premium, which is of order $1/y$ in
both examples. Without a gap, $y/\tau$ is small relative to $1/y$
when $y$ is small relative to $\sqrt{\tau}$. With a fixed positive
gap, the additional term $\Delta y/\sqrt{\tau}$ instead requires
$y$ to be small relative to $\tau^{1/4}$. More precisely, below a
sufficiently small constant multiple of the respective scale, the
estimates preserve an upward inventory bias of order $1/y$.
The comparison applies for sufficiently large $\tau$, with constants
depending on the fixed distribution. These are ranges where the
proof guarantees inventory protection, not target inventory levels;
the comparison alone is not a regret bound.

% Appendix~\ref{app:stars-continuous} quantifies how this effect weakens
% inventory protection and obtains the positive-gap upper bounds in
Appendix~\ref{app:stars-continuous} uses these inventory-dependent
estimates to bound accumulated compensation and obtain the
% positive-gap upper bounds in
% Theorem~\ref{thm:stars-upper}. It also extends the analysis to the
positive-gap upper bounds in Theorem~\ref{thm:stars-upper}. For
$F_{0,\Delta}$ with a fixed positive gap, the resulting bound is
$O(T^{1/4}\log T)$, compared with $O(\log^2 T)$ for uniform prices
without a gap. The appendix also extends the analysis to the
continuous class in Definition~\ref{def:continuous-family}, where
$\beta$ describes the mass near the median edges.

Atomic prices can instead produce a logarithmic inventory reserve.
Consider the equal-mass example on $\{0.2,0.4,0.6,0.8\}$, again
temporarily removing the upper capacity. At price $0.6$, the
mean-threshold rule may hold or buy when inventory is low, but sells
% when inventory is plentiful. The switching inventory grows
when inventory is plentiful. The inventory level separating these
behaviors grows
logarithmically with the number of future periods. The rule continues
to sell at $0.8$ whenever inventory is available, so its reluctance
to sell at $0.6$ preserves stock for more valuable opportunities.
This reserve emerges from the continuation values rather than being
specified as an input. Appendices~\ref{app:stars-atomic}
and~\ref{app:finite-atomic} combine the compensation and inventory
estimates to obtain the $O(\log T)$ guarantee for the finite-atomic
class.

Finally, the guarantees concern finite-scenario \STARS{}, not only
the mean-threshold rule used to explain these mechanisms. For
continuous prices, the comparison in Section~\ref{sec:stars-performance}
adds at most $O(\log T)$ regret under the prescribed sampling schedule.
For atomic prices, the proof controls compensation along the sampled
policy's own inventory path. Both arguments account for the effect
of sampling on subsequent decisions and inventories.
\par
\endgroup

\section{Numerical experiments}
\label{sec:numerics}

We perform numerical experiments to study the performance of \STARS{}
and the effect of the initial inventory. We consider the canonical
positive-gap distributions $F_{\beta,\Delta}$ in
Example~\ref{ex:positive-gap}, with $\Delta=0.2$ and
$\beta\in\{0,1,2\}$. We vary the horizon over
$T\in\{10,20,50,100,200,500,1000,2000,5000,10{,}000\}$ and set
$B=T/2$. We consider two initial states: zero inventory, $y_0=0$, and an interior state, $y_0=\lfloor B/2\rfloor$.

We compare three policies. \textsf{Dynamic Programming (DP)} computes the optimal online value function by backward induction, using the known price distribution and enforcing all storage and terminal constraints. For the distributions considered here, the expectations in the recursion are evaluated analytically without any price discretization.
The \textsf{DP} provides the optimal online benchmark.
\STARS{} uses Algorithm~\ref{alg:stars} with $K_t=50$ fresh future
scenarios at each decision epoch with a choice of actions. %\footnote{Theorem~\ref{thm:stars-upper} uses $K_t\ge(T-t+1)^2$; here we assess practical performance with a fixed scenario budget.} 
Finally, certainty equivalence (\textsf{CE}) uses the fluid quantile $q_t^{\textsf{CE}}(y)=1/2-(y-y_0)/(2(T-t+1))$.
It tentatively buys below $F^{-1}(q_t^{\textsf{CE}}(y))$ and sells above
it, then projects the action onto the feasible set, as in
Section~\ref{ce:section}. 

For each distribution, horizon, and initial state, we generate $200$
independent price sequences and evaluate all three policies on the same
sequences. On each sequence, we calculate regret as the optimal
feasible hindsight profit minus the policy's realized profit. We plot the sample mean
of these paired differences and $95\%$ confidence
intervals across replications. Both axes are linear in all figures.

We first compare the three policies when $y_0=y_T=0$ in Figure~\ref{fig:numerics-zero}. Across the three distributions, \STARS{} tracks the optimal online \textsf{DP} substantially more
closely than \textsf{CE} does, and the separation from \textsf{CE} grows over the plotted horizons. These observations complement Figure~\ref{fig:dp-boundary-fit}, where the fitted \textsf{DP} curves are consistent with the polynomial exponent $(\beta+1)/(2\beta+4)$ in Theorems~\ref{lb:thm-positive-gap} and~\ref{thm:stars-upper}. The fixed-budget \STARS{} policy performs close to this benchmark. In contrast, \textsf{CE} exhibits growth close to $\sqrt T$ over the tested horizons.

\begin{figure}[htbp]
    \centering
    % Native PGFPlots input. Source: full-20260924-v1, 200 paired paths per setting.
\begingroup
\definecolor{numDP}{HTML}{0072B2}
\definecolor{numSTARS}{HTML}{D55E00}
\definecolor{numCE}{HTML}{009E73}
\pgfplotsset{
  num series/.style={line width=0.7pt,mark size=1.5pt},
  num errors/.style={error bars/.cd,y dir=both,y explicit,
    error bar style={line width=0.45pt},
    error mark options={rotate=90,mark size=1.4pt}}
}
\begin{tikzpicture}
\begin{groupplot}[
  group style={group size=3 by 1,horizontal sep=0.65cm,ylabels at=edge left},
  width=0.275\textwidth,height=3.8cm,scale only axis,
  xmode=normal,ymode=normal,scaled ticks=false,
  xmin=0,xmax=10500,ymin=0,ymax=47,
  xtick={0,5000,10000},xticklabels={0,5000,10000},
  ytick={0,10,20,30,40},
  grid=major,grid style={black!10},
  axis line style={black!55},tick style={black!55},
  tick label style={font=\fontsize{8}{9.5}\selectfont},
  label style={font=\fontsize{9}{10.5}\selectfont},
  title style={font=\fontsize{9}{10.5}\selectfont},
  xlabel={Horizon $T$},ylabel={Expected regret},
  legend style={font=\fontsize{8.5}{10}\selectfont,draw=none,fill=none,
    at={(0.5,-0.35)},anchor=north,legend columns=-1,
    /tikz/every even column/.append style={column sep=8pt}},
  legend cell align=left
]
\nextgroupplot
\addplot+[num series,color=numDP,mark=*,forget plot,num errors]
  table[x=T,y=mean,y error minus=errlow,y error plus=errhigh]
  {figures/numerics/algorithm_zero_beta_0_dp.dat};
\addplot+[num series,color=numSTARS,mark=square*,forget plot,num errors]
  table[x=T,y=mean,y error minus=errlow,y error plus=errhigh]
  {figures/numerics/algorithm_zero_beta_0_stars.dat};
\addplot+[num series,color=numCE,mark=triangle*,forget plot,num errors]
  table[x=T,y=mean,y error minus=errlow,y error plus=errhigh]
  {figures/numerics/algorithm_zero_beta_0_ce.dat};
\nextgroupplot
\addplot+[num series,color=numDP,mark=*,num errors]
  table[x=T,y=mean,y error minus=errlow,y error plus=errhigh]
  {figures/numerics/algorithm_zero_beta_1_dp.dat};
\addlegendentry{\textsf{Dynamic Programming (DP)}}
\addplot+[num series,color=numSTARS,mark=square*,num errors]
  table[x=T,y=mean,y error minus=errlow,y error plus=errhigh]
  {figures/numerics/algorithm_zero_beta_1_stars.dat};
\addlegendentry{\textsf{STARS} ($K=50$)}
\addplot+[num series,color=numCE,mark=triangle*,num errors]
  table[x=T,y=mean,y error minus=errlow,y error plus=errhigh]
  {figures/numerics/algorithm_zero_beta_1_ce.dat};
\addlegendentry{\textsf{CE}}
\nextgroupplot
\addplot+[num series,color=numDP,mark=*,forget plot,num errors]
  table[x=T,y=mean,y error minus=errlow,y error plus=errhigh]
  {figures/numerics/algorithm_zero_beta_2_dp.dat};
\addplot+[num series,color=numSTARS,mark=square*,forget plot,num errors]
  table[x=T,y=mean,y error minus=errlow,y error plus=errhigh]
  {figures/numerics/algorithm_zero_beta_2_stars.dat};
\addplot+[num series,color=numCE,mark=triangle*,forget plot,num errors]
  table[x=T,y=mean,y error minus=errlow,y error plus=errhigh]
  {figures/numerics/algorithm_zero_beta_2_ce.dat};
\end{groupplot}
\def\numHeaderZeroTitle{(a) $\beta=0$}
\def\numHeaderOneTitle{(b) $\beta=1$}
\def\numHeaderTwoTitle{(c) $\beta=2$}
\input{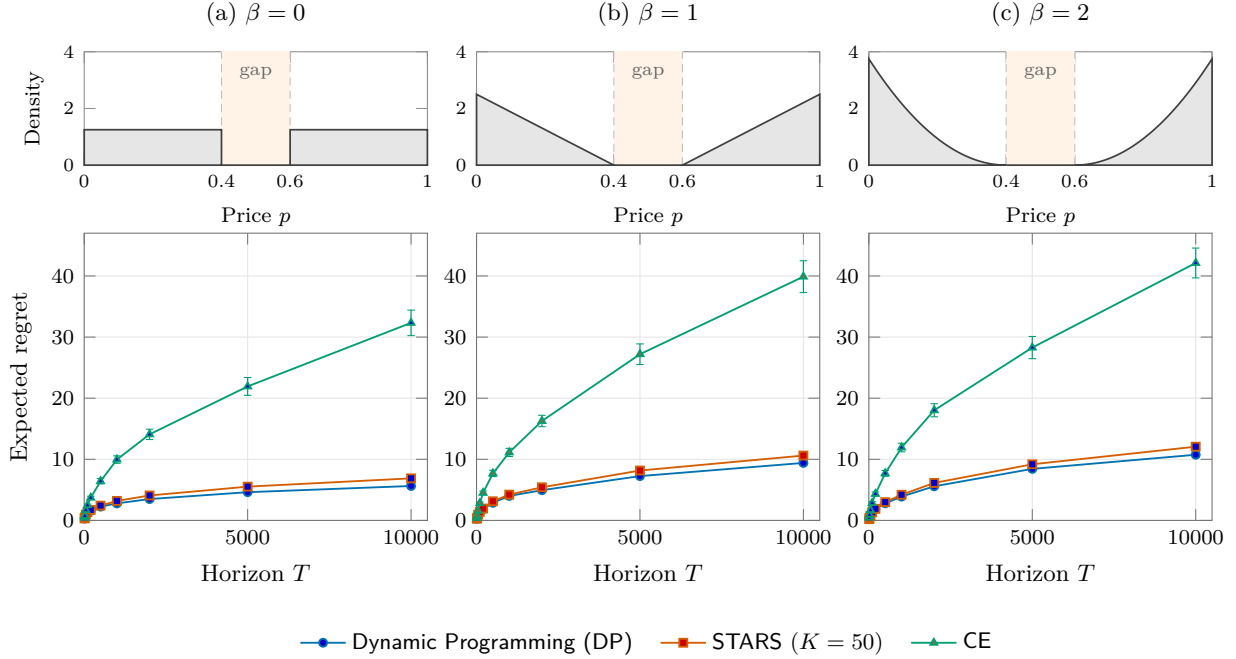}
\end{tikzpicture}
\endgroup
    \caption{Policy regret from zero initial inventory, $y_0=y_T=0$, for
    $F_{\beta,0.2}$ with $B=T/2$.}
    \label{fig:numerics-zero}
\end{figure}

We next repeat the comparison with $y_0=y_T=\lfloor B/2\rfloor$, keeping
the distributions, horizons, and capacity unchanged (Figure~\ref{fig:numerics-interior}). Comparing Figures~\ref{fig:numerics-zero} and~\ref{fig:numerics-interior}
shows that the same positive-gap distribution can be substantially easier
when the initial inventory provides slack from both boundaries. This is
consistent with the distinction between the boundary lower bound in
Theorem~\ref{lb:thm-positive-gap} and the achievable interior rates in
Proposition~\ref{prop:interior-transfer}. Taken together, the experiments support the importance of
initial inventory and show that scenario averaging can perform close to
optimal online control with a modest, distribution-independent number
of scenarios over the tested range.

\begin{figure}[htbp]
    \centering
    % Native PGFPlots input. Source: full-20260924-v1, 200 paired paths per setting.
\begingroup
\definecolor{numDP}{HTML}{0072B2}
\definecolor{numSTARS}{HTML}{D55E00}
\definecolor{numCE}{HTML}{009E73}
\pgfplotsset{
  num series/.style={line width=0.7pt,mark size=1.5pt},
  num errors/.style={error bars/.cd,y dir=both,y explicit,
    error bar style={line width=0.45pt},
    error mark options={rotate=90,mark size=1.4pt}}
}
\begin{tikzpicture}
\begin{groupplot}[
  group style={group size=3 by 1,horizontal sep=0.65cm,ylabels at=edge left},
  width=0.275\textwidth,height=3.8cm,scale only axis,
  xmode=normal,ymode=normal,scaled ticks=false,
  xmin=0,xmax=10500,ymin=0,ymax=14,
  xtick={0,5000,10000},xticklabels={0,5000,10000},
  ytick={0,3,6,9,12},
  grid=major,grid style={black!10},
  axis line style={black!55},tick style={black!55},
  tick label style={font=\fontsize{8}{9.5}\selectfont},
  label style={font=\fontsize{9}{10.5}\selectfont},
  title style={font=\fontsize{9}{10.5}\selectfont},
  xlabel={Horizon $T$},ylabel={Expected regret},
  legend style={font=\fontsize{8.5}{10}\selectfont,draw=none,fill=none,
    at={(0.5,-0.35)},anchor=north,legend columns=-1,
    /tikz/every even column/.append style={column sep=8pt}},
  legend cell align=left
]
\nextgroupplot
\addplot+[num series,color=numDP,mark=*,forget plot,num errors]
  table[x=T,y=mean,y error minus=errlow,y error plus=errhigh]
  {figures/numerics/algorithm_interior_beta_0_dp.dat};
\addplot+[num series,color=numSTARS,mark=square*,forget plot,num errors]
  table[x=T,y=mean,y error minus=errlow,y error plus=errhigh]
  {figures/numerics/algorithm_interior_beta_0_stars.dat};
\addplot+[num series,color=numCE,mark=triangle*,forget plot,num errors]
  table[x=T,y=mean,y error minus=errlow,y error plus=errhigh]
  {figures/numerics/algorithm_interior_beta_0_ce.dat};
\nextgroupplot
\addplot+[num series,color=numDP,mark=*,num errors]
  table[x=T,y=mean,y error minus=errlow,y error plus=errhigh]
  {figures/numerics/algorithm_interior_beta_1_dp.dat};
\addlegendentry{\textsf{Dynamic Programming (DP)}}
\addplot+[num series,color=numSTARS,mark=square*,num errors]
  table[x=T,y=mean,y error minus=errlow,y error plus=errhigh]
  {figures/numerics/algorithm_interior_beta_1_stars.dat};
\addlegendentry{\textsf{STARS} ($K=50$)}
\addplot+[num series,color=numCE,mark=triangle*,num errors]
  table[x=T,y=mean,y error minus=errlow,y error plus=errhigh]
  {figures/numerics/algorithm_interior_beta_1_ce.dat};
\addlegendentry{\textsf{CE}}
\nextgroupplot
\addplot+[num series,color=numDP,mark=*,forget plot,num errors]
  table[x=T,y=mean,y error minus=errlow,y error plus=errhigh]
  {figures/numerics/algorithm_interior_beta_2_dp.dat};
\addplot+[num series,color=numSTARS,mark=square*,forget plot,num errors]
  table[x=T,y=mean,y error minus=errlow,y error plus=errhigh]
  {figures/numerics/algorithm_interior_beta_2_stars.dat};
\addplot+[num series,color=numCE,mark=triangle*,forget plot,num errors]
  table[x=T,y=mean,y error minus=errlow,y error plus=errhigh]
  {figures/numerics/algorithm_interior_beta_2_ce.dat};
\end{groupplot}
\def\numHeaderZeroTitle{(a) $\beta=0$}
\def\numHeaderOneTitle{(b) $\beta=1$}
\def\numHeaderTwoTitle{(c) $\beta=2$}
\input{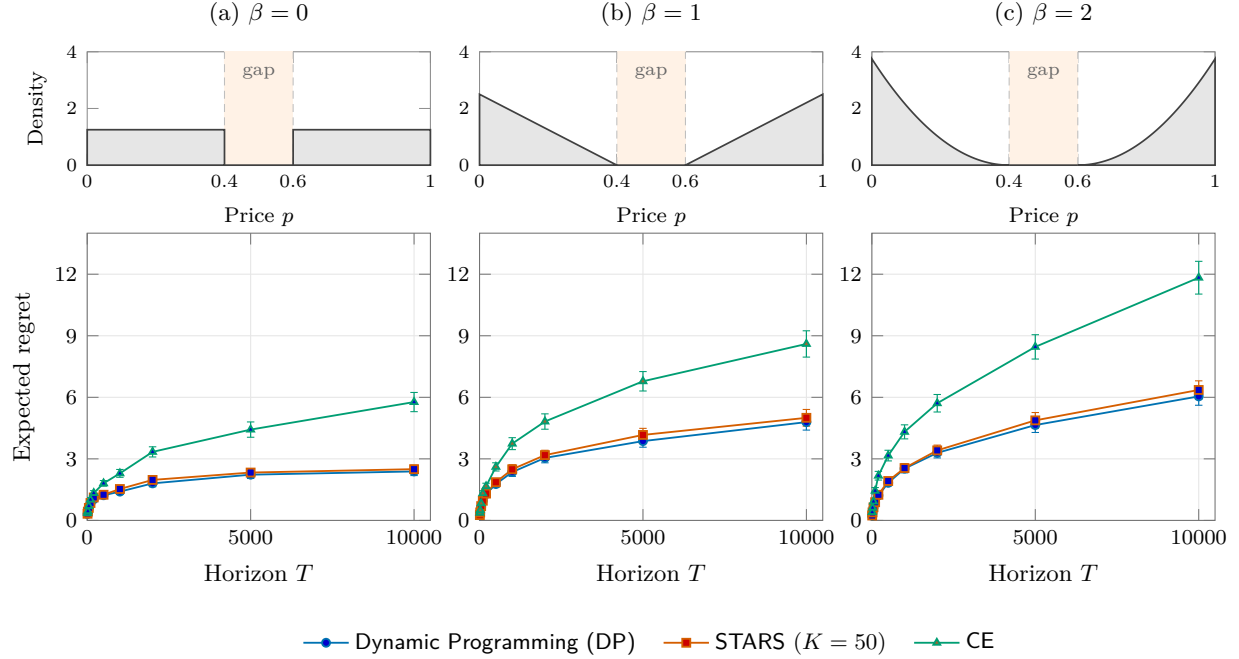}
\end{tikzpicture}
\endgroup
    \caption{Policy regret from an interior state,
    $y_0=y_T=\lfloor B/2\rfloor$, with the same distributions and
    experimental settings as Figure~\ref{fig:numerics-zero}.}
    \label{fig:numerics-interior}
\end{figure}
\section{Conclusion}
\label{sec:conclusion}

In this work, we studied online trading with storage and identified initial
inventory and local price geometry as fundamental drivers of regret.
For every $\beta\ge0$,
% we constructed a bounded-density distribution with
% a median gap for which
{we proved that, for every fixed distribution in the regular
continuous positive-gap family,}
every online policy starting and ending with zero
inventory, with capacity proportional to $T$, incurs
$\Omega(T^{(\beta+1)/(2\beta+4)})$ regret. In contrast, linear
slack from both boundaries allows resource-allocation guarantees to
carry over. Thus, a median gap can create a new polynomial penalty when
the trader must build inventory before exploiting selling opportunities.

We analyzed \STARS{}, which averages optimized continuation values over
simulated future price sequences. Its holding region arises endogenously,
without distribution-specific fine-tuning. With sufficiently many
scenarios, the same policy achieves logarithmic regret for finite atomic
prices and the stated regret spectra for regular continuous prices.
% On the canonical positive-gap instances,
{For every fixed distribution in the regular continuous
positive-gap class,}
these bounds match the lower
bound up to logarithmic factors. The
analysis combines marginal compensation with control of how often the
policy visits each inventory state. Numerically, \STARS{} performs close
to the optimal online dynamic program with only $50$ scenarios and
substantially improves on certainty equivalence in the positive-gap
instances considered.

These results suggest several follow-up questions. First, how many
scenarios are needed for near-optimal regret, and can that number depend
on the current inventory and the difference between candidate action
values? The numerical performance with a small fixed budget motivates
sharper sampling guarantees. Second, how does regret change when the
initial inventory and capacity grow on intermediate scales between zero
initial inventory and linear slack? Finally, it would be useful to study
unknown or temporally dependent prices, transaction costs, and storage
losses, to understand how these features change the value of holding
inventory.

\section*{AI usage disclosure}
GPT-5.6 Sol, Claude Fable 5, and GPT-6 Astra provided excellent research assistance.
They helped develop rigorous proofs for many of our results by formalizing proof details, checking intermediate arguments, and locating relevant concentration inequalities and convergence results. They also helped organize portions of the exposition. 
The authors developed the research question, model, main results and algorithmic ideas without much AI assistance. The authors take responsibility for the manuscript and any remaining errors.

\bibliographystyle{plainnat}
\begingroup
\setlength{\bibsep}{2pt}
\bibliography{references}

@article{mayne2000,
  title={Constrained model predictive control: Stability and optimality},
  author={Mayne, David Q and Rawlings, James B and Rao, Christopher V and Scokaert, Pierre OM},
  journal={Automatica},
  volume={36},
  number={6},
  pages={789--814},
  year={2000},
  publisher={Elsevier}
}

@article{mesbah2016,
  title={Stochastic model predictive control: An overview and perspectives for future research},
  author={Mesbah, Ali},
  journal={IEEE Control Systems Magazine},
  volume={36},
  number={6},
  pages={30--44},
  year={2016},
  publisher={IEEE}
}

@article{besbes2025,
  title={Dynamic resource allocation: Algorithmic design principles and spectrum of achievable performances},
  author={Besbes, Omar and Kanoria, Yash and Kumar, Akshit},
  journal={Operations Research},
  volume={73},
  number={3},
  pages={1273--1288},
  year={2025},
  publisher={INFORMS}
}

@article{vera2021bayesian,
  title={The bayesian prophet: A low-regret framework for online decision making},
  author={Vera, Alberto and Banerjee, Siddhartha},
  journal={Management Science},
  volume={67},
  number={3},
  pages={1368--1391},
  year={2021},
  publisher={INFORMS}
}

@article{bray2025,
  title={Logarithmic regret in multisecretary and online linear programs with continuous valuations},
  author={Bray, Robert L},
  journal={Operations Research},
  volume={73},
  number={4},
  pages={2188--2203},
  year={2025},
  publisher={INFORMS}
}

@article{arlotto2019uniformly,
  title={Uniformly bounded regret in the multisecretary problem},
  author={Arlotto, Alessandro and Gurvich, Itai},
  journal={Stochastic Systems},
  volume={9},
  number={3},
  pages={231--260},
  year={2019},
  publisher={INFORMS}
}

@article{bumpensanti2020re,
  title={A re-solving heuristic with uniformly bounded loss for network revenue management},
  author={Bumpensanti, Pornpawee and Wang, He},
  journal={Management Science},
  volume={66},
  number={7},
  pages={2993--3009},
  year={2020},
  publisher={INFORMS}
}

@article{balseiro2024survey,
  title={Survey of dynamic resource-constrained reward collection problems: Unified model and analysis},
  author={Balseiro, Santiago R and Besbes, Omar and Pizarro, Dana},
  journal={Operations Research},
  volume={72},
  number={5},
  pages={2168--2189},
  year={2024},
  publisher={INFORMS}
}

@article{talluri1999randomized,
  title={A randomized linear programming method for computing network bid prices},
  author={Talluri, Kalyan and Van Ryzin, Garrett},
  journal={Transportation science},
  volume={33},
  number={2},
  pages={207--216},
  year={1999},
  publisher={INFORMS}
}

@inproceedings{sinclair2023hindsight,
  title={Hindsight learning for mdps with exogenous inputs},
  author={Sinclair, Sean R and Frujeri, Felipe Vieira and Cheng, Ching-An and Marshall, Luke and Barbalho, Hugo De Oliveira and Li, Jingling and Neville, Jennifer and Menache, Ishai and Swaminathan, Adith},
  booktitle={International Conference on Machine Learning},
  pages={31877--31914},
  year={2023},
  organization={PMLR}
}

@article{chen2026feature,
  title={Feature-based dynamic matching},
  author={Chen, Yilun and Kanoria, Yash and Kumar, Akshit and Zhang, Wenxin},
  journal={Operations Research},
  volume={74},
  number={2},
  pages={788--803},
  year={2026},
  publisher={INFORMS}
}

@article{kumar2026dark,
  title={Feature-based Dynamic Matching in the Dark},
  author={Kumar, Akshit},
  journal={Available at SSRN 6839679},
  year={2026}
}

@inproceedings{correa2026trading,
  title={Trading prophets},
  author={Correa, Jos{\'e} and Cristi, Andr{\'e}s and D{\"u}tting, Paul and Hajiaghayi, MohammadTaghi and Olkowski, Jan and Schewior, Kevin},
  booktitle={Proceedings of the 24th ACM Conference on Economics and Computation},
  pages={490--510},
  year={2023}
}

@article{aggarwal2026asymmetric,
  title={Asymmetric Trading Prophets},
  author={Aggarwal, Gagan and Gupta, Anupam and Wang, Yifan and Zhao, Mingfei},
  journal={arXiv preprint arXiv:2607.01516},
  year={2026}
}

@inproceedings{azar2026initial,
  title={Trading prophets with initial capital},
  author={Azar, Yossi and Buchbinder, Niv and Levin, Roie and Vardi, Or},
  booktitle={2026 SIAM Symposium on Simplicity in Algorithms (SOSA)},
  pages={81--90},
  year={2026},
  organization={SIAM}
}

@inproceedings{kleinberg2005multiple,
  title={A multiple-choice secretary algorithm with applications to online auctions.},
  author={Kleinberg, Robert D},
  booktitle={SODA},
  volume={5},
  pages={630--631},
  year={2005}
}

@article{glynn1996liapounov,
  author  = {Glynn, Peter W. and Meyn, Sean P.},
  title   = {A {Liapounov} bound for solutions of the {Poisson} equation},
  journal = {The Annals of Probability},
  year    = {1996},
  volume  = {24},
  number  = {2},
  pages   = {916--931},
}

@article{jasin2012re,
  title={A re-solving heuristic with bounded revenue loss for network revenue management with customer choice},
  author={Jasin, Stefanus and Kumar, Sunil},
  journal={Mathematics of Operations Research},
  volume={37},
  number={2},
  pages={313--345},
  year={2012},
  publisher={INFORMS}
}
\endgroup
\clearpage
\appendix
% Distinguish appendix equations from main-text equations.
\setcounter{equation}{0}
\renewcommand{\theequation}{App.\arabic{equation}}
\renewcommand{\theHequation}{appendix.\arabic{equation}}
\etocdepthtag.toc{appendix}
\begingroup
\renewcommand{\contentsname}{Appendix contents}
\etocsettagdepth{main}{none}
\etocsettagdepth{appendix}{subsection}
\etocsetnexttocdepth{subsection}
\pdfbookmark[0]{Appendix contents}{appendix.contents}
\tableofcontents
\endgroup
\clearpage
\section{Proof of Proposition \ref{prop:interior-transfer}}
\label{app:interior-transfer}

We prove Proposition~\ref{prop:interior-transfer} by converting a
multisecretary policy into a trading policy. Selecting a price corresponds
to selling; rejecting it corresponds to buying. We show that the
resulting inventory stays close to its initial level with high
probability, and modify the policy on the remaining paths to preserve
feasibility. We then prove the selection-regret bounds using the local
median assumptions. For a positive median gap, the key step separates
each price into a two-price component and a within-side remainder.

\paragraph{From selections to balanced trades.}
Begin with an even horizon $n\ge2$ and a multisecretary quota $k=n/2$.
Let $a_t\in\{0,1\}$ denote selection, $b_t=k-\sum_{s<t}a_s$ the remaining
quota, and $\tau_t=n-t+1$. We specify the selection threshold below.
Every version uses two completion rules: reject when $b_t=0$, and
accept all remaining prices when $b_t=\tau_t$. These rules preserve
$0\le b_t\le\tau_t$, because an interior integer quota allows either
decision. In particular, $b_{n+1}=0$.

Define virtual trades $x_t=1-2a_t$ and displacement
$D_j=\sum_{t=1}^j x_t$, with $D_0=0$. Selection means selling and
rejection means buying. The identity
$D_{t-1}=2b_t-\tau_t$ gives both $|D_{t-1}|\le\tau_t$ and $D_n=0$.
Writing $S_k(\mathbf p)$ for the sum of the $k$ largest prices, we have
\begin{align*}
 R_n^{\rm virtual}
 &=2\sum_{t=1}^n p_ta_t-\sum_{t=1}^n p_t,
 &V_n^{\rm stat}
 &=2S_k(\mathbf p)-\sum_{t=1}^n p_t.
\end{align*}
Here $V_n^{\rm stat}$ removes only intermediate inventory bounds.
Balance requires equally many purchases and sales; exchange assigns
them to the smallest and largest prices, respectively. Buying the
smaller and selling the larger of any two held prices cannot reduce
profit, so an optimum uses $k$ of each. The identities give static
regret exactly twice the completed policy's multisecretary regret.

\paragraph{Specifying the quantile rule.}
We work in quantile space to treat continuous prices and atoms together. Use fresh
independent $W_t\sim\operatorname{Unif}[0,1]$ and set
$U_t=F(p_t^-)+W_t(F(p_t)-F(p_t^-))$. Then the $U_t$ are i.i.d.\ uniform
and independent of $\mathcal F_{t-1}$. Indeed, drawing a uniform $U_t$
first and setting $p_t=F^{-1}(U_t)$ has exactly this joint law: conditional
on an atom, the uniform draw lies uniformly within that atom's quantile
interval, whereas at non-atoms it equals $F(p_t)$ almost surely.
Include $W_t$ in the policy randomization defining $\mathcal F_t$.

Let $q_t$ be the pre-decision rejection quantile,
so that $a_t=\one\{U_t\ge q_t\}$ almost surely. At quota endpoints use
$q_t=1$ when $b_t=0$ and $q_t=0$ when $b_t=\tau_t$.
Write $q_t^{\textsf{CE}}=1-b_t/\tau_t$ and
$r_\tau=\sqrt{2\log(\tau)/\tau}$. At interior quota states use the following rules.
\begin{enumerate}[label=$(\alph*)$]
\item For finite atomic prices, use Algorithm~1 of \cite{besbes2025}
with all atom-boundary quantiles designated as gaps and
$\varepsilon_*=\min_i w_i$. Its main phase snaps the \textsf{CE}
quantile to a nearest designated gap within $r_{\tau_t}$, when one
exists, and otherwise uses the \textsf{CE} quantile. The final
$n\wedge L_F$ steps freeze the \textsf{CE} quantile at the start of
that phase, where
$L_F=\lfloor64\log(1/\varepsilon_*)/\varepsilon_*^2\rfloor$.
A one-point law has zero trading regret and can be handled by holding.
\item For continuous prices with $\Delta=0$, use
$q_t=q_t^{\textsf{CE}}$ throughout, and set $L_F=0$.
\item For continuous prices with $\Delta>0$, designate only the
median quantile $1/2$. Set $L_F=\lfloor256\log2\rfloor$, the cutoff
in Algorithm~1 of \cite{besbes2025} for two clusters of mass $1/2$.
When $\tau_t>L_F$, set $q_t=1/2$ if
$|q_t^{\textsf{CE}}-1/2|\le r_{\tau_t}$, and otherwise set
$q_t=q_t^{\textsf{CE}}$. In the final $n\wedge L_F$ steps, freeze
the \textsf{CE} quantile at that phase's start.
\end{enumerate}
The endpoint completion rules override every frozen or adjusted
threshold. Thus each policy selects exactly $k=n/2$ prices. This
specification allows arbitrary additional gaps away from the median
in the continuous families.

\paragraph{The online inventory has sublinear fluctuations.}
Completion preserves the main-phase bound, since its endpoint
quantiles equal $q_t^{\textsf{CE}}$. In the final phase both quantiles lie
in $[0,1]$. Consequently,
\begin{align*}
 \E[x_t\mid\mathcal F_{t-1}]
 &\stackrel{(a)}{=}2q_t-1
 \stackrel{(b)}{=}-\frac{D_{t-1}}{\tau_t}+e_t,
 & |e_t|&\le
 \begin{cases}2r_{\tau_t},&\tau_t>L_F,\\2,&\tau_t\le L_F.
 \end{cases}
\end{align*}
Step (a) uses the conditional uniformity of $U_t$; step (b) uses
$D_{t-1}=2b_t-\tau_t$ and defines $e_t=2(q_t-q_t^{\textsf{CE}})$.
Thus the fluid part of the drift points toward zero. The total possible
outward perturbation satisfies
\begin{align*}
 \sum_{t=1}^n|e_t|
 &\le A_n:=2\sum_{\tau=L_F+1}^{n}
           \sqrt{\frac{2\log\tau}{\tau}}+2L_F
 \stackrel{(a)}{\le}4\sqrt{2n\log n}+2L_F=o(n),
\end{align*}
with an empty sum interpreted as zero. Step (a) uses
$\sum_{\tau=1}^n\tau^{-1/2}\le2\sqrt n$.

The restoring drift lets us control the largest inventory displacement.
Put $\xi_t=x_t-\E[x_t\mid\mathcal F_{t-1}]$ and
$M_j=\sum_{t=1}^j\xi_t$. This is a martingale with increments bounded
in absolute value by two. If $D_j>0$, let $\ell<j$ be its last zero.
Since $D$ starts at zero and has unit increments, it is positive between
$\ell$ and $j$. Therefore
\begin{align*}
 D_j
 &=\sum_{t=\ell+1}^j
       \left(-\frac{D_{t-1}}{\tau_t}+e_t+\xi_t\right)
 \stackrel{(a)}{\le} A_n+M_j-M_\ell
 \le A_n+2\max_{i\le n}|M_i|.
\end{align*}
Step (a) drops the nonpositive restoring drift and bounds the sum of
perturbations by $A_n$. For a negative excursion, apply the same argument
to $-D$. Thus $\max_j|D_j|\le A_n+2\max_j|M_j|$ on every realization.

The bounded-increment exponential martingale gives
$\Pp(\max_{j\le n}|M_j|\ge z)\le2\exp(-z^2/(8n))$.
Indeed, conditional Hoeffding gives
$\E[e^{\theta\xi_t}\mid\mathcal F_{t-1}]\le e^{2\theta^2}$.
Stopping the resulting nonnegative supermartingale at the first
crossing of $z$, or at $n$, bounds the upper-tail probability by
$\exp(-\theta z+2n\theta^2)$. Set $\theta=z/(4n)$ and repeat for $-M$.
For fixed $F,\eta$, take $n$ sufficiently large that $A_n\le\eta n/2$.
It follows that
\begin{align*}
 \Pp\left(\max_{j\le n}|D_j|\ge\eta n\right)
 &\le\Pp\left(\max_{j\le n}|M_j|\ge\eta n/4\right)
 \le2e^{-\eta^2n/128}.
\end{align*}
Thus the adaptive online trajectory stays within a linear distance of
its initial inventory with exponentially high probability.

\paragraph{Feasibility and transfer of regret.}
Now suppose $y_0\wedge(B-y_0)\ge\eta n$. Execute the virtual trades
until a proposed trade first leaves $[0,B]$. At that time step, instead
move one unit toward $y_0$, continue moving toward it, and hold once it
is reached. The rule uses only observed information. Before a switch,
the actual displacement equals $D_{t-1}$; the inequality
$|D_{t-1}|\le\tau_t$ proves that the remaining time suffices for the
return. Each return step remains within $[0,B]$. If no switch occurs,
$D_n=0$ already gives terminal balance. The policy is therefore feasible
on \emph{every} price sequence.

A switch is possible only if the virtual path has
$\max_j|D_j|\ge\eta n$. The virtual and feasible policies agree otherwise,
and their profits differ by at most $2n$ on any path. Let
$\Reg_{\rm MS}(n)$ denote the specified completed selection
policy's expected regret against the sum of the $n/2$ largest prices.
For all sufficiently large even $n$,
\begin{align*}
 \Reg(\pi^{\textsf{CwG}};B,n,y_0)
 &\stackrel{(a)}{\le}
       \E[V_n^{\rm stat}-R_n^{\rm virtual}]
       +2n\Pp(\text{switch}) \stackrel{(b)}{\le}
       2\Reg_{\rm MS}(n)
       +4ne^{-\eta^2n/128}.
\end{align*}
Step (a) uses storage hindsight $\le V_n^{\rm stat}$ and the pathwise
profit comparison. Step (b) uses the two reward identities
and the maximal bound. The storage
constraints therefore add only an exponentially small error to the
multisecretary regret comparison.

For odd $T\ge3$, use $n=T-1$ and hold in the final time step. Pointwise,
$V_T^{\rm stat}\le V_n^{\rm stat}+1$. To see this, remove the final
action from a balanced $T$-step optimizer. If it is a hold, the remaining
path is balanced. Otherwise, also change one opposite-sign action among
the first $n$ time steps to a hold; such an action exists by balance.
The two removed opposite-sign rewards sum to a difference of two prices,
which is at most one. The resulting $n$-step static path is feasible.
This proves the claimed inequality, so odd horizons add at most one to
the transfer bound. For $T=1$ every feasible policy holds. The finitely
many smaller horizons excluded by $A_n\le\eta n/2$ are absorbed into
constants depending on $F,\eta$.

\paragraph{Finite atomic prices.}
Each atom is a cluster whose quantile function is constant, and
$\varepsilon_*=\min_i w_i>0$. The discrete \textsf{CwG} corollary
in \cite{besbes2025} gives the original selection policy regret
$O_F(1)$. Completion preserves this bound. Indeed, couple the original
and completed policies on the same quantiles. They agree until
completion first forces a selection; thereafter the completed policy
selects every remaining price, including all original selections.
Nonnegative prices make its reward weakly larger on every realization.
Thus $\Reg_{\rm MS}(n)=O_F(1)$, proving part (a) after transfer.

\paragraph{Separating the median gap from the within-side prices.}
For the rest of the proof, let
$F\in\mathcal F^{\rm ac}_{\beta,\Delta,\varepsilon_0}$ and put
$\alpha=1/(\beta+1)$. Define
$h(u)=F^{-1}(u)-a^--\Delta\one\{u>1/2\}$ for
$u\in(0,1)\setminus\{1/2\}$.
The remainder $h$ is nonpositive below the median and nonnegative
above it. Moreover, there is a finite constant $C_F$ such that
$|h(u)|\le C_F|u-1/2|^\alpha$ for all $u\ne1/2$.
For $|u-1/2|<\varepsilon_0$, this follows from
Definition~\ref{def:continuous-family} and $|r_\pm|\le1/2$.
Outside that window, $|h(u)|\le1$ and
$|u-1/2|^\alpha\ge\varepsilon_0^\alpha$.
Thus one can take
$C_F=\max\{3c_-/2,3c_+/2,\varepsilon_0^{-\alpha}\}$.
All assertions at $u=1/2$ may be defined arbitrarily, since this event
has probability zero.

This decomposition separates two kinds of selection error. Choosing
the wrong side of the median costs the fixed gap $\Delta$. Within
either side, the error depends on the distance of the price from its
median edge. To formalize this, let $z_t=\one\{U_t>1/2\}$ and let
$I^*$ contain the indices of the $k=n/2$ largest quantiles. These
indices maximize both the sum of the actual prices and the sum of
the binary prices $z_t$. Since $\sum_ta_t=k$, the constant $a^-$
cancels, giving
\begin{align*}
 S_k(\mathbf p)-\sum_ta_tp_t
 &\stackrel{(a)}{=}
 \Delta\left(S_k(\mathbf z)-\sum_ta_tz_t\right)
 +\sum_{t\in I^*}h(U_t)-\sum_ta_th(U_t)\\
 &\stackrel{(b)}{\le}
 \Delta\left(S_k(\mathbf z)-\sum_ta_tz_t\right)
 +\sum_t|h(U_t)|\one\{a_t\ne z_t\}.
\end{align*}
Step (a) substitutes the price decomposition. For (b), bound
$\sum_{t\in I^*}h(U_t)$ by $\sum_th(U_t)^+$ and use the sign of $h$:
the difference $h(U_t)^+-a_th(U_t)$ equals the displayed
disagreement charge.

Conditional on $\mathcal F_{t-1}$, the threshold $q_t$ is fixed
and $U_t$ is uniform. Disagreement with the median decision occurs
exactly between $q_t$ and $1/2$, up to null endpoints. Therefore
\begin{align*}
 \E\!\left[|h(U_t)|\one\{a_t\ne z_t\}\mid\mathcal F_{t-1}\right]
 &=\int_{q_t\wedge(1/2)}^{q_t\vee(1/2)}|h(u)|\,du \stackrel{(a)}{\le}
 \frac{C_F}{1+\alpha}|q_t-1/2|^{1+\alpha}.
\end{align*}
Step (a) integrates the bound on $h$. We next control these quantile
deviations and, when $\Delta>0$, the binary selection loss.

\paragraph{Continuous prices without a median gap.}
Here $\Delta=0$, so the binary term vanishes. Under the \textsf{CE}
rule, $e_t=0$ and $|q_t-1/2|=|D_{t-1}|/(2\tau_t)$.
Write $s_\tau=\E[D_{n-\tau}^2]$ for the second moment with $\tau$ steps
remaining. Since $x_t^2=1$, conditioning on the history gives, for
$\tau\ge2$,
\begin{align*}
 s_{\tau-1}
 &\stackrel{(a)}{=}(1-2/\tau)s_\tau+1
 \stackrel{(b)}{\le}\tau-1.
\end{align*}
Step (a) expands $(D_{t-1}+x_t)^2$ and uses drift $-D_{t-1}/\tau$.
Step (b) proves $s_\tau\le \tau$ by downward induction starting at
$s_n=0$; its coefficient is nonnegative for $\tau\ge2$.
Because $1+\alpha\le2$, Jensen's inequality now yields
$\E|q_t-1/2|^{1+\alpha}
\le 2^{-(1+\alpha)}s_{\tau_t}^{(1+\alpha)/2}/\tau_t^{1+\alpha}
\le 2^{-(1+\alpha)}\tau_t^{-(1+\alpha)/2}$.
Summing the conditional disagreement charges proves
$\Reg_{\rm MS}(n)=O_F(\log n)$ when $\alpha=1$, and
$O_F(n^{(1-\alpha)/2})$ when $\alpha<1$. Indeed, the sum is harmonic
in the first case, and at most
$1+\int_1^n u^{-(1+\alpha)/2}\,du$ in the second.
Since $(1-\alpha)/2=\beta/(2\beta+2)$, transfer gives part (b).

\paragraph{Continuous prices with a positive median gap.}
First consider the binary loss. The variables $z_t$ are i.i.d.\
Bernoulli$(1/2)$. Conditional on $z_t$, the quantile $U_t$ is uniform
on the corresponding half of $[0,1]$. These are precisely the
randomized within-atom quantiles in \textsf{CwG} applied to the two
prices $\{0,1\}$, each with probability $1/2$. Their only gap quantile
is $1/2$. Our median-only rule has the same thresholds, quotas, and
final static phase as this binary policy, with the same completion
rules. The coupling is therefore exact. The finite-type guarantee
of \cite{besbes2025}, together with completion's pathwise reward
improvement proved above, gives
$\E[S_k(\mathbf z)-\sum_ta_tz_t]=O(1)$ uniformly in $n$.

It remains to bound the within-side charges. In the main phase
$\tau=\tau_t>L_F$, the drift perturbation satisfies $|e_t|\le2r_\tau$.
Expanding the squared displacement and using Young's inequality gives
\begin{align*}
 s_{\tau-1}
 &=(1-2/\tau)s_\tau+2\E[D_{t-1}e_t]+1\\
 &\stackrel{(a)}{\le}
 (1-3/(2\tau))s_\tau+1+2\tau\E[e_t^2]\stackrel{(b)}{\le}
 (1-3/(2\tau))s_\tau+1+16\log \tau.
\end{align*}
Step (a) uses $2|De|\le D^2/(2\tau)+2\tau e^2$; step (b) uses
$e_t^2\le8\log(\tau)/\tau$. Downward induction from $s_n=0$ shows
$s_\tau\le36\tau\log n$ throughout the main phase. In fact, substituting
this bound on the right gives
$36(\tau-3/2)\log n+1+16\log \tau
\le36(\tau-1)\log n$, since $1\le2\log n$ for $n\ge2$.
Snapping to the median can only reduce the threshold's distance from
$1/2$, so $|q_t-1/2|\le|D_{t-1}|/(2\tau)$. Another application of
Jensen gives
$\E|q_t-1/2|^{1+\alpha}
\le C_\alpha(\log n/\tau)^{(1+\alpha)/2}$.
The final phase has at most the fixed $L_F$ steps, each contributing
at most one to the residual charge, so it adds only $O(1)$.

The binary loss is also $O(1)$. Summing the main-phase charges thus
gives $\Reg_{\rm MS}(n)=O_F(\log^2 n)$ when $\alpha=1$, and
$O_F(n^{(1-\alpha)/2}(\log n)^{(1+\alpha)/2})$ when $\alpha<1$.
Use $(1+\alpha)/2=(\beta+2)/(2\beta+2)$ and the transfer bound to
obtain part (c). The continuous argument uses only the local
upper-quantile envelope and bounded support; all distributions in
Definition~\ref{def:continuous-family} satisfy these requirements.

\section{Proof of the positive-gap boundary lower bound}
\label{lb:app-positive-gap}

We prove Theorem~\ref{lb:thm-positive-gap} by comparing the cost of
keeping inventory feasible online with the cost of doing so in hindsight.
The key distinction is that hindsight can choose inexpensive opportunities
to hold before inventory runs out, whereas an online policy must make
these choices without knowing which opportunities will be needed.
Fix $\beta\ge0$, $\varepsilon_0\in(0,1/4]$, $\Delta\in(0,1)$, and
$\rho\in(0,1)$. Let
$F\in\mathcal F^{\mathrm{ac}}_{\beta,\Delta,\varepsilon_0}$ be fixed
independently of $T$, and write $\kappa:=\beta+1$ and
$a^-:=a_F^-$, $a^+:=a_F^+$. Prices are i.i.d.\ with distribution $F$.
Absolute continuity and the median-gap definition imply
$F(a^-)=F(a^+)=1/2$ and $a^+-a^-=\Delta$.
Thus each side of the gap has probability $1/2$, even when the two
conditional price distributions differ. The local quantile bounds below
replace the exact power-law formulas of the canonical examples.

A possibly randomized online policy is represented by independent fresh
seeds $(\zeta_t)_{t=1}^T$, independent of the entire price sequence.  Put
$\mathcal F_0:=\{\varnothing,\Omega\}$ and
$\mathcal F_t:=\sigma(p_s,\zeta_s:s\leq t)$.  The action $x_t$ is
$\mathcal F_t$-measurable and takes values in $\{-1,0,1\}$. Almost surely,
it satisfies $y_t=y_{t-1}+x_t$, $0\leq y_t\leq B:=\lceil\rho T\rceil$,
and $y_0=y_T=0$. All conditional statements below include the seed in the
filtration.  Consequently, future prices and their median-side signs remain
independent of $\mathcal F_t$.

\subsection{Marginal compensation and departures from the static median-sign rule}

We first express regret in terms of departures from the static median-sign rule.
Buying at every low price and selling at every high price is generally
infeasible when the initial inventory level is zero. Any feasible policy must modify some of
these actions. The following accounting separates the fixed cost imposed
by the median gap from the additional cost of the particular prices at
which these modifications occur.

Let $V_t^{\mathsf{hs}}(y;\mathbf p_{\geq t})$ be the storage-constrained
hindsight value from periods $t,\ldots,T$, starting with inventory $y$
immediately before period $t$ and ending at zero.  Set
$V_{T+1}^{\mathsf{hs}}(0;\varnothing)=0$.  For a feasible current action
$x$, define its marginal compensation by
\begin{align}
    \partial R_t(y,x;\mathbf p_{\geq t})
    :=V_t^{\mathsf{hs}}(y;\mathbf p_{\geq t})
      -\left[-p_tx+V_{t+1}^{\mathsf{hs}}
          (y+x;\mathbf p_{>t})\right].
    \label{lb:eq-partial-regret}
\end{align}
It is nonnegative because the hindsight value maximizes over the feasible
current actions.

Let $\xi_t$ be the action proposed by the static median-sign rule:
$\xi_t:=+1$ when $p_t\leq a^-$ and $\xi_t:=-1$ when $p_t\geq a^+$.
The sum of its proposed trades is
$S_t:=\sum_{s=1}^t\xi_s$, with $S_0:=0$.
The gap has no mass, so this defines $\xi_t$ almost surely, and the
variables $(\xi_t)$ are i.i.d. Rademacher.  Define the globally nonnegative
edge premiums $W_t:=(a^--p_t)\one\{\xi_t=+1\}$ and
$Z_t:=(p_t-a^+)\one\{\xi_t=-1\}$.
Thus $0\leq W_t,Z_t\leq1$, and on the corresponding side
$p_t=a^--W_t$ or $p_t=a^++Z_t$.

For a feasible policy $\pi$, let
\begin{align}
    u_t&:=\one\{\xi_t=-1\}(1+x_t^\pi),
    &U_t&:=\sum_{s=1}^tu_s,
    \notag\\
    \ell_t&:=\one\{\xi_t=+1\}(1-x_t^\pi),
    &L_t&:=\sum_{s=1}^t\ell_s.
    \label{lb:eq-corrections}
\end{align}
Both $u_t$ and $\ell_t$ belong to $\{0,1,2\}$.  They count deviations,
in units, from buying on the lower side and selling on the upper side.

\begin{lemma}
\label{lb:lem-accounting}
Let $\pi$ be a feasible cyclic policy and $\pi^{\mathsf{hs}}$ a
hindsight-optimal feasible policy. Put
$M_t:=\max_{0\leq s\leq t}(-S_s)^+$ and define the correction cost
\begin{align}
    C_T(\pi)
    :=\sum_{t:\xi_t=+1}W_t\ell_t
      +\sum_{t:\xi_t=-1}(\Delta+Z_t)u_t.
    \label{lb:eq-correction-cost}
\end{align}
The following statements hold pathwise.
\begin{enumerate}[label=$(\alph*)$]
\item The inventory satisfies
\begin{align}
    y_t=S_t+U_t-L_t.
    \label{lb:eq-state-ledger}
\end{align}
\item The upper-side correction count satisfies $U_T\geq M_T$.
\item The marginal compensations satisfy
\begin{align}
    \sum_{t=1}^T
       \partial R_t(y_{t-1}^\pi,x_t^\pi;\mathbf p_{\geq t})
    =R_T(\pi^{\mathsf{hs}};0)-R_T(\pi;0)
    =C_T(\pi)-C_T(\pi^{\mathsf{hs}}).
    \label{lb:eq-two-accountings}
\end{align}
\end{enumerate}
\end{lemma}

\begin{proof}
At a lower-side price, $x_t^\pi=\xi_t-\ell_t$, while at an
upper-side price, $x_t^\pi=\xi_t+u_t$.  Hence
\begin{align*}
    y_t
    &\stackrel{(a)}{=}\sum_{s=1}^t x_s^\pi
     \stackrel{(b)}{=}S_t+U_t-L_t,
\end{align*}
where $(a)$ uses $y_0=0$, and $(b)$ uses
\eqref{lb:eq-corrections}.  Since $y_t\geq0$ and $L_t\geq0$,
$U_t\geq-S_t$.  The process $U$ is nondecreasing, so maximizing over
$t$ proves the mandatory-reflection claim.

Let $R_T^\circ:=\sum_{t=1}^T-p_t\xi_t$ be the reward of the
unconstrained actions of the static median-sign rule.  Direct substitution gives
\begin{align*}
    R_T(\pi;0)-R_T^\circ
    &\stackrel{(c)}{=}
      a^-L_T-\sum_{\xi_t=+1}W_t\ell_t
      -a^+U_T-\sum_{\xi_t=-1}Z_tu_t\\
    &\stackrel{(d)}{=}
      a^-S_T-\Delta U_T
      -\sum_{\xi_t=+1}W_t\ell_t
      -\sum_{\xi_t=-1}Z_tu_t\\
    &\stackrel{(e)}{=}a^-S_T-C_T(\pi).
\end{align*}
Here $(c)$ uses $p_t=a^--W_t$ or $p_t=a^++Z_t$ according to the
side, $(d)$ uses $L_T-U_T=S_T$, which follows from
$y_T=0$ and \eqref{lb:eq-state-ledger}, and $(e)$ uses
$a^+-a^-=\Delta$ and \eqref{lb:eq-correction-cost}.  The terms
$R_T^\circ+a^-S_T$ are common to all cyclic feasible policies.
Applying the identity to $\pi$ and $\pi^{\mathsf{hs}}$ proves the
second equality in \eqref{lb:eq-two-accountings}.

For the first equality, sum \eqref{lb:eq-partial-regret} along the
trajectory of $\pi$.  Consecutive hindsight values cancel, the initial
term is $V_1^{\mathsf{hs}}(0;\mathbf p)=R_T(\pi^{\mathsf{hs}};0)$, and
the terminal term is $V_{T+1}^{\mathsf{hs}}(0;\varnothing)=0$ because
$y_T=0$.  This proves the pathwise telescope.
\end{proof}

Taking expectations in \eqref{lb:eq-two-accountings} gives the identity
$\Reg(\pi;B,T,0)=\E[C_T(\pi)-C_T(\pi^{\mathsf{hs}})]$.
We work with $C_T$ because it separates the fixed gap cost from the
edge premiums, while preserving the exact regret comparison.

\subsection{A deadline-respecting hindsight comparator}

We next construct a feasible hindsight path with a small correction cost.
Each new minimum of the walk generated by the static median-sign rule gives a deadline by which a
sale must have been replaced by a hold. Hindsight can meet each deadline
by selecting a low-premium arrival among those already available. We
first bound these premiums and then restore the terminal inventory.

Let $U$ be uniform on $(0,1)$ and define the nondecreasing functions
\begin{align*}
    \psi_+(u)&:=F^{-1}((1+u)/2)-a^+,\\
    \psi_-(u)&:=a^--F^{-1}((1-u)/2),
    \qquad 0<u<1.
\end{align*}
The conditional upper-side premium $Z_t$ has the law of $\psi_+(U)$,
and the conditional lower-side premium $W_t$ has the law of
$\psi_-(U)$. These representations follow by drawing a uniform quantile
on the corresponding half of the distribution. Conditional on the entire
sign path, premiums at distinct times are independent, with the
appropriate side-specific law.

Put $u_0:=\varepsilon_0$. Definition~\ref{def:continuous-family}, applied
at $q=u/2$, gives
\begin{align}
    d_\pm u^{1/\kappa}\leq\psi_\pm(u)
    \leq b_\pm u^{1/\kappa},
    \qquad 0<u\leq u_0,
    \label{lb:eq-local-premium-bounds}
\end{align}
where $d_\pm:=c_\pm2^{-1-1/\kappa}>0$ and $b_\pm:=3d_\pm$.
Here the constants $c_\pm$ are those in that definition, and the bounds
use $|r_\pm|\leq1/2$. Since $0\leq\psi_\pm(u)\leq1$, the constants
$A_\pm:=b_\pm\vee u_0^{-1/\kappa}$ give the global bounds
\begin{align}
    \psi_\pm(u)\leq A_\pm u^{1/\kappa},
    \qquad 0<u<1.
    \label{lb:eq-premium-quantile-bounds}
\end{align}
Indeed, \eqref{lb:eq-local-premium-bounds} applies below $u_0$, while
$1\leq u_0^{-1/\kappa}u^{1/\kappa}$ above it.

The lower bound for $\psi_+$ also controls the probability of a cheap
upper-side arrival. Put $z_0:=(d_+/2)u_0^{1/\kappa}>0$.
For $0<z\leq z_0$, monotonicity and
$\psi_+(u_0)\geq d_+u_0^{1/\kappa}=2z_0$ imply that
$\psi_+(U)\leq z$ requires $U<u_0$.
Equation~\eqref{lb:eq-local-premium-bounds} then implies
$U\leq(z/d_+)^\kappa$. Consequently,
\begin{align}
    \Pp(a^+\leq p_t\leq a^++z)
    =\tfrac12\Pp(\psi_+(U)\leq z)
    \leq K_Fz^\kappa,\qquad 0<z\leq z_0,
    \quad K_F:=\tfrac12d_+^{-\kappa}.
    \label{lb:eq-edge-mass-bound}
\end{align}
All these constants depend only on the fixed distribution and its
specified class parameters. The proof needs these power bounds; it
allows different leading constants and nonzero remainders at the two
edges.

\begin{lemma}
\label{lb:lem-order-statistic}
For either side $\nu\in\{-,+\}$, let $X_1,\ldots,X_m$ be i.i.d.\ with
the law of $\psi_\nu(U)$, and let $X_{(j:m)}$ be their $j$th order
statistic. For every $1\leq j\leq m$,
$\E[X_{(j:m)}]\leq A_\nu(j/(m+1))^{1/\kappa}$.
\end{lemma}

\begin{proof}
If $U_{(j:m)}$ is the $j$th order statistic of $m$ independent uniform
random variables, monotonicity gives
$X_{(j:m)}\stackrel{d}{=}\psi_\nu(U_{(j:m)})$.
Equation~\eqref{lb:eq-premium-quantile-bounds} therefore bounds its
expectation by $A_\nu\E[U_{(j:m)}^{1/\kappa}]$.
Since $u\mapsto u^{1/\kappa}$ is concave for $\kappa\geq1$,
Jensen's inequality and $\E[U_{(j:m)}]=j/(m+1)$ give the claim.
\end{proof}

For $y\geq1$, define the first-passage deadline
$\sigma_y:=\inf\{t\geq0:S_t=-y\}$, with $\inf\varnothing=\infty$.

\begin{lemma}
\label{lb:lem-deadline-sum}
Fix $\kappa\geq1$, put $a:=1/\kappa\in(0,1]$, and define
\begin{align*}
    \Phi_\kappa(T)
    :=\log T\mathbbm{1}\{\kappa=1\}
      +T^{(\kappa-1)/(2\kappa)}\mathbbm{1}\{\kappa>1\}.
\end{align*}
There exists a constant $C_\kappa<\infty$, depending only on $\kappa$,
such that $\E[\sum_{y:\sigma_y\leq T}(y/(\sigma_y+1))^a]
\leq C_\kappa\Phi_\kappa(T)$ for every $T\geq2$.
\end{lemma}

\begin{proof}
The probability that the walk first reaches $-y$ at time $s$ satisfies
\begin{align}
    \Pp(\sigma_y=s)
    =\frac{y}{s}\Pp(S_s=-y);
    \label{lb:eq-first-passage-identity}
\end{align}
both sides are zero when one of $s,y$ is even and the other is odd.
For completeness,
consider the walk during its first $s-1$ steps. Among paths ending at
$-y+1$, exclude those that have already reached $-y$. Reversing every
step after the first visit to $-y$ gives a one-to-one correspondence
between these excluded paths and all paths ending at $-y-1$.
With the convention that a
binomial coefficient is zero when its lower index is not an integer in its
usual range, appending the final down-step therefore gives
\begin{align*}
    2^s\Pp(\sigma_y=s)
    &\stackrel{(a)}{=}\binom{s-1}{(s-y)/2}
      -\binom{s-1}{(s-y-2)/2}\\
    &\stackrel{(b)}{=}\frac{y}{s}\binom{s}{(s-y)/2}
     =\frac{y}{s}\,2^s\Pp(S_s=-y),
\end{align*}
where $(a)$ counts the paths that have not yet reached $-y$ and appends
their final down-step, and $(b)$ uses
$\binom{s-1}{j}-\binom{s-1}{j-1}=(s-2j)\binom sj/s$
at $j=(s-y)/2$. This proves \eqref{lb:eq-first-passage-identity} directly.
Consequently,
\begin{align*}
    \E\left[
       \sum_{y:\sigma_y\leq T}
       \left(\frac{y}{\sigma_y+1}\right)^a
    \right]
    &\stackrel{(a)}{=}
      \sum_{s=1}^T\sum_{y=1}^s
      \left(\frac{y}{s+1}\right)^a
      \frac{y}{s}\Pp(S_s=-y)\\
    &\stackrel{(b)}{\leq}
      \sum_{s=1}^T s^{-a-1}\E|S_s|^{a+1}\\
    &\stackrel{(c)}{\leq}
      \sum_{s=1}^T s^{-(a+1)/2}\\
    &\stackrel{(d)}{\leq}C_\kappa\Phi_\kappa(T).
\end{align*}
Here $(a)$ uses \eqref{lb:eq-first-passage-identity}; $(b)$ uses
$(s+1)^{-a}s^{-1}\leq s^{-a-1}$ and groups the terms as a
one-sided moment of $S_s$; $(c)$ uses Lyapunov's inequality with
$a+1\leq2$ and $\E S_s^2=s$; and $(d)$ uses
$\sum_{s=1}^Ts^{-1}\leq1+\log T
\leq(1+1/\log2)\log T$ for $T\geq2$ when $a=1$, while for $a<1$ it uses
$1+\int_1^T x^{-(a+1)/2}\,dx\leq C_aT^{(1-a)/2}$.
\end{proof}

\begin{lemma}
\label{lb:lem-upper-deadlines}
There exists a constant $C_F<\infty$, depending only on
$F$, such that for every $T\geq2$ there is a hindsight selection
of upper-side arrivals with the following properties.
\begin{enumerate}[label=$(\alph*)$]
\item Exactly $M_T$ arrivals are selected, with at least $M_t$ selected
arrivals by every time $t$.
\item If each selected arrival is changed from a sale to a hold, the
expected total edge premium is at most $C_F\Phi_\kappa(T)$.
\end{enumerate}
\end{lemma}

\begin{proof}
Fix the sign path.  At time $\sigma_y\leq T$, the number of upper-side
arrivals is
\begin{align*}
    N_y^+
    :=\#\{s\leq\sigma_y:\xi_s=-1\}
    =\frac{\sigma_y+y}{2}\geq y.
\end{align*}
At each deadline $\sigma_y$, select the cheapest upper-side arrival that
has occurred by that deadline and has not already been selected.  The
deadlines are strictly increasing in $y$, so this selects one new arrival
at each of the $M_T$ attained levels.  Moreover, $M_t\geq y$ if and only
if $\sigma_y\leq t$, which proves the deadline property.

Conditional on the sign path, upper-side premiums are independent with
the law of $\psi_+(U)$. At most $y-1$ of the
$N_y^+$ available arrivals were removed before deadline $y$; therefore the
newly selected premium is no larger than the $y$th order statistic among
those $N_y^+$ premiums.  Lemma~\ref{lb:lem-order-statistic} and the formula
$N_y^+=(\sigma_y+y)/2$ give the conditional bound
\begin{align*}
    \E[\text{$y$th selected premium}\mid(\xi_s)_{s\leq T}]
    &\stackrel{(a)}{\leq}
      A_+\left(\frac{y}{N_y^++1}\right)^{1/\kappa}\\
    &\stackrel{(b)}{\leq} C_F
      \left(\frac{y}{\sigma_y+1}\right)^{1/\kappa}.
\end{align*}
Here $(a)$ uses Lemma~\ref{lb:lem-order-statistic} after at most $y-1$
removals, and $(b)$ uses $N_y^++1\geq(\sigma_y+1)/2$ and
$A_+2^{1/\kappa}\leq2A_+$, so the constant depends only on $F$.
Summing over attained deadlines and then applying
Lemma~\ref{lb:lem-deadline-sum} proves the premium bound.
\end{proof}

The upper selection makes inventory nonnegative but generally leaves
$S_T+M_T$ units at time step $T$. We remove this surplus by holding at
selected lower-side prices instead of buying. Applying the same deadline
construction in reverse time ensures that these holds occur late enough
to preserve nonnegative inventory.

\begin{lemma}
\label{lb:lem-reverse-repair}
Put $H_T:=S_T+M_T$. There exists a constant
$C_F<\infty$, depending only on $F$, such that
for every $T\geq2$ there is a hindsight selection of exactly $H_T$ lower-side
arrivals with the following properties.
\begin{enumerate}[label=$(\alph*)$]
\item The expected total edge premium is at most
$C_F\Phi_\kappa(T)$.
\item Together with the upper selection of
Lemma~\ref{lb:lem-upper-deadlines}, the selected holds produce a path
that is nonnegative at every time and ends at zero.
\end{enumerate}
\end{lemma}

\begin{proof}
Define the reversed walk $\widehat S_r:=S_{T-r}-S_T$ for
$0\leq r\leq T$.
Its increments are $-\xi_{T-r+1}$, so it is again a simple symmetric
walk.  An upper-side arrival for the reversed walk is precisely a
lower-side arrival in forward time, and
\begin{align*}
    \max_{0\leq r\leq T}(-\widehat S_r)^+
    &\stackrel{(a)}{=}\max_{0\leq s\leq T}(S_T-S_s)^+
    \stackrel{(b)}{=}S_T+M_T=H_T.
\end{align*}
Here $(a)$ substitutes the reversed walk and reindexes $s=T-r$, and
$(b)$ uses $\min_{s\leq T}S_s=-M_T\leq S_T$.
Apply the greedy construction in Lemma~\ref{lb:lem-upper-deadlines} to
$\widehat S$.  It selects exactly $H_T$ forward lower-side arrivals and,
by Lemma~\ref{lb:lem-order-statistic} applied to the lower-side law
$\psi_-(U)$ and the same deadline-sum calculation, has expected premium
at most $C_F\Phi_\kappa(T)$.

Let $L_t^\sigma$ be the number of these selected lower arrivals up to
forward time $t$.  The reversed deadline property at reverse time
$T-t$ is, on writing
$\mathcal R_t^{\mathrm{rev}}:=\max_{s\in[t,T]}(S_T-S_s)^+$,
$L_T^\sigma-L_t^\sigma\geq\mathcal R_t^{\mathrm{rev}}$.
Since $L_T^\sigma=H_T$, it follows that
\begin{align}
    L_t^\sigma\leq S_T+M_T-\mathcal R_t^{\mathrm{rev}}.
    \label{lb:eq-lower-count-prebound}
\end{align}
We now verify the forward-time inequality
\begin{align}
    L_t^\sigma\leq S_t+M_t.
    \label{lb:eq-lower-forward-bound}
\end{align}
If $M_t=M_T$, then
$\mathcal R_t^{\mathrm{rev}}\geq(S_T-S_t)^+$, and hence
$\mathcal R_t^{\mathrm{rev}}\geq S_T+M_T-S_t-M_t$.  If $M_t<M_T$, a time
$s^\star\in[t,T]$ satisfies $S_{s^\star}=-M_T$, so
$\mathcal R_t^{\mathrm{rev}}\geq S_T+M_T$.  Since $S_t+M_t\geq0$, the latter
bound also gives
$\mathcal R_t^{\mathrm{rev}}\geq S_T+M_T-S_t-M_t$.  Substitution into
\eqref{lb:eq-lower-count-prebound} proves
\eqref{lb:eq-lower-forward-bound} in both cases.

Let $U_t^\sigma$ be the count from the upper selection.  Its deadline
property gives $U_t^\sigma\geq M_t$.  Therefore
\begin{align*}
    y_t^\sigma
    &=S_t+U_t^\sigma-L_t^\sigma
      \stackrel{(a)}{\geq} S_t+M_t-L_t^\sigma
      \stackrel{(b)}{\geq}0,
\end{align*}
where $(a)$ uses the upper deadline property $U_t^\sigma\geq M_t$, and
$(b)$ is \eqref{lb:eq-lower-forward-bound}. At time
step $T$, $U_T^\sigma=M_T$ and
$L_T^\sigma=S_T+M_T$, so $y_T^\sigma=0$.
\end{proof}

\begin{proposition}
\label{lb:prop-comparator}
There exist finite constants $C_{F,\rho}$ and $C_\rho$,
depending only on their displayed parameters, such that for every $T\geq2$
there is a capacity-feasible hindsight policy $\sigma$ satisfying:
\begin{enumerate}[label=$(\alph*)$]
\item Its expected correction cost obeys
$\E[C_T(\sigma)]
    \leq
    \Delta\E[M_T]
    +C_{F,\rho}\Phi_\kappa(T)
    +C_{\rho}.$

\item $C_T(\pi^{\mathsf{hs}})\leq C_T(\sigma)$ pathwise.
\end{enumerate}
\end{proposition}

\begin{proof}
First use the upper and lower selected holds from
Lemmas~\ref{lb:lem-upper-deadlines} and
\ref{lb:lem-reverse-repair}.  They define actions in
$\{-1,0,1\}$: a selected arrival is held, and every unselected arrival
uses the action prescribed by the static median-sign rule.  The path is nonnegative and cyclic by
Lemma~\ref{lb:lem-reverse-repair}.  It remains to enforce the capacity.

Define $\mathcal G_T:=\{\max_{0\leq t\leq T}|S_t|\leq\rho T/4\}$.
The maximal Hoeffding inequality for a Rademacher walk gives
$\Pp(\mathcal G_T^c)
    \leq2\exp\left(-\frac{\rho^2T}{32}\right).$
On $\mathcal G_T$, $U_t^\sigma\leq U_T^\sigma=M_T$ and
$L_t^\sigma\geq0$, whence
\begin{align*}
    y_t^\sigma
    &=S_t+U_t^\sigma-L_t^\sigma
      \stackrel{(a)}{\leq} |S_t|+M_T
      \leq2\max_{s\leq T}|S_s| \stackrel{(b)}{\leq}\frac{\rho T}{2}
      \leq B.
\end{align*}
Here $(a)$ uses $U_t^\sigma\leq M_T$ and $L_t^\sigma\geq0$, and
$(b)$ uses the defining inequality of $\mathcal G_T$.
On $\mathcal G_T^c$, replace the selected path by the all-hold path,
which stays at zero and is feasible.  This event is known to the hindsight
comparator because it depends on the complete sign path.

On $\mathcal G_T$, the upper holds pay the fixed amount $\Delta M_T$.
The two expected edge-premium bounds are each
$C_F\Phi_\kappa(T)$.  On $\mathcal G_T^c$, every all-hold
correction costs at most one, so $C_T(\sigma)\leq T$.  Multiplying by the preceding Hoeffding bound on
$\Pp(\mathcal G_T^c)$ and enlarging the constant proves
the expected-cost bound. Finally, the common-term revenue identity
in the proof of Lemma~\ref{lb:lem-accounting} shows that maximizing
hindsight reward over feasible policies is equivalent to minimizing
$C_T$.  Since $\sigma$ is feasible, the pathwise inequality for
$\pi^{\mathsf{hs}}$ follows.
\end{proof}

\subsection{Reflected-walk boundary estimates}

To compare an online policy with this hindsight path, we need two
properties of the walk generated by the static median-sign rule: how often it is close to its running
minimum, and how likely it is to avoid that minimum after moving away.

Put $\overline Y_t:=S_t+M_t$ for $0\leq t\leq T$.
It is the inventory of the static median-sign rule with sales blocked
at zero, ignoring the upper boundary and terminal requirement. It obeys
$\overline Y_t=0\vee(\overline Y_{t-1}+\xi_t)$ and starts at zero.

\begin{lemma}
\label{lb:lem-boundary-estimates}
There exist universal constants $c,C>0$ such that for every $T\geq1$:
\begin{enumerate}[label=$(\alph*)$]
\item For every $L\geq0$,
$\sum_{t=1}^T\Pp(\overline Y_{t-1}\leq L)\leq C(L+1)\sqrt T$.
\item $\E[M_T]\geq c\sqrt T$.
\end{enumerate}
\end{lemma}

\begin{proof}
The reflection principle gives, for every integer $j\geq0$,
\begin{align*}
    \Pp(\overline Y_t=j)
    =\Pp(S_t=j)+\Pp(S_t=-j-1).
\end{align*}
The binomial coefficient bound
$\sup_{m\in\mathbb Z}\Pp(S_t=m)\leq C/\sqrt{t+1}$ therefore yields
$\Pp(\overline Y_t\leq L)\leq C(L+1)/\sqrt{t+1}$.
Summing this estimate for $t=0,\ldots,T-1$ and using
$\sum_{t=0}^{T-1}(t+1)^{-1/2}\leq2\sqrt T$ proves
the boundary-occupation claim.

For the second claim, $M_T\geq(-S_T)^+$ and symmetry gives
$\E[(-S_T)^+]=\E|S_T|/2$. We bound this absolute moment using
$\E S_T^2=T$ and
$\E S_T^4=3T^2-2T\leq3T^2$.  Paley--Zygmund applied to $S_T^2$ gives
\begin{align*}
    \Pp(S_T^2\geq T/2)
    &\stackrel{(a)}{\geq}\frac{(1-1/2)^2(\E S_T^2)^2}{\E S_T^4}
    \stackrel{(b)}{\geq}\frac1{12}.
\end{align*}
Here $(a)$ applies Paley--Zygmund to the nonnegative variable $S_T^2$
at half its mean, and $(b)$ substitutes the second- and fourth-moment bounds.
Thus $\E|S_T|\geq\sqrt{T/2}/12$, proving
the reflection-size claim.
\end{proof}

\begin{lemma}
\label{lb:lem-survival}
Let $X_s=x+S_s$ be a simple symmetric random walk started from an integer
$x\geq1$, and let $\tau_0:=\inf\{s\geq0:X_s=0\}$.  There is a universal
$c>0$ such that, for every $s\geq0$,
$\Pp_x(\tau_0>s)\geq c\bigl(1\wedge(x/\sqrt{s+1})\bigr)$.
\end{lemma}

\begin{proof}
Reflection of paths at their first hit of zero gives
\begin{align*}
    \Pp_x(\tau_0>s)
    =\Pp(1-x\leq S_s\leq x),
\end{align*}
with the lattice parity inherited from $S_s$.  If
$1\leq x\leq\sqrt{s+1}/4$, the interval $[1-x,x]$ contains at least $c_0x$ points of
the admissible parity, all lying within distance $\sqrt{s+1}$ of zero.
The elementary central-binomial lower bound is
$\Pp(S_s=j)\geq c_1/\sqrt{s+1}$ for every such admissible $j$.
Summing gives $c x/\sqrt{s+1}$.  If
$x>\sqrt{s+1}/4$, monotonicity of the survival probability in the
starting height and the preceding estimate at an integer height of order
$\sqrt{s+1}$ give a positive constant.  Adjusting $c$ handles the
finitely many cases in which $\sqrt{s+1}/4<1$.
\end{proof}

\subsection{Scarcity and the last return}

An online policy faces a tradeoff. Near the boundary, inexpensive
upper-side corrections are scarce. Far from the boundary, an inexpensive
correction may turn out to be unnecessary: if the free walk never
returns to its running minimum, that correction increases the total
number of corrections beyond the minimum required for feasibility.
We quantify these two costs separately.

Fix an integer $L$ satisfying $2\leq L\leq\sqrt T$ and a price threshold
$z\in(0,z_0]$.  Split all upper-side correction units into
\begin{align}
    U_{\mathrm{in}}
    &:=\sum_{t=1}^T u_t
       \one\{Z_t\leq z,\ \overline Y_{t-1}\leq L\},
    \notag\\
    U_{\mathrm{out}}
    &:=\sum_{t=1}^T u_t
       \one\{Z_t\leq z,\ \overline Y_{t-1}>L\},
    \notag\\
    U_{\mathrm{exp}}
    &:=\sum_{t=1}^T u_t\one\{Z_t>z\}.
    \label{lb:eq-upper-decomposition}
\end{align}
Because $u_t=0$ off the upper side, these three variables partition
$U_T$. Define the excess correction count $A_T:=U_T-M_T\geq0$.

\begin{lemma}
\label{lb:lem-scarcity}
For the fixed integer $2\leq L\leq\sqrt T$ and $z\in(0,z_0]$, there
exists a constant $C_F<\infty$, depending only on
$F$, such that
$\E[U_{\mathrm{in}}]\leq C_Fz^\kappa(L+1)\sqrt T$.
\end{lemma}

\begin{proof}
The event involving $\overline Y_{t-1}$ is measurable with respect to
the past signs, whereas $p_t$ is independent of that past. Thus
\begin{align*}
    \E[U_{\mathrm{in}}]
    &\stackrel{(a)}{\leq}2\sum_{t=1}^T
       \Pp(\xi_t=-1,Z_t\leq z,\overline Y_{t-1}\leq L)\\
    &\stackrel{(b)}{\leq}2K_Fz^\kappa\sum_{t=1}^T
       \Pp(\overline Y_{t-1}\leq L)
     \stackrel{(c)}{\leq}C_Fz^\kappa(L+1)\sqrt T,
\end{align*}
where $(a)$ uses $u_t\leq2$ and $u_t=0$ off the upper side,
$(b)$ uses the stated independence and
$\Pp(\xi_t=-1,Z_t\leq z)\leq K_Fz^\kappa$ from
\eqref{lb:eq-edge-mass-bound}, and $(c)$ uses the boundary-occupation bound in
Lemma~\ref{lb:lem-boundary-estimates}.
\end{proof}

\begin{lemma}
\label{lb:lem-last-return}
For the fixed integer $2\leq L\leq\sqrt T$ and $z\in(0,z_0]$, there
exists a universal constant $c>0$ such that
$\E[A_T]\geq c(L/\sqrt T)\E[U_{\mathrm{out}}]$.
\end{lemma}

\begin{proof}
Let $g_T:=\max\{0\leq t\leq T:\overline Y_t=0\}$ be the last return
of the reflected free walk to zero. Since $M_t$ is
constant after $g_T$, $M_T=M_{g_T}$.  Feasibility and
\eqref{lb:eq-state-ledger} imply $U_{g_T}\geq M_{g_T}$.  Therefore,
pathwise,
\begin{align}
    A_T
    &=U_T-M_T
      \stackrel{(a)}{=}(U_{g_T}-M_{g_T})+\sum_{t>g_T}u_t
      \notag\\
    &\stackrel{(b)}{\geq}\sum_{t=1}^Tu_t\one\{t>g_T\}.
    \label{lb:eq-after-last-return}
\end{align}
Here $(a)$ uses $M_T=M_{g_T}$ and splits the correction sum at $g_T$,
and $(b)$ uses $U_{g_T}\geq M_{g_T}$.

For each $t$, put
$I_t:=\one\{\xi_t=-1,Z_t\leq z,\overline Y_{t-1}>L\}$.
The product $u_tI_t$ is $\mathcal F_t$-measurable.  On $I_t=1$,
$\overline Y_{t-1}$ is an integer at least $L+1$ and the current sign is
$-1$, so $\overline Y_t\geq L$.  Conditional on $\mathcal F_t$, future
signs are fresh Rademacher variables.  Until its next hit of zero, the
reflected walk starting from $\overline Y_t$ is an ordinary simple walk.
Lemma~\ref{lb:lem-survival}, $L\leq\sqrt T$, and
$T-t+1\leq T+1$ therefore give
\begin{align}
    \Pp(t>g_T\mid\mathcal F_t)
    \geq c\left(1\wedge
             \frac{L}{\sqrt{T-t+1}}\right)
    \geq c'\frac{L}{\sqrt T}
    \qquad\text{on }\{I_t=1\}.
    \label{lb:eq-conditional-last-return}
\end{align}
Using \eqref{lb:eq-after-last-return}, the tower property, and
\eqref{lb:eq-conditional-last-return},
\begin{align*}
    \E[A_T]
    &\stackrel{(a)}{\geq}
      \sum_{t=1}^T\E[u_tI_t\one\{t>g_T\}]\\
    &\stackrel{(b)}{=}
      \sum_{t=1}^T
      \E\!\left[u_tI_t
        \Pp(t>g_T\mid\mathcal F_t)\right]\\
    &\stackrel{(c)}{\geq}
      c'\frac{L}{\sqrt T}\sum_{t=1}^T\E[u_tI_t]
     =c'\frac{L}{\sqrt T}\E[U_{\mathrm{out}}].
\end{align*}
Here $(a)$ restricts the nonnegative sum in
\eqref{lb:eq-after-last-return}, $(b)$ uses
$\mathcal F_t$-measurability and the tower property, and $(c)$ uses
\eqref{lb:eq-conditional-last-return}.
\end{proof}

\subsection{Balancing scarcity and excess corrections}

We now choose the inventory layer and the price-premium cutoff together.
The choice makes cheap corrections outside the layer at least as costly
as the expensive corrections they replace. The scarcity of cheap
corrections inside the layer then gives the lower bound.

\begin{proof}[Proof of Theorem~\ref{lb:thm-positive-gap}]
The common-term identity in Lemma~\ref{lb:lem-accounting}, optimality of
$\pi^{\mathsf{hs}}$, and Proposition~\ref{lb:prop-comparator} imply
\begin{align}
    \Reg(\pi;B,T,0)
    &=\E[C_T(\pi)-C_T(\pi^{\mathsf{hs}})] \notag\\
    &\stackrel{(a)}{\geq}
      \Delta\E[A_T]+z\E[U_{\mathrm{exp}}]
      -C_{F,\rho}\Phi_\kappa(T)-C_\rho.
    \label{lb:eq-reduced-cost}
\end{align}
In step $(a)$, the online upper corrections contribute at least
$\Delta U_T+zU_{\mathrm{exp}}$, the online lower premium is
nonnegative, the comparator contributes at most
$\Delta\E M_T+C\Phi_\kappa(T)+C$, and
$A_T=U_T-M_T$.

The partition \eqref{lb:eq-upper-decomposition} and
$U_T=M_T+A_T$ give
\begin{align}
    U_{\mathrm{exp}}
    =M_T+A_T-U_{\mathrm{in}}-U_{\mathrm{out}}.
    \label{lb:eq-expensive-identity}
\end{align}
Substituting \eqref{lb:eq-expensive-identity} into
\eqref{lb:eq-reduced-cost} and dropping the nonnegative term
$z\E[A_T]$ yields
\begin{align}
    \Reg(\pi;B,T,0)
    &\geq
      \Delta\E[A_T]
      +z\E[M_T]
      -z\E[U_{\mathrm{in}}]
      -z\E[U_{\mathrm{out}}]
      -C\Phi_\kappa(T)-C \notag\\
    &\stackrel{(b)}{\geq}
      \left(c\Delta\frac{L}{\sqrt T}-z\right)
      \E[U_{\mathrm{out}}]
      +z\E[M_T]-z\E[U_{\mathrm{in}}]
      -C\Phi_\kappa(T)-C,
    \label{lb:eq-before-balance}
\end{align}
where $(b)$ uses Lemma~\ref{lb:lem-last-return}.

Choose
\begin{align*}
    z:=\frac{c\Delta}{2}\frac{L}{\sqrt T}.
\end{align*}
For the value of $L$ below, $z\leq z_0$ for all sufficiently large $T$.
The coefficient of $\E[U_{\mathrm{out}}]$ in
\eqref{lb:eq-before-balance} is nonnegative.  Lemmas
\ref{lb:lem-boundary-estimates} and \ref{lb:lem-scarcity}, followed by
this choice of $z$, therefore give
\begin{align}
    \Reg(\pi;B,T,0)
    &\stackrel{(c)}{\geq}
      c_1L
      -Cz^{\kappa+1}(L+1)\sqrt T
      -C\Phi_\kappa(T)-C \notag\\
    &\stackrel{(d)}{\geq}
      c_1L
      -c_2\frac{L^{\kappa+2}}{T^{\kappa/2}}
      -C\Phi_\kappa(T)-C.
    \label{lb:eq-final-balance}
\end{align}
Here $(c)$ uses $\E M_T\geq c\sqrt T$ and
Lemma~\ref{lb:lem-scarcity}; $(d)$ uses $L\geq2$ and
$z\asymp_{\Delta}L/\sqrt T$.

Let $\gamma:=\kappa/(2(\kappa+1))$ and
$L:=\lfloor aT^\gamma\rfloor$, where $a\in(0,1)$ is fixed below.
For sufficiently large $T$,
$2\leq L\leq\sqrt T$, $z\leq z_0$, and
$L\geq(a/2)T^\gamma$. Moreover,
$L^{\kappa+2}/T^{\kappa/2}\leq a^{\kappa+2}T^\gamma$ because
$\gamma(\kappa+2)-\kappa/2=\gamma$.  Choose $a$ sufficiently small
that $c_2a^{\kappa+2}\leq c_1a/4$.  The first two terms in
\eqref{lb:eq-final-balance} are then at least
$(c_1a/4)T^\gamma$ for all sufficiently large $T$.

It remains to compare the hindsight premium.  If $\kappa=1$, then
$\Phi_\kappa(T)=\log T=o(T^{1/4})$. If $\kappa>1$, then
$\gamma-(\kappa-1)/(2\kappa)=1/(2\kappa(\kappa+1))>0$,
so again $\Phi_\kappa(T)=o(T^\gamma)$.  Enlarging the threshold
$T_0(F,\rho)$ absorbs these lower-order terms and all preceding
endpoint restrictions.  Since
$\gamma=(\beta+1)/(2\beta+4)$, this proves
the bound in Theorem~\ref{lb:thm-positive-gap}.

The bound combines the exact accounting in Lemma~\ref{lb:lem-accounting}
with the feasible hindsight comparison in Proposition~\ref{lb:prop-comparator}.
Lemma~\ref{lb:lem-scarcity} limits cheap corrections near the floor,
while Lemma~\ref{lb:lem-last-return} bounds the cost of making them farther
away. Every constant is independent of $T$ and of the online policy;
it may depend on the fixed distribution $F$ and on $\rho$. In particular,
the positive gap and the two local edge constants are held fixed as
$T$ varies.
\end{proof}

\section{Certainty equivalence for uniform prices}
\label{ce:appendix}

We prove Proposition~\ref{ce:theorem} by comparing certainty equivalence
with the static median-sign rule. From an interior state, the policy's restoring
drift keeps inventory away from the physical boundaries. Starting with zero
inventory, the same rule repeatedly rejects high-price selling opportunities.
Throughout, prices
$p_1,\ldots,p_T$ are independent and uniform on $[0,1]$, and
$\mathcal F_t=\sigma(p_1,\ldots,p_t)$, with $\mathcal F_0$ trivial.
We use the projected certainty-equivalent rule in
Section~\ref{ce:section}. Constants do not depend on
$T,B,y_0$ unless a dependence is indicated.

\subsection{Feasibility and the cost of departing from the static median-sign rule}

If $|y-y_0|\le\tau_t$, moving one unit toward $y_0$, or holding when
$y=y_0$, belongs to $\mathcal A_t(y;y_0)$. Hence this set is nonempty
at every viable state. Its projection yields another viable state and
changes inventory by at most one. Starting at $y_0$, induction proves
physical feasibility at every time and $y_T=y_0$. In particular,
$|y_{t-1}-y_0|\le\tau_t$, so the threshold belongs to $[0,1]$.
The inventory is $\mathcal F_{t-1}$-measurable before the next price
arrives, and that price is conditionally uniform.

We first verify that the threshold maximizes the fluid objective. Put
$q_*=1/2-d/(2\tau_t)$ and $x_*(q)=1$ below $q_*$ and $-1$ above it.
It satisfies the fluid constraint. For uniform prices, any other
feasible profile $x$ satisfies
\begin{align*}
 \int_0^1(-q)\{x_*(q)-x(q)\}\,dq
 &\stackrel{(a)}{=}\int_0^1(q_*-q)\{x_*(q)-x(q)\}\,dq
 \stackrel{(b)}{\ge}0.
\end{align*}
Here $(a)$ uses the common integral of the two feasible profiles, and
$(b)$ uses the nonnegative integrand on both sides of $q_*$. At $q_*=0$ or one,
the constraint fixes the profile almost everywhere. At an interior
threshold, the same argument applies to an atomless general law using
$F^{-1}(q_*)$ as the multiplier. We now specialize to uniform prices
to bound regret.

\begin{lemma}
\label{ce:accounting}
For a price vector with no observation equal to $1/2$, set
$x_t^*=\one\{p_t<1/2\}-\one\{p_t>1/2\}$ and
$H_T=\sum_{t=1}^T|p_t-1/2|$.
For any feasible cyclic trading path $\pi$, define its premium by
\begin{align*}
 G_T(\pi)&=\sum_{t=1}^T g_t(x_t^\pi),&
 g_t(x)&=|p_t-1/2|+(p_t-1/2)x.
\end{align*}
Then the following statements hold.
\begin{enumerate}[label=$(\alph*)$]
\item For every $x\in\{-1,0,1\}$, $0\le g_t(x)\le1$.
\item The realized reward satisfies $R_T(\pi;y_0)=H_T-G_T(\pi)$.
\item For any other feasible path $\pi^{\rm cmp}$, including one that
uses the entire price sequence,
\begin{align*}
 \Reg(\pi;B,T,y_0)&\le\E G_T(\pi),\\
 \Reg(\pi;B,T,y_0)&\ge\E G_T(\pi)-\E G_T(\pi^{\rm cmp})
\end{align*}
\end{enumerate}
\end{lemma}
\begin{proof}
The premium is zero for the action prescribed by the static median-sign rule, $|p_t-1/2|$ for holding,
and $2|p_t-1/2|$ for the opposite trade; hence the stated bounds.
Cyclic feasibility gives $\sum_t x_t^\pi=0$, and therefore
\begin{align*}
 H_T-G_T(\pi)
 &\stackrel{(a)}=-\sum_{t=1}^T(p_t-1/2)x_t^\pi
 \stackrel{(b)}=-\sum_{t=1}^T p_tx_t^\pi=R_T(\pi;y_0),
\end{align*}
where $(a)$ uses the definition of the premium and $(b)$ uses the zero total
trade. The hindsight optimum also has nonnegative premium, proving
the upper bound. It earns at least the revenue of any feasible
comparison path, proving the comparison lower bound.
\end{proof}

\subsection{Interior states: contraction, concentration, and terminal projection}

We first control the squared displacement of inventory from its target;
this determines the expected premium of the tentative fluid trade.
We then show that a linearly distant physical boundary is reached with
exponentially small probability. Both bounds apply to the actual
projected inventory, without conditioning on future boundary avoidance.

\begin{lemma}
\label{ce:bridge}
Let $D_t=y_t^{\textup{\textsf{CE}}}-y_0$.
\begin{enumerate}[label=$(\alph*)$]
\item For every $1\le t\le T$, $\E D_{t-1}^2\le\tau_t$, and $D_T=0$.
\item For every $z>0$,
$\Pp(\max_{0\le t\le T}|D_t|\ge z)\le2\exp\{-z^2/(32T)\}$.
\end{enumerate}
\end{lemma}
\begin{proof}
Write $\widetilde x_t$ for the tentative action. Conditional on
$\mathcal F_{t-1}$, its mean is $-D_{t-1}/\tau_t$ and its square is
one, outside a probability-zero price-tie event. The centered inventory
after projection is the clipping of $D_{t-1}+\widetilde x_t$ to
$[(-y_0)\vee(-(\tau_t-1)),(B-y_0)\wedge(\tau_t-1)]$.
This interval contains zero, so clipping cannot increase the absolute
value. For $\tau_t\ge2$,
\begin{align*}
 \E[D_t^2\mid\mathcal F_{t-1}]
 &\stackrel{(a)}\le
 D_{t-1}^2+2D_{t-1}\E[\widetilde x_t\mid\mathcal F_{t-1}]+1
 \\
 &\stackrel{(b)}=
 (1-2/\tau_t)D_{t-1}^2+1.
\end{align*}
Here $(a)$ uses clipping and expands the square, and $(b)$ substitutes
the conditional mean. The coefficient is nonnegative. Starting with
$\E D_0^2=0\le T$, taking expectations and using the tower property
gives inductively
$(1-2/\tau_t)\tau_t+1=\tau_t-1$ as an upper bound for
$\E D_t^2$. This proves the moment claim through $t=T-1$;
the final inventory is zero in centered coordinates by feasibility.

For concentration, put $\xi_t=\widetilde x_t+D_{t-1}/\tau_t$,
$Z_t=\sum_{s=1}^t\xi_s$, and $Z_0=0$.
Then $Z_t$ is a martingale and $|\xi_t|\le2$. Fix a time with
$D_t>0$, and let $\ell<t$ be the last time at which $D_\ell=0$.
Such a time exists, and all subsequent inventory values up to $t$ are
positive: integer increments of size at most one cannot cross zero
without visiting it. On this excursion, projection can only reduce a
positive tentative value, and the predictable drift is nonpositive.
Thus
\begin{align*}
 D_t
 &\stackrel{(a)}{\le}\sum_{s=\ell+1}^t\widetilde x_s \stackrel{(b)}{=}Z_t-Z_\ell-\sum_{s=\ell+1}^t D_{s-1}/\tau_s
 \stackrel{(c)}{\le}Z_t-Z_\ell.
\end{align*}
Here $(a)$ sums the nonpositive projection corrections from $D_\ell=0$,
$(b)$ substitutes the definition of $Z$, and $(c)$ uses
$D_{s-1}\ge0$ on this excursion. On a negative excursion the calculation with signs reversed
gives $-D_t\le Z_\ell-Z_t$. Hence
$\max_t|D_t|\le2\max_t|Z_t|$. The maximal Azuma--Hoeffding bound
for martingale differences bounded by two gives
$\Pp(\max_t|Z_t|\ge z/2)\le2e^{-z^2/(32T)}$;
the pathwise comparison therefore proves the maximal-displacement claim.
\end{proof}

\begin{lemma}
\label{ce:terminal}
Let $\mathcal B$ be the event that some tentative next inventory lies
outside $[0,B]$, and put $d_0=y_0\wedge(B-y_0)$.
\begin{enumerate}[label=$(\alph*)$]
\item On $\mathcal B^c$, at most one tentative action is changed by projection.
\item If $d_0>0$, then $\Pp(\mathcal B)\le2e^{-d_0^2/(32T)}$.
\end{enumerate}
\end{lemma}
\begin{proof}
A unit tentative step leaves $[0,B]$ only from inventory zero or $B$.
Thus $\mathcal B$ implies $\max_t|D_t|\ge d_0$, and
the claimed probability bound follows from Lemma~\ref{ce:bridge}.

On $\mathcal B^c$, only the terminal cone can modify an action.
If $|D_{t-1}|=\tau_t$, the threshold is zero or one and the action
is necessarily inward, hence feasible. If
$|D_{t-1}|\le\tau_t-2$, both unit steps satisfy the cone. The only
remaining possibility is $|D_{t-1}|=\tau_t-1$: an outward step is
changed to a hold. When $\tau_t=1$ and $D_{t-1}=0$, either unit
step is changed to a hold, which is the same remaining possibility.

Before the first such hold, $D_{t-1}$ has the parity of $t-1$.
Equality $|D_{t-1}|=\tau_t-1=T-t$ requires $T$ to be odd.
After a hold has occurred, the parity changes by one relative to
the all-unit-step path, so this equality cannot hold again unless
another hold has already occurred. Considering the first putative
second hold rules that out. For even $T$ there is no first hold;
for odd $T$ there is at most one. These statements exclude only
the probability-zero price ties already specified.
\end{proof}

\begin{proof}[Proof of the interior bound in Proposition~\ref{ce:theorem}]
When $T=1$, every cyclic feasible path holds, so regret is zero.
Assume $T\ge2$.
Fix the observed history and set $s=-D_{t-1}/(2\tau_t)$, the
threshold's displacement from the median. On the interval between
these thresholds the tentative action reverses the action prescribed by the static median-sign rule,
and elsewhere the two agree. Therefore
\begin{align}
 \E[g_t(\widetilde x_t)\mid\mathcal F_{t-1}]
 =2\int_0^{|s|}u\,du
 =\frac{D_{t-1}^2}{4\tau_t^2}.
 \label{ce:tentative-premium}
\end{align}
For any actual and tentative actions,
$g_t(x_t)-g_t(\widetilde x_t)
=(p_t-1/2)(x_t-\widetilde x_t)$ is at most one and is zero
when the actions agree. Lemma~\ref{ce:terminal} consequently gives
the pathwise bound
$G_T(\pi^{\textsf{CE}})\le\sum_{t=1}^T g_t(\widetilde x_t)+1+T\one_{\mathcal B}$.
Taking expectations, without conditioning on $\mathcal B^c$, yields
\begin{align*}
 \Reg(\pi^{\textsf{CE}};B,T,y_0)
 &\stackrel{(a)}\le\E G_T(\pi^{\textsf{CE}})\\
 &\stackrel{(b)}\le\frac14\sum_{t=1}^T
       \frac{\E D_{t-1}^2}{\tau_t^2}
       +1+2T e^{-d_0^2/(32T)}\\
 &\stackrel{(c)}\le\frac14(1+\log T)+1
       +2T e^{-\eta^2T/32}
 \le C_\eta\log T.
\end{align*}
Here $(a)$ is the upper bound in Lemma~\ref{ce:accounting},
$(b)$ uses the tower property, \eqref{ce:tentative-premium}, and
the physical-event bound in Lemma~\ref{ce:terminal}, and $(c)$ uses
the moment bound in Lemma~\ref{ce:bridge},
$d_0\ge\eta T$, and
$\sum_{j=1}^Tj^{-1}\le1+\int_1^T u^{-1}du$.
The last bound is uniform in $T\ge2$, since $\log T\ge\log2$ and
$\sup_{u\ge0}u e^{-\eta^2u/32}<\infty$.
\end{proof}

\subsection{Zero initial inventory: repeated visits and profitable delayed sales}

The lower bound follows from two observations. Certainty equivalence
spends at least order $\sqrt T$ time steps near zero inventory, and
some sales of its last unit can be profitably postponed by one time
step. We count these improvements directly and show that they can all
be made along one feasible hindsight path.

\begin{proof}[Proof of the bound for zero initial inventory in Proposition~\ref{ce:theorem}]
Assume $y_0=y_T=0$, $B\ge1$, and $T\ge2$. Throughout this proof,
$y_t$ denotes the inventory under $\pi^{\textsf{CE}}$.

\emph{Count the visits to zero.}
Let $\xi_t=\one\{p_t<1/2\}-\one\{p_t>1/2\}$,
$S_t=\sum_{s=1}^t\xi_s$, and $S_0=0$. These signs are the proposed
actions of the static median-sign rule. With sales blocked at zero,
its inventory is
$Z_0=0$ and $Z_t=0\vee(Z_{t-1}+\xi_t)$. It ignores the upper
boundary and terminal target and is used only to compare inventories.
Put $M_t=\max_{0\le s\le t}(-S_s)^+$. Induction on $t$ gives
$Z_t=S_t+M_t$, and $M_t$ increases by one exactly when
$Z_{t-1}=0$ and $\xi_t=-1$. Independence of the current price gives
\begin{align*}
 \E\sum_{t=1}^T\one\{Z_{t-1}=0\}
 \stackrel{(a)}=2\E\sum_{t=1}^T
       \one\{Z_{t-1}=0,\xi_t=-1\}
 \stackrel{(b)}=2\E M_T.
\end{align*}
Here $(a)$ uses the probability $1/2$ of a high price at any fixed
past history, and $(b)$ counts the increases of $M_t$.

The \textsf{CE} threshold is at most $1/2$, so its tentative action
$\widetilde x_t$ satisfies $\widetilde x_t\le\xi_t$ on the same
price sequence. Projection clips the tentative next inventory to
$[0,B\wedge(\tau_t-1)]$. Therefore, if $y_{t-1}\le Z_{t-1}$,
\begin{align*}
 y_t
 &\stackrel{(a)}\le 0\vee(y_{t-1}+\xi_t)
 \stackrel{(b)}\le 0\vee(Z_{t-1}+\xi_t)=Z_t.
\end{align*}
Step $(a)$ uses $\widetilde x_t\le\xi_t$ and the fact that clipping
at the upper endpoint can only reduce inventory. Step $(b)$ uses
the induction hypothesis. Starting from $y_0=Z_0=0$, this proves
$0\le y_t\le Z_t$ for every $t$. In particular, with
$N_0=\sum_{t=1}^T\one\{y_{t-1}=0\}$, we have
$\E N_0\ge2\E M_T$.

For completeness, the scale of $M_T$ follows directly from the
symmetric random walk. Its moments satisfy $\E S_T^2=T$ and
$\E S_T^4=3T^2-2T\le3T^2$. H\"older's inequality gives
$T\le(\E|S_T|)^{2/3}(\E S_T^4)^{1/3}$, hence
$\E|S_T|\ge\sqrt{T/3}$ and
$\E M_T\ge\E(-S_T)^+=\E|S_T|/2\ge\sqrt T/(2\sqrt3)$.
The matching upper bound follows from the reflection identity
$\Pp(M_T\ge y)=\Pp(S_T\le-y)+\Pp(S_T<-y)$ for integers $y\ge1$.
To obtain this identity, reflect all subsequent steps after the
first visit to $-y$; paths that visit $-y$ and end above it correspond
one-to-one to paths ending below it. Summing over $y$ yields
$\E M_T\le2\E(-S_T)^+=\E|S_T|\le\sqrt T$.

\emph{Count the opportunities to sell the last unit.}
Let $A$ count transitions from inventory zero to one. Every such
transition occurs before $T$, and, at zero before the final step,
a purchase occurs exactly when $p_t<1/2$. Also let
$N_1=\sum_{t=1}^{T-1}\one\{y_{t-1}=1\}$.
The number of transitions from one to zero equals $A$, because the
path starts and ends at zero and trades at most one unit per step.
At most one of these downward transitions is at time $T$. Thus
\begin{align*}
 \E N_1
 &\stackrel{(a)}\ge \E A-1\\
 &\stackrel{(b)}=\frac12
       \E\sum_{t=1}^{T-1}\one\{y_{t-1}=0\}-1\\
 &\stackrel{(c)}\ge \frac12\E N_0-\frac32
 \ge \E M_T-\frac32
 \ge \frac{\sqrt T}{2\sqrt3}-\frac32.
\end{align*}
Here $(a)$ counts the downward transitions before the final step,
$(b)$ uses the independent uniform price when the inventory level is zero, and
$(c)$ removes at most one final-step indicator.

\emph{Turn these opportunities into hindsight profit.}
For $1\le t<T$, define
$E_t=\{y_{t-1}=1,\ p_t\in[1/2,3/4],\
p_{t+1}\in[7/8,1]\}$.
On $E_t$, certainty equivalence sells at $t$ and holds at $t+1$.
Change these two actions to hold at $t$ and sell at $t+1$.
Inventory is now one instead of zero only between these two decisions;
at time $t+1$ it is again zero, and every later inventory is unchanged.
The modified pair therefore respects capacity even when $B=1$ and
preserves the terminal target, including when $t=T-1$. Its profit
increases by $p_{t+1}-p_t\ge1/8$.

The events $E_t$ and $E_{t+1}$ cannot both occur, because $E_t$
implies $y_t=0$, while $E_{t+1}$ requires $y_t=1$. Hence all selected
two-step intervals are disjoint. Apply every selected change to the
original \textsf{CE} path. This constructs a feasible hindsight path
whose gain is at least $\frac18\sum_{t=1}^{T-1}\one_{E_t}$.
Finally, $y_{t-1}$ is determined by $\mathcal F_{t-1}$, while
$p_t$ and $p_{t+1}$ are independent of it. Consequently,
\begin{align*}
 \Reg(\pi^{\textsf{CE}};B,T,0)
 &\stackrel{(a)}\ge\frac18\sum_{t=1}^{T-1}\Pp(E_t) \stackrel{(b)}=\frac1{256}\E N_1
 \stackrel{(c)}\ge\frac1{256}
       \left(\frac{\sqrt T}{2\sqrt3}-\frac32\right).
\end{align*}
Step $(a)$ uses the feasible improved path as a lower bound on
hindsight profit. For $(b)$, the two independent price intervals have
probabilities $1/4$ and $1/8$, so
$\Pp(E_t\mid\mathcal F_{t-1})=\one\{y_{t-1}=1\}/32$.
Step $(c)$ uses the visit count above. For $T\ge108$, the last
expression is at least $\sqrt T/(1024\sqrt3)$, proving part~(b).
\end{proof}

\section{Benchmark structure, concavity, and illustrative examples}
\label{app:basic-facts}

\subsection{When the static hindsight solution respects storage}

If the initial inventory is a linear distance from both boundaries,
the static hindsight solution is feasible with high probability. The
reason is that prices are i.i.d.: conditional on their ranks, purchases
and sales arrive in a uniformly random order. We bound the resulting
inventory fluctuations below.

\begin{proposition}\label{prop:static-interior}
Fix $\eta>0$ and assume $y_0\wedge(B-y_0)\ge\eta T$. Let prices be i.i.d.\
on $[0,1]$, and let $V^{\mathrm{stat}}(\mathbf p)$ be the best realized
profit with unit trades and $\sum_tx_t=0$, but without intermediate
storage constraints. Then
\begin{align*}
 0\le\E[V^{\mathrm{stat}}(\mathbf p)-V_1^{\mathsf{hs}}(B,T,y_0;\mathbf p)]
 \le 2T^2e^{-\eta^2T/2}.
\end{align*}
\end{proposition}
\begin{proof}
For $m$ purchases and $m$ sales, an exchange argument assigns purchases to the $m$ smallest prices and sales to the $m$ largest. Increasing $m$ by one until $\lfloor T/2\rfloor$ adds a nonnegative price spread. Thus the balanced rank prescription in the introduction is optimal. Break any price ties by independent continuous random numbers; this selects an optimal prescription without changing profit.

Conditional on the ordered price--tie-breaker pairs, their temporal ordering is a uniform permutation. The action labels are therefore a uniform permutation of a population with equal numbers of $+1$ and $-1$, and at most one zero. Let $S_j$ be its partial sum. For fixed $j<T$, expose the first $s\le j$ labels. The Doob martingale is
$M_s=\E[S_j\mid\text{first }s\text{ labels}]=(T-j)S_s/(T-s)$.
Given the first $s-1$ labels, the range of $M_s$ has length at most two, because $(T-j)/(T-s)\le1$ and the next label lies in $[-1,1]$. A centered random variable with range length two has conditional moment-generating function at most $e^{\theta^2/2}$: this follows by twice differentiating its log moment-generating function, whose second derivative is a variance bounded by one. Iteration gives $\E e^{\theta S_j}\le e^{j\theta^2/2}$. Markov's inequality with $\theta=h/j$ gives $\Pp(|S_j|\ge h)\le2e^{-h^2/(2j)}$. Also $S_T=0$. A union bound gives
$\Pp(\max_{j\le T}|S_j|\ge\eta T)\le2T e^{-\eta^2T/2}$.
Off this event the rank prescription is storage-feasible, so the static and storage optima agree. The difference is always nonnegative and at most $T$, since all-hold is feasible and each reward is at most one. This proves Proposition~\ref{prop:static-interior}. The concentration argument applies to the hindsight ordering; the adaptive online inventory is treated separately in Appendix~\ref{app:interior-transfer}.
\end{proof}

\subsection{Concavity for every terminal target}

The concavity argument works pathwise. Starting from optimal plans at
two outer states, we average their inventory paths and
round in both directions. The resulting plans start at the middle state,
remain feasible, and preserve the sum of the original profits.

\begin{proof}[Proof of Proposition~\ref{prop:main-concavity}]
Fix the future prices and adjacent feasible states $j-1,j,j+1$. Choose optimal inventory paths $w^-_s,w^+_s$ from the outer states, both ending at $y_0$. Their midpoint $m_s=(w^-_s+w^+_s)/2$ starts at $j$, ends at $y_0$, lies in $[0,B]$, and satisfies $|m_s-m_{s-1}|\le1$.

Define the two integer paths $\underline w_s=\lfloor m_s\rfloor$ and $\overline w_s=\lceil m_s\rceil$. Both have the same initial and terminal states $j,y_0$, remain within capacity, and have successive differences in $\{-1,0,1\}$. For example, $m_s\le m_{s-1}+1$ implies $\lfloor m_s\rfloor\le\lfloor m_{s-1}\rfloor+1$, and the reverse inequality and ceiling argument are identical. Since $m_s$ is an integer or half-integer,
$\underline w_s+\overline w_s=w^-_s+w^+_s$.
Subtracting this identity at successive time steps shows that the sums of their trades agree at each price. Hence the sum of their profits equals the sum of the two outer optimal profits. Each rounded path is feasible from $j$, so
\begin{align*}
 2V_{t+1}^{\mathsf{hs}}(j;\mathbf p_{>t})
 \ge V_{t+1}^{\mathsf{hs}}(j-1;\mathbf p_{>t})
      +V_{t+1}^{\mathsf{hs}}(j+1;\mathbf p_{>t}).
\end{align*}
This proves pathwise discrete concavity. Taking expectations, or averaging the same inequality over the common scenarios, preserves it. Subtracting adjacent continuation values gives the ordering of marginal values. Comparing the three action advantages in Section~\ref{sec:stars} then shows that the lower threshold is no larger than the upper threshold, so the buy and sell advantages cannot both be positive. This gives the threshold rule in Proposition~\ref{prop:main-concavity}, with the specified tie rule.
\end{proof}

\begin{lemma}\label{lem:main-marginal-range}
Fix a terminal target, a future sequence in $[0,1]$, and adjacent feasible
continuation states $j-1,j$. Then $0\le\Delta V^{\mathsf{hs}}(j)\le1$.
\end{lemma}
\begin{proof}
Take an optimal plan from $j-1$. It cannot sell at every future time step, because the same target would then be unreachable from the also feasible state $j$. From $j$, follow those initial sales one unit higher until its first non-sale action $u\in\{0,1\}$; replace that action by $u-1$. The paths coalesce and can agree thereafter. The new plan is feasible and earns one additional price, which is nonnegative. This proves the lower bound. For the upper bound take an optimal plan from $j$, follow its initial purchases one unit lower, and replace its first non-purchase $u\in\{-1,0\}$ by $u+1$. That action exists by the same reachability argument. The paths coalesce, and the lower-start plan loses only that time step's price, at most one. Thus $V^{\mathsf{hs}}(j)-V^{\mathsf{hs}}(j-1)\le1$.
\end{proof}

\subsection{A positive compensation at zero, regardless of the future length}
\label{app:nonvanishing-proof}

At zero inventory, the policy must decide whether to buy before knowing
the next selling opportunity. Among all online decisions, holding
minimizes the expected compensation when the current price is sufficiently
high. An even higher next price can nevertheless make buying optimal in
hindsight. We use these two observations to obtain a positive compensation
bound for every feasible online policy.

% Statement relocated from the former Section 4.4; its proof is unchanged.
\begingroup

\begin{proposition}
\label{prop:main-nonvanishing}
Assume independent uniform prices on $[0,1]$ and terminal inventory
zero. Fix $B\ge1$ and any feasible online policy $\pi$. For every
$t<T$,
\begin{align*}
 \E[\partial R_t(y_{t-1}^\pi,x_t^\pi;\mathbf p_{\ge t})
 \mid\mathcal F_{t-1}]\ge\frac1{384}\one\{y_{t-1}^\pi=0\}.
\end{align*}
\end{proposition}
\par
\endgroup

\begin{proof}[Proof of Proposition~\ref{prop:main-nonvanishing}]
Fix the past $\mathcal F_{t-1}$ on the event
$y_{t-1}^{\pi}=0$. Let $\tau=\tau_t-1=T-t\ge1$ be the number of prices after the current decision,
let $U_1,\ldots,U_\tau$ be these independent uniform future prices, and put
$W=V_\tau^B(1;\mathbf U)-V_\tau^B(0;\mathbf U)$ and
$\lambda=\E[W\mid\mathcal F_{t-1}]$. The i.i.d.\ assumption makes the
conditional law of $\mathbf U$ independent of the past. We first show
\begin{align*}
 U_1\le W\le1,\qquad \lambda\le\log2<3/4.
\end{align*}
For the first inequality, take an optimal future plan from zero. If its
first action is to hold, a plan starting from one sells instead; if its
first action is to buy, a plan starting from one holds instead. In both
cases the two plans have the same inventory after this action, so the
new plan can follow the original continuation and remains feasible.
Its profit increases by $U_1$. The range upper bound follows from
Lemma~\ref{lem:main-marginal-range}.

For $a\in(1/2,1)$, call a future price cheap if it is at most $a$. Let $L_j,H_j$ count cheap and expensive prices in the first $j$ future time steps, and let $E_a=\{L_j\ge H_j\text{ for all }j\le \tau\}$. On this event, take an optimal plan from one and follow it from zero, one unit lower, until its first cheap time step at which it does not buy. This time step must exist: otherwise the original inventory would remain at least $1+L_j-H_j\ge1$ through the terminal time, contradicting its target zero. Before that time the lower path is feasible because the original inventory is at least one. At that time replace the original action $u\in\{-1,0\}$ by $u+1$. The paths coalesce at a cost at most $a$. Thus $W\le a$ on $E_a$; the capacity constraint is respected because the lower path never exceeds the original path.

The walk $H_j-L_j$ moves up with probability $1-a$ and down with probability $a$. Set $\theta=a/(1-a)>1$. Then $\theta^{H_j-L_j}$ is a nonnegative martingale. Stop it at its first hit of one or at $\tau$. Optional stopping at this bounded time gives $\Pp(E_a^c\mid\mathcal F_{t-1})\le1/\theta=(1-a)/a$. Consequently
\begin{align*}
 \lambda
 &\stackrel{(a)}{=}\int_0^1\Pp(W>a\mid\mathcal F_{t-1})\,da
 \stackrel{(b)}{\le}\frac12+\int_{1/2}^1\frac{1-a}{a}\,da
 \stackrel{(c)}{=}\log2.
\end{align*}
where (a) is the tail integral formula for $0\le W\le1$,
(b) uses $W\le a$ on $E_a$ and the bound
$\Pp(E_a^c\mid\mathcal F_{t-1})\le(1-a)/a$, and (c) evaluates the integral.
It remains to compare this unavoidable uncertainty with an arbitrary
online decision. At zero inventory the feasible current actions are
$x_t^\pi\in\{0,1\}$. Let $\zeta_t$ collect any fresh randomization used by
the policy; it is independent of the actual future prices, and
$x_t^\pi$ is determined by $\mathcal F_{t-1},p_t,\zeta_t$. Subtracting
the continuation value from zero from both action values gives the
pathwise identity
$\partial R_t(0,x_t^\pi;p_t,\mathbf U)
=(W-p_t)^+-x_t^\pi(W-p_t)$. Therefore,
\begin{align*}
 &\E[\partial R_t(0,x_t^\pi;p_t,\mathbf U)
          \mid\mathcal F_{t-1},p_t,\zeta_t]\\
 &\quad\stackrel{(a)}{=}
   \E[(W-p_t)^+\mid\mathcal F_{t-1},p_t]
       -x_t^\pi(\lambda-p_t)\\
 &\quad\stackrel{(b)}{\ge}
   \E[(W-p_t)^+\mid\mathcal F_{t-1},p_t]
       -(\lambda-p_t)^+.
\end{align*}
Here (a) uses nonanticipativity and independence of the future prices,
and (b) uses $x_t^\pi\in\{0,1\}$. The last expression is nonnegative
by convexity of the positive-part function. When $p_t\ge3/4$, we have
$\lambda<p_t$, so it is at least
$\E[(U_1-p_t)^+\mid\mathcal F_{t-1},p_t]$. Averaging over the current
price and the policy's randomization, conditional on the past, now gives
\begin{align*}
 \E[\partial R_t(0,x_t^\pi;p_t,\mathbf U)\mid\mathcal F_{t-1}]
 &\stackrel{(a)}{\ge}\int_{3/4}^1\int_p^1(u-p)\,du\,dp
 \stackrel{(b)}{=}\frac1{384}.
\end{align*}
Here (a) uses the preceding conditional bound and the independent
uniform laws of $p_t$ and $U_1$, and (b) integrates
$\int_p^1(u-p)\,du=(1-p)^2/2$.
The argument applies almost surely on the event of zero current
inventory, for every $\tau\ge1$ and $B\ge1$, and to every feasible online
policy. On the complementary event, nonnegative compensation gives the
claimed lower bound of zero.
\end{proof}

\subsection{Frequent visits with zero cost for two prices}

A two-price law illustrates how frequent boundary visits can have
zero cost. The reference rule visits zero repeatedly and pays zero
compensation. The following example explains why we must account for
both visits and their costs.

\Needspace{10\baselineskip}
\begin{lemma}\label{prop:binary-free}
Let $\Pp(p=a)=\Pp(p=b)=1/2$ with $0\le a<b\le1$, $y_0=0$, and
$B=\lceil\rho T\rceil$. The auxiliary reference rule $\pi^\circ$ satisfies:
\begin{enumerate}[label=$(\alph*)$]
\item For every $T$, $\Reg(\pi^\circ;B,T,0)=0$.
\item There exist constants $c_\rho>0$ and $T_0(\rho)<\infty$, depending
only on $\rho$, such that the expected number of visits to zero is at
least $c_\rho\sqrt T$ for all $T\ge T_0(\rho)$.
\end{enumerate}
\end{lemma}
\begin{proof}
The marginal recursion of Lemma~\ref{lem:stars-concavity} preserves the price support. Every real marginal $W$ therefore belongs to $\{a,b\}$. At $A_i=a$, $(a-W)^+=0$; at $A_i=b$, $(W-b)^+=0$. Equation~\eqref{eq:main-atomic-hinge} is zero for every existing channel. Forced actions also have zero compensation, so the exact regret decomposition proves zero regret.

For every real index $y\le \tau\wedge B$ with $\tau\ge1$, the events that all future prices equal $a$ or all equal $b$ have positive probability. On the all-$c$ event, every feasible continuation earns $cy$ by terminal balance, so its marginal equals $c$. Hence $a<\lambda_\tau^B(y)<b$. Whenever not constrained by the upper or terminal boundary, the reference rule buys at $a$, sells at $b$ if positive inventory is available, and holds at $b$ when inventory is zero.

Take $m=\lfloor cT\rfloor$ with $c=(\rho/2)\wedge(1/4)$. In the first $m$ time steps, an inventory starting from zero is at most $m<B$ and the terminal ceiling is at least $T-m>m$. Thus the actual process there is exactly the lower-reflected fair walk. At time $j$, its distribution equals the running maximum of an ordinary fair walk $S$ of length $j$, by reversing its i.i.d.\ increments. Reflection at the first hit of one yields
\begin{align*}
 \Pp(y_j=0)
 &\stackrel{(a)}{=}\Pp(\max_{i\le j}S_i<1)
 \stackrel{(b)}{=}\Pp(-1\le S_j<1)
 \stackrel{(c)}{\ge}c_0(j+1)^{-1/2}.
\end{align*}
where (a) reverses the i.i.d.\ increments of the lower-reflected walk,
(b) reflects at the first hit of one, and (c) uses the
central-binomial estimate in Lemma~\ref{lem:stars-fair}.
Summing for $0\le j<m$ gives at least $c_\rho\sqrt T$ visits. Each visit has zero expected compensation, giving zero regret.
\end{proof}

\subsection{Why fixed simulation counts suffice in the uniform special case}

For uniform prices, the sampled policy and the reference rule have
the same inventory transition probabilities. We use their common
state distributions to compare the one-step compensation and obtain
an exact relation between their regrets.

\begin{proposition}\label{prop:uniform-finite-K}
Assume uniform prices and zero initial and terminal inventory, and fix an
integer $K\ge1$.
\begin{enumerate}[label=$(\alph*)$]
\item Algorithm~\ref{alg:stars} has the same inventory transition kernel
as the reference rule $\pi^\circ$.
\item $\Reg(\textup{\STARS}_K;B,T,0)=(1+1/K)\Reg(\pi^\circ;B,T,0)$.
\end{enumerate}
\end{proposition}
\begin{proof}
Fix a time and feasible inventory. The shared scenarios keep empirical marginal thresholds ordered, and Lemma~\ref{lem:main-marginal-range} puts them in $[0,1]$. For each existing buy threshold, integrating the independent current uniform price gives purchase probability $\E_{\zeta_t}\widehat\lambda=\lambda$. The sell probability is similarly $1-\lambda$. Missing actions have probability zero and forced actions agree. Thus the two policies have the same time-dependent transition kernel at every state. They start from the same inventory, so their marginal state distributions agree at every time step.

For one channel, freeze the actual and simulated futures. Its current-price-integrated compensation is $(W-\widehat\lambda)^2/2$ by \eqref{eq:main-area}. The actual future is independent of the $K$ simulated futures, so
\begin{align*}
 \E_{\mathbf p_{>t},\zeta_t}(W-\widehat\lambda)^2
 &\stackrel{(a)}{=}\Var_{\mathbf p_{>t}}(W)+\Var_{\zeta_t}(\widehat\lambda)
 \stackrel{(b)}{=}(1+1/K)\Var_{\mathbf p_{>t}}(W).
\end{align*}
where (a) uses independence and the unbiasedness
$\E_{\zeta_t}\widehat\lambda=\E_{\mathbf p_{>t}}W$,
and (b) uses the variance of the mean of $K$ i.i.d.\ copies.
Add the existing channels, multiply by the common state probabilities, and sum over time using \eqref{eq:main-weighted-occupation} and the corresponding reference identity. This proves the identity.
\end{proof}

\section{\textsf{STARS} and the marginal compensation identity}
\label{app:stars-foundation}\label{app:stars-compensation}

This appendix begins the proof of Theorem~\ref{thm:stars-upper}.
Throughout, prices are i.i.d.\ with law $F$ on $[0,1]$,
$y_0=y_T=0$, and $B=\lceil\rho T\rceil$ for fixed $0<\rho<1$.
Our first step is to express the loss from a decision in terms of the
hindsight buy and sell thresholds. We begin with the
population-threshold reference rule $\pi^\circ$ from
Section~\ref{sec:stars}, which replaces the scenario averages by their
expectations. This gives a compensation bound as a function of the
current inventory. We then account for the finite scenarios used by
\textsf{STARS}: Appendix~\ref{app:stars-sampling} transfers the continuous
reference bounds, and Appendix~\ref{app:finite-atomic} bounds atomic
compensation directly along the sampled inventory path.

We retain the paper's time-indexed notation for the realized path.
When analyzing a hypothetical future, we instead index its value by
the number of remaining prices. Thus
\begin{align*}
 V_\tau^B(y;\mathbf p)
 &=V_{T-\tau+1}^{\mathsf{hs}}(y;\mathbf p),\qquad
 \Delta V_\tau^B(y;\mathbf p)
 =\Delta V_{T-\tau+1}^{\mathsf{hs}}(y;\mathbf p),\\
 \lambda_\tau^B(y)&=\lambda_{T-\tau+1}(y).
\end{align*}
Both sides describe the same $\tau$ remaining prices, the same capacity,
and the same zero terminal inventory. The superscript $B$ only makes
the capacity visible in the auxiliary notation. At decision epoch $t$,
write $\tau_t=T-t+1$ for the number of remaining decisions, including
the current one. Its continuation therefore has $\tau_t-1=T-t$
prices, and $\lambda_{\tau_t-1}^B(y)=\lambda_{t+1}(y)$.

\subsection{The optimization \textsf{STARS} performs}
At decision epoch $t$, \textsf{STARS} observes inventory $y_{t-1}$ and price $p_t$.
An action $x_t\in\{-1,0,1\}$ means sell, hold, or buy, respectively;
reward is $-p_tx_t$ and $y_t=y_{t-1}+x_t$. There are $\tau_t-1$
future prices after this action. Use the observed-history filtration
$\mathcal F_t$ from Section~\ref{sec:model}, including the scenarios
drawn through epoch $t$. Sample counts $K_t$ are
$\mathcal F_{t-1}$-measurable and at least $\tau_t^2$.
Feasible current actions specialize to
$\mathcal A_t(y)=\{x\in\{-1,0,1\}:0\le y+x\le B\wedge(\tau_t-1)\}$.
The constraint $y+x\le\tau_t-1$ ensures that inventory can be reduced to zero
by epoch $T$; it is as important as the physical capacity.

For a specified $\tau$-price future $\mathbf p$, let $V_\tau^B(y;\mathbf p)$
be the maximum reward starting with $y$ units and ending at zero.
Its domain is $0\le y\le B\wedge \tau$, with $V_0^B(0)=0$.
At epoch $t$, choosing the current action and optimizing its continuation gives
\begin{align*}
 V_{\tau_t}^B(y;p,\mathbf p)
 =\max_{x\in\mathcal A_t(y)}\{-px+V_{\tau_t-1}^B(y+x;\mathbf p)\}.
\end{align*}
\textsf{STARS} draws
$K$ independent i.i.d.\ future vectors $\mathbf p^{(1)},\ldots,\mathbf p^{(K)}$,
uses the \emph{same} vectors for all actions, and maximizes
$-p_t x+K^{-1}\sum_{j=1}^K V_{\tau_t-1}^B(y+x;\mathbf p^{(j)})$
over $x\in\mathcal A_t(y)$.
Only the current action is implemented; the procedure repeats at the
next epoch. The auxiliary rule $\pi^\circ$ uses $\E V_{\tau_t-1}^B(y+x)$ instead.
We prefer holding whenever it is a feasible maximizer. The future
scenarios are independent of actual prices and are drawn from the
known price distribution.

We compare with the feasible hindsight optimum, writing
$\Reg_T(\pi):=\Reg(\pi;B,T,0)=\E V_T^B(0;\mathbf p)
-\E\sum_{t=1}^T(-p_t x_t^\pi)$.
Only in the reference calculations of this appendix and
Appendix~\ref{app:stars-continuous}, $\Reg_T$ without a policy means
$\Reg_T(\pi^\circ)$. We identify sampled trajectories explicitly.

\subsection{Concavity gives the buy and sell thresholds}
An extra stored unit changes future reward by
$\Delta V_\tau^B(y)=V_\tau^B(y)-V_\tau^B(y-1)$, whose mean is
$\lambda_\tau^B(y)=\E\Delta V_\tau^B(y)$, for $1\le y\le B\wedge \tau$.
To use these differences as buy and sell thresholds, we first show
that an additional unit is less valuable when more inventory is
already available. The following induction proves this decreasing
marginal-value property. It also gives a recursion that will allow us
to characterize the distribution of the thresholds.

\Needspace{10\baselineskip}
\begin{lemma}\label{lem:stars-concavity}\label{lem:stars-queue}
Fix a future price vector in $[0,1]^\tau$ and a capacity $B$.
\begin{enumerate}[label=$(\alph*)$]
\item $V_\tau^B$ is discretely concave, its real marginals form a
decreasing list, and each marginal is one of the future prices.
Thus the marginals lie in $[0,1]$, and $\E V_\tau^B$ is also concave.
\item Prepending price $p$ updates every new real marginal by
\begin{align}
 \widetilde{\Delta V}(y)=\Delta V_\tau^B(y+1)\vee\bigl(p\wedge\Delta V_\tau^B(y-1)\bigr).
 \label{eq:stars-minmax}
\end{align}
For this formula only, use ordering conventions $\Delta V_\tau^B(y)=+\infty$
for $y\le0$ and $\Delta V_\tau^B(y)=-\infty$ beyond the old domain. They are
not random values whose expectations are taken.
\end{enumerate}
\end{lemma}
\begin{proof}
Induct on the number of future prices. The zero-price domain is a
singleton; after one price the first marginal, if present, is $p$.
Suppose the old list decreases. At an inventory where the hold action
is feasible, Bellman's maximization improves the hold reward by
$(p-\Delta V(y))^++(\Delta V(y+1)-p)^+$, omitting missing actions; at most one term
is positive. Here and in the next display, suppress the old horizon
and write $\Delta V_y=\Delta V_\tau^B(y)$. Subtract the new values at
adjacent old states:
\begin{align*}
 \widetilde{\Delta V}(y)
 &\stackrel{(a)}{=}\Delta V_y+(p-\Delta V_y)^++(\Delta V_{y+1}-p)^+
      -(p-\Delta V_{y-1})^+-(\Delta V_y-p)^+\\*
 &\stackrel{(b)}{=}p+(\Delta V_{y+1}-p)^+-(p-\Delta V_{y-1})^+\\*
 &\stackrel{(c)}{=}\Delta V_{y+1}\vee\bigl(p\wedge\Delta V_{y-1}\bigr).
\end{align*}
where (a) subtracts the two Bellman improvements over holding,
(b) uses $a+(p-a)^+-(a-p)^+=p$, and (c) uses
$\Delta V_{y+1}\le \Delta V_{y-1}$. At the floor the update is $p\vee\Delta V_2$;
a newly feasible top state $y=\tau+1$ must sell and gives
$p\wedge\Delta V_\tau^B(\tau)$.
At the physical cap the unavailable purchase is omitted. The first
marginal from the singleton case is $p$. These checks establish
\eqref{eq:stars-minmax} at every boundary without subtracting infinite values.

Both inputs $\Delta V(y+1),\Delta V(y-1)$ decrease in $y$, so the updated list also
decreases. Each new entry is an old entry or $p$. This closes the
induction, proving support-valued marginals and discrete concavity.
Expectation preserves the concavity inequalities, and its first
differences are $\lambda_\tau^B(y)$ by linearity.
\end{proof}

Relative to holding, the reference rule's buying advantage is
$\lambda_{\tau_t-1}^B(y+1)-p$, its selling advantage is
$p-\lambda_{\tau_t-1}^B(y)$, and holding has advantage zero.
It therefore buys below the first threshold, sells above the second,
and holds between them. Thus the hold interval arises endogenously
from the continuation values, without a distribution-specific holding
rule. Missing actions are omitted.
If $y=\tau_t$, selling is compulsory. At the last epoch, only the action
reaching zero is feasible.

\subsection{Compensation arises between the corresponding thresholds}
Fix \textsf{STARS}'s actual current inventory $y$ and the realized future. Let
hindsight optimize from this \emph{same} $y$. Its loss from being
forced to take \textsf{STARS}'s current action $x_t$ is the paper's realized
marginal compensation, evaluated on that path:
\begin{align*}
 \partial R_t=V_{\tau_t}^B(y;p_t,\mathbf p_{>t})
       +p_tx_t-V_{\tau_t-1}^B(y+x_t;\mathbf p_{>t})\ge0.
\end{align*}
Summing along \textsf{STARS} cancels successive continuation values:
\begin{align}
 V_T^B(0;\mathbf p)-\sum_{t=1}^T(-p_tx_t)=\sum_{t=1}^T\partial R_t,
 \label{eq:stars-telescoping}
\end{align}
because the last continuation is $V_0^B(0)=0$. At each epoch we
compare the action chosen by \textsf{STARS} with the hindsight-optimal
action, both selected from the feasible buy, hold, and sell choices
at \textsf{STARS}'s current inventory.

When the hold action is feasible, put $W_b=\Delta V_{\tau_t-1}^B(y+1)$ and
$W_s=\Delta V_{\tau_t-1}^B(y)$
for the existing channels. Let $B^{\mathrm H},S^{\mathrm H}$ and
$B^{\mathrm S},S^{\mathrm S}$ indicate hindsight and \textsf{STARS} buying/selling.
The exact pathwise charge is
\begin{align}
 \partial R_t=|p_t-W_b|\one\{B^{\mathrm H}\ne B^{\mathrm S}\}
          +|p_t-W_s|\one\{S^{\mathrm H}\ne S^{\mathrm S}\}.
 \label{eq:stars-disagreement}
\end{align}
Indeed an action's advantage over holding is
$(W_b-p_t)B+(p_t-W_s)S$. Hindsight chooses each positive coefficient;
subtracting \textsf{STARS}'s chosen coefficient gives \eqref{eq:stars-disagreement}.
A buy-versus-sell disagreement has two terms because it changes the
residual inventory by two units. Zero coefficients cost nothing.

For $\pi^\circ$, a buy charge occurs only between $W_b$ and $\E W_b$, and a sell
charge only between $W_s$ and $\E W_s$. With holding at equality, the
buy interval is $[W_b\wedge\E W_b,W_b\vee\E W_b)$,
and the sell interval is $(W_s\wedge\E W_s,W_s\vee\E W_s]$.
The endpoint conventions matter at atoms of the \emph{mean} threshold;
at the random threshold the price distance is zero.

\begin{proposition}\label{prop:stars-compensation}
Define $H_F(z)=\E_p(z-p)^+$ and
$J_F(W)=\E H_F(W)-H_F(\E W)$. Write
$\partial\bar R_\tau^B(y)$ for the expected value of $\partial R_t$
when $\pi^\circ$ starts a decision at $y$, there are $\tau$
prices after that decision, and the independent current price and future are
averaged out. This is the conditional mean of the realized
compensation:
$\E[\partial R_t(y_{t-1}^\circ,x_t^\circ;\mathbf p_{\ge t})
\mid\mathcal F_{t-1}]=\partial\bar R_{\tau_t-1}^B(y_{t-1}^\circ)$.
The same fixed-state quantity is defined at any feasible $y$, regardless
of which policy visits it.
\begin{enumerate}[label=$(\alph*)$]
\item When the hold action is feasible,
\begin{align*}
 \partial\bar R_\tau^B(y)=J_F(\Delta V_\tau^B(y+1))+J_F(\Delta V_\tau^B(y)),
\end{align*}
with unavailable channels omitted. Forced actions cost zero.
\item The reference rule satisfies
$\Reg_T(\pi^\circ)=\sum_{t=1}^T\E \partial\bar R_{\tau_t-1}^B(y_{t-1}^\circ)$.
\end{enumerate}
\end{proposition}
\begin{proof}
For a buy marginal $W$ with mean $\lambda$, the reference action's realized contribution
is $(W-p)\one\{p<\lambda\}$. Conditional on $p$, its expectation over the
independent future is $(\lambda-p)^+$. Hindsight's contribution is
$(W-p)^+$. Averaging the difference over $p$ gives $J_F(W)$.
The equality concerns expected contributions after averaging over
the current price and the future. Selling gives the same Jensen difference
because $\E_p(p-z)^+=H_F(z)+\E p-z$ and affine terms cancel.
Finally $y_{t-1}$ is measurable with respect to past prices, while
current and future prices remain i.i.d.\ and independent of that past.
Conditional expectation and the tower property in
\eqref{eq:stars-telescoping} prove the stated expected-regret identity.
\end{proof}

We next specialize this identity to the price laws considered in the
paper. These calculations explain which threshold errors are costly.
For uniforms, $H_F(z)=z^2/2$, so $J_F(W)=\Var(W)/2$.
For a finite law with masses $w_i$ at $z_i$,
\begin{align*}
 J_F(W)=\sum_i w_i\bigl(\E(W-z_i)^+\wedge\E(z_i-W)^+\bigr).
\end{align*}
To check this, set the two positive parts' expectations to $A,D$;
the summand is $A-(A-D)^+=A\wedge D$.
For continuous prices with median interval $[a^-,a^+]$, subtract
the affine function $\ell(z)=H_F(a^-)+(z-a^-)/2$. Convexity gives
$R_F=H_F-\ell\ge0$, zero on the median interval, and
\begin{align*}
 J_F(W)=\E R_F(W)-R_F(\E W)\le\E R_F(W).
\end{align*}
When the interval is a single point, $\ell$ is the tangent there.
Uncertainty within a price gap causes no charge because no price can
arrive in that interval. More generally, the charge depends on both
the probability of a disagreement and the price difference when it
occurs.

% Detailed compensation statements retained after shortening Section 4.4.
% Technical statements relocated from the former Section 4.4.
% Their original labels are retained for existing manuscript references.
\begingroup

\subsection{Expected compensation and inventory occupation}
\label{app:compensation-estimates}

Fix the past $\mathcal F_{t-1}$, so the current inventory $y$ and
scenario count $K=K_t$ are known. Write $\tau_t=T-t+1$ for the
number of decisions remaining, including the current decision.
The continuation after this decision has $\tau_t-1$ prices.
At a state where holding is feasible, hindsight's thresholds are
$b^{\mathsf{hs}}=\Delta V_{t+1}^{\mathsf{hs}}(y+1;\mathbf p_{>t})$
for buying and
$s^{\mathsf{hs}}=\Delta V_{t+1}^{\mathsf{hs}}(y;\mathbf p_{>t})$
for selling. \STARS{} uses their sample averages
$\widehat b=\widehat\lambda_{t+1}(y+1)$ and
$\widehat s=\widehat\lambda_{t+1}(y)$.
Throughout, a term corresponding to an unavailable action is omitted.
If holding is infeasible, the action is forced and its compensation is zero.

For each buy or sell threshold, compensation is the distance from
the price to hindsight's threshold whenever the two decisions disagree.

We average in two stages. First fix the actual future
$\mathbf p_{>t}$ and the fresh scenario randomness $\zeta_t$,
in addition to the past. All four thresholds are then fixed; only
the independent current price is averaged. For a density $f$,
\begin{align}
 \E\left[\partial R_t(y,x_t^K(y,p_t;\zeta_t);\mathbf p_{\ge t})
       \mid\mathcal F_{t-1},\mathbf p_{>t},\zeta_t\right]=
 \int_{b^{\mathsf{hs}}\wedge\widehat b}^{b^{\mathsf{hs}}\vee\widehat b}
                  |b^{\mathsf{hs}}-p|f(p)\,dp
 +\int_{s^{\mathsf{hs}}\wedge\widehat s}^{s^{\mathsf{hs}}\vee\widehat s}
                  |s^{\mathsf{hs}}-p|f(p)\,dp.
 \label{eq:main-area}
\end{align}
For a general law, use $F(dp)$ in place of $f(p)\,dp$, with
the buying interval $[b^{\mathsf{hs}}\wedge\widehat b,
b^{\mathsf{hs}}\vee\widehat b)$ and the selling interval
$(s^{\mathsf{hs}}\wedge\widehat s,s^{\mathsf{hs}}\vee\widehat s]$.
These conventions retain the rule of holding at an estimated threshold.
A tie at a hindsight threshold has zero loss.

Second, average these areas over the independent actual and simulated
futures. The result is
$\E[\partial R_t(y,x_t^K;\mathbf p_{\ge t})\mid\mathcal F_{t-1}]$,
as required by Proposition~\ref{prop:main-decomposition}.

For the calculations below, denote the mean thresholds by
$b^\circ=\E[b^{\mathsf{hs}}\mid\mathcal F_{t-1}]=\lambda_{t+1}(y+1)$
and
$s^\circ=\E[s^{\mathsf{hs}}\mid\mathcal F_{t-1}]=\lambda_{t+1}(y)$.
These are the thresholds of the mean-threshold rule $\pi^\circ$
from Section~\ref{sec:proof-ideas}, evaluated at this same inventory.

\Needspace{10\baselineskip}
\begin{proposition}
\label{prop:main-local-compensation}
Condition on $\mathcal F_{t-1}$, let $y=y_{t-1}^{\pi^K}$, and
suppose the hold action is feasible at $y$.
Use the mean thresholds $b^\circ$ and $s^\circ$ defined above,
and let $K=K_t\ge1$.
\begin{enumerate}[label=$(\alph*)$]
\item For a finite atomic law with support $A_1<\cdots<A_m$ and
masses $w_1,\ldots,w_m$, the mean-threshold rule $\pi^\circ$
has conditional expected compensation
\begin{align}
 \E[\partial R_t(y,x_t^\circ(y,p_t);\mathbf p_{\ge t})
                                 \mid\mathcal F_{t-1}] &=\sum_{i=1}^m w_i\Big\{
    \big(\E[(b^{\mathsf{hs}}-A_i)^+\mid\mathcal F_{t-1}]
           \wedge\E[(A_i-b^{\mathsf{hs}})^+\mid\mathcal F_{t-1}]\big) \notag\\
   &\qquad + \big(\E[(s^{\mathsf{hs}}-A_i)^+\mid\mathcal F_{t-1}]
           \wedge\E[(A_i-s^{\mathsf{hs}})^+\mid\mathcal F_{t-1}]\big)
                    \Big\}.
 \label{eq:main-atomic-hinge}
\end{align}
With $K$ simulated futures, \textup{\STARS{}} satisfies
\begin{align*}
 \E[\partial R_t(y,x_t^K(y,p_t;\zeta_t);\mathbf p_{\ge t})
                                     \mid\mathcal F_{t-1}] \le2\E[\partial R_t(y,x_t^\circ(y,p_t);\mathbf p_{\ge t})
                                     \mid\mathcal F_{t-1}]
                +\frac{40}{3eK}.
\end{align*}
\item For uniform prices on $[0,1]$, the expected compensation of
\textup{\STARS{}} is exactly
\begin{align*}
 \E[\partial R_t(y,x_t^K;\mathbf p_{\ge t})\mid\mathcal F_{t-1}] =\frac{1+1/K}{2}
 \left\{\Var(b^{\mathsf{hs}}\mid\mathcal F_{t-1})
             +\Var(s^{\mathsf{hs}}\mid\mathcal F_{t-1})\right\}.
\end{align*}
\item Fix $\rho\in(0,1)$, let $B=\lceil\rho T\rceil$, and take
$F\in\mathcal F^{\mathrm{ac}}_{\beta,\Delta,\varepsilon_0}$
from Definition~\ref{def:continuous-family}, allowing $\Delta=0$.
Set $\alpha=1/(\beta+1)$ and $\tau_t=T-t+1$. There are constants
$C,c_\rho>0$, depending on $F,\rho$ but not $T$, such that for
$K_t\ge\tau_t^2$,
\begin{align}
 \E[\partial R_t(y,x_t^K;\mathbf p_{\ge t})\mid\mathcal F_{t-1}] \le C\left\{ \left(\frac{1}{y + 1} \right)^{1 + \alpha}
   +\left(\frac{y+1}{\tau_t}\right)^{1+\alpha}
   +\tau_t^{-\frac{1+\alpha}{2}}
   +\epsT+\mathbbm{1}\Big\{y>\frac{B}{2}\Big\}\right\},
 \label{eq:main-continuous-charge}
\end{align}
where $\epsT=CT^2e^{-c_\rho T}$. The same bound holds for the
reference action $x_t^\circ(y,p_t)$.
\end{enumerate}
\end{proposition}

For atomic prices, condition on $p_t=A_i$. If the mean-threshold rule
buys, its expected buying loss is
$\E[(A_i-b^{\mathsf{hs}})^+\mid\mathcal F_{t-1}]$; otherwise it is
$\E[(b^{\mathsf{hs}}-A_i)^+\mid\mathcal F_{t-1}]$.
The mean threshold selects the smaller of these quantities.
Selling gives the other minimum in \eqref{eq:main-atomic-hinge}.
Lemma~\ref{lem:sampled-compensation} bounds the additional sampling
loss at this same inventory.

For uniform prices, each integral in \eqref{eq:main-area} is a
triangle: the buying loss is $(b^{\mathsf{hs}}-\widehat b)^2/2$.
Conditional independence and equal means give
$\E[(b^{\mathsf{hs}}-\widehat b)^2\mid\mathcal F_{t-1}]
=(1+1/K)\Var(b^{\mathsf{hs}}\mid\mathcal F_{t-1})$;
selling is identical. This explains the variance connection with
\cite{bray2025}. Here the random threshold is the value of a stored
unit subject to all future inventory constraints.

For continuous prices, Step 2 of Appendix~\ref{app:stars-continuous}
combines the occupied-support area bound with the threshold moments
in Lemma~\ref{lem:stars-uniformmoments} to prove
\eqref{eq:stars-sc-compensation}. Lemma~\ref{lem:sampled-compensation}
then gives part (c).

For scenario counts $K_1,\ldots,K_T$
chosen in advance, sum the decomposition over feasible inventories
$\mathcal Y_t=\{0,\ldots,B\wedge(T-t+1)\}$:
\begin{align}
 \Reg(\pi^K;B,T,0)
 =\sum_{t=1}^T\sum_{y\in\mathcal Y_t}
    \Pp(y_{t-1}^{\pi^K}=y)\,
    \E_{p_t,\mathbf p_{>t},\zeta_t}^{(K_t)}
    [\partial R_t(y,x_t^{K_t}(y,p_t;\zeta_t);\mathbf p_{\ge t})].
 \label{eq:main-weighted-occupation}
\end{align}
The inner expectation fixes $y$ and averages fresh prices and
scenarios. This is \emph{compensation-weighted occupation}.
\par
\endgroup

\section{A random-walk representation of hindsight thresholds}
\label{app:stars-scalar}

The compensation identity reduces the next step to a distributional
question: what is the law of $\Delta V_\tau^B(y)$? We answer it by fixing
a price level $u$ and counting how many marginals exceed that price.
Precisely, $C_\tau^B(u)=\#\{y:1\le y\le B\wedge \tau,\ \Delta V_\tau^B(y)>u\}$.
The decreasing list and \eqref{eq:stars-minmax} imply
\begin{align*}
 C_0^B(u)=0,\qquad
 C_{j+1}^B(u)=\clip_{[0,B]}(C_j^B(u)+\xi_{j+1}),\qquad
 \Pp(\xi_{j+1}=1)=1-F(u).
\end{align*}
The signs are i.i.d.\ and prices are processed backward. A price above
$u$ increases the cutoff by one; a price at or below $u$ decreases it.
Thus $\{\Delta V_\tau^B(y)>u\}=\{C_\tau^B(u)\ge y\}$, including atomic prices.
The count describes the hindsight problem on a hypothetical future;
we will use its law to analyze the thresholds that determine \textsf{STARS}'s
inventory movements.

We can also express the thresholds for every price law in uniform
quantile coordinates. Set $p_j=F^{-1}(U_j)$, where
$F^{-1}(u)=\inf\{p:F(p)\ge u\}$ for $0<u<1$, with endpoint values
defined by one-sided limits, is the quantile and $U_j$ are i.i.d.\ uniform.
Minimum and maximum
commute with a nondecreasing function, so induction in
\eqref{eq:stars-minmax} gives $\Delta V_\tau^B(y)=F^{-1}(d_\tau^B(y))$,
where $d_\tau^B(y)$ is the marginal for the uniform future. It is one of
the continuous observations $U_j$, hence has no atom at a fixed
quantile or a discontinuity of $F^{-1}$.

\subsection{Hitting and survival estimates}
Remove the physical upper cap and omit the superscript $B$ for those
objects. For an ordinary walk $S_j$ started at zero, write
$\sigma_y=\inf\{j:S_j=y\}$ and $\Pp_a$ when its upward probability
is $a$. Real marginal indices still satisfy $1\le y\le \tau$.

\begin{lemma}\label{lem:stars-hitting}
Fix integers $1\le y\le \tau$.
\begin{enumerate}[label=$(\alph*)$]
\item For every price level $u$,
$\Pp(\Delta V_\tau(y)>u)=\Pp_{1-F(u)}(\sigma_y\le \tau)$.
\item For $0<a<1/2$,
\begin{align}
 \Pp_a(\sigma_y\le \tau)
 &=\left(\frac a{1-a}\right)^y\Pp_{1-a}(\sigma_y\le \tau)
 \le\left(\frac a{1-a}\right)^y.
 \label{eq:stars-likelihood}
\end{align}
\end{enumerate}
\end{lemma}
\begin{proof}
The lower-reflected recursion is
$C_\tau=\max_{0\le j\le \tau}(S_\tau-S_j)$. Reversing i.i.d.\ increments
identifies its law with $\max_{j\le \tau}S_j$, giving
the first identity in the lemma. A path first hitting $y$ has $y$ more upward than
downward steps. Its likelihood ratio under $a$ versus $1-a$ is
$(a/(1-a))^y$. Sum these ratios over all first-hit paths by epoch $\tau$
to obtain \eqref{eq:stars-likelihood}.
\end{proof}

\begin{lemma}\label{lem:stars-fair}
Fix integers $\tau,y\ge1$.
\begin{enumerate}[label=$(\alph*)$]
\item With absolute comparison constants $c,C>0$,
\begin{align*}
 \Pp_{1/2}(\sigma_y>\tau)&=\Pp(-y\le S_\tau<y)
 \asymp \frac{y}{\sqrt{\tau+1}}\wedge1.
\end{align*}
\item For every fixed sign probability and $x\ge0$,
\begin{align}
 \Pp\left(\max_{j\le \tau}(S_j-\E S_j)\ge x\right)
 &\le e^{-x^2/(2\tau)}.\label{eq:stars-hoeffding}
\end{align}
The same inequality holds for the negative centered walk.
\item If $a>1/2$ and $y\le(2a-1)\tau$, then
\begin{align}
 \Pp_a(\sigma_y>\tau)
 \le e^{-((2a-1)\tau-y)^2/(2\tau)}.\label{eq:stars-failure}
\end{align}
\end{enumerate}
\end{lemma}
\begin{proof}
Reflection after the first hit of $y$ gives
$\Pp(\sigma_y\le \tau)=\Pp(S_\tau\ge y)+\Pp(S_\tau>y)$; symmetry proves
the claimed fair-survival identity. A central binomial mass is of order
$(\tau+1)^{-1/2}$. One direct verification at $\tau=2m$ is
$4^{-m}\binom{2m}{m}=\prod_{j=1}^m(1-1/(2j))$: take logarithms,
use $-x-x^2\le\log(1-x)\le-x$ for $x\le1/2$, and sum. At odd
times multiply the previous central mass by $(2m+1)/(2m+2)$.
There are $y$ parity points in $[-y,y)$ unless the support cuts them
off, giving the upper bound. For $y\le\sqrt \tau/4$, all their masses
are comparable to the central mass: their ratios are products of
$(m-i+1)/(m+i)$ at even times, or $(m-i+1)/(m+i+1)$ at odd times,
with at most $\sqrt \tau/4$ factors. Their log ratios are bounded below
by $-C\sum_{i\le\sqrt \tau/4}i/\tau\ge-C'$. This gives the lower bound
in this range; monotonicity in $y$ extends it to the remaining range.
Finitely many small $\tau$ are absorbed in the constants.

For a sign $X$, the second derivative of its log moment-generating
function is a tilted variance, at most one. Integrating twice yields
$\E e^{\lambda(X-\E X)}\le e^{\lambda^2/2}$. Stop the resulting
nonnegative exponential supermartingale at the first centered crossing
of $x$ or epoch $\tau$. Its value at a crossing is at least
$e^{\lambda x-\tau\lambda^2/2}$, so the crossing probability is at most
$e^{-\lambda x+\tau\lambda^2/2}$. Taking $\lambda=x/\tau$ proves
\eqref{eq:stars-hoeffding}. Negating the signs gives the lower tail;
$\{\sigma_y>\tau\}\subset\{S_\tau<y\}$ then gives \eqref{eq:stars-failure}.
\end{proof}

\subsection{Capacity and uniform-coordinate moments}
The walk is simpler to analyze without the physical capacity. We
therefore first bound how much removing that capacity changes a
hindsight threshold. This is a comparison of hypothetical futures;
the probability that the online inventory reaches the ceiling will
be bounded separately.

\Needspace{8\baselineskip}
\begin{lemma}\label{lem:stars-capacity}
Fix $\rho\in(0,1)$ and $B=\lceil\rho T\rceil$. There are constants
$c_\rho,C_\rho>0$ such that, with
$\epsT=C_\rho T^2e^{-c_\rho T}$, the following hold uniformly over
price laws, price levels, and $\tau\le T$.
\begin{enumerate}[label=$(\alph*)$]
\item For real indices $y\le \tau\wedge(B/2+1)$,
\begin{align*}
 0\le\Pp(\Delta V_\tau(y)>u)-\Pp(\Delta V_\tau^B(y)>u)
 \le\epsT.
\end{align*}
\item In the same index range, means and hinge expectations change
by at most $\epsT$, and expectations of a bounded-variation function
change by at most its total variation times $\epsT$.
\item Capping lowers every real marginal stochastically, including
outside the range in part~(a).
\end{enumerate}
\end{lemma}
\begin{proof}
Use the same signs in both count recursions; the capped count is
smaller. With nonpositive drift, a discrepancy requires an upward
excursion of height $B$ over some interval. With nonnegative drift,
differing tail events at $y$ require a descent from $B$ to below $y$:
between the last visit to $B$ and the first subsequent visit below
$y$ there is no reflection, and the interval sum is at most $-(B-y)$.
After centering, both events are at least as unlikely as the
corresponding zero-drift deviation. There are at most $(T+1)^2$
intervals, each of length at most $T$. Since $B-y\ge\rho T/4$ for
sufficiently large $T$, \eqref{eq:stars-hoeffding} and a union bound give
the claimed capacity-error bound. Enlarge its constant for smaller $T$.
Integrate the tail difference for means/hinges; Stieltjes integration
gives the bounded-variation assertion.
\end{proof}

\begin{lemma}\label{lem:stars-uniformmoments}
Fix $h>0$. There is a constant $C_h<\infty$ such that, for all
$\tau\ge1$ and $1\le y\le \tau$,
\begin{align*}
 \E|d_\tau(y)-1/2|^h\le C_h\{y^{-h}+(y/(\tau+1))^h+(\tau+1)^{-h/2}\}.
\end{align*}
For $d_\tau^B(y)$ at $y\le B/2+1$, add $C_h\epsT$.
\end{lemma}
\begin{proof}
For $0<z<1/2$, \eqref{eq:stars-likelihood} gives
$\Pp(d_\tau(y)>1/2+z)\le[(1/2-z)/(1/2+z)]^y\le e^{-4yz}$;
the last log ratio has derivative $-4/(1-4z^2)\le-4$.
On the lower side the tail is
$\Pp_{1/2+z}(\sigma_y>\tau)$, bounded by one below $z=y/\tau$ and by
$e^{-\tau z^2/2}$ above it, using \eqref{eq:stars-failure}.
Integrating $hz^{h-1}$ against these two tails gives
$C_h[y^{-h}+(y/\tau)^h+\tau^{-h/2}]$, equivalent to the claim since $\tau\ge1$.
Finally $|z-1/2|^h$ has total variation $2(1/2)^h$; apply
Lemma~\ref{lem:stars-capacity} in uniform coordinates.
\end{proof}

% % Keep the continuous-price proof and its six steps together in the contents.
% \addtocontents{toc}{\protect\Needspace{11\baselineskip}}
\section{Continuous prices: thresholds, inventory occupation, and regret}
\label{app:stars-continuous}

This appendix proves the continuous-price bounds outlined in
Section~\ref{sec:proof-ideas}. We analyze the mean-threshold rule
$\pi^\circ$; Appendix~\ref{app:stars-sampling} then transfers its
guarantees to finite-scenario \STARS{}.

\subsection{Assumptions and results}

We prove the continuous part of Theorem~\ref{thm:stars-upper}, for
the family in Definition~\ref{def:continuous-family}. Thus $F$ is
atomless on $[0,1]$ and has a density bounded above by a finite
constant $M$; equivalently, its distribution function is
$M$-Lipschitz. Using the median edges $a^-,a^+$ and gap $\Delta$
from Definition~\ref{def:median-edges}, write
\begin{align*}
 g_+(q)&=F^{-1}(1/2+q)-a^+,&
 g_-(q)&=a^--F^{-1}(1/2-q),\qquad 0<q<1/2.
\end{align*}
Set $g_\pm(0)=0$, use one-sided limits at $q=1/2$, and put
$\alpha=1/(\beta+1)\in(0,1]$.
The local assumptions can be written explicitly: for some $q_0\in(0,1/2)$
and constants $c_\pm>0$,
\begin{align}
 g_\pm(q)&=c_\pm q^\alpha(1+r_\pm(q)),\qquad 0<q\le q_0,\notag\\
 |r_\pm(q)|&\le\tfrac12,\qquad r_\pm(q)\longrightarrow0
 \quad(q\downarrow0).
 \label{eq:stars-sc-regularity}
\end{align}
When $\alpha=1$, additionally
$|r_+(q)|+|r_-(q)|=O(q^\eta)$ for some $\eta>0$.
Consequently $g_\pm(q)\sim c_\pm q^\alpha$ for $\alpha<1$,
and $g_\pm(q)=c_\pm q+O(q^{1+\eta})$ for $\alpha=1$.
Uniform prices correspond to $\alpha=1$, $\Delta=0$, and $g_\pm(q)=q$.
Other support gaps away from the median are permitted.
We take $T\ge2$ in the auxiliary estimates below; when $T=1$, the
only feasible action is to hold, so both reference and sampled regret
are zero.

All constants may depend on the fixed distribution, its regularity moduli, and $\rho$,
but not on $T$. In particular, bounded support and
\eqref{eq:stars-sc-regularity} imply the useful global bound
\begin{align*}
 g_+(q)+g_-(q)\le Cq^\alpha,\qquad 0\le q\le1/2.
\end{align*}

\begin{lemma}
\label{thm:stars-continuous}
Fix $\rho\in(0,1)$ and an atomless distribution $F$ with bounded
density and the edge definitions and regularity assumptions above,
including \eqref{eq:stars-sc-regularity} and the quantitative remainder
when $\alpha=1$. Let $\alpha\in(0,1]$, $y_0=y_T=0$, and
$B=\lceil\rho T\rceil$. There is a constant $C=C(F,\rho)<\infty$,
which may depend on the regularity moduli but not on $T$, such that
the reference rule $\pi^\circ$ satisfies the following bounds.
\begin{enumerate}[label=$(\alph*)$]
\item If $\Delta=0$, then
\begin{align*}
 \Reg_T\le C\log^2 T\one\{\alpha=1\}
          +CT^{(1-\alpha)/2}\one\{0<\alpha<1\}.
\end{align*}
\item If the fixed gap satisfies $\Delta>0$, then
\begin{align*}
 \Reg_T\le CT^{1/4}\log T\one\{\alpha=1\}
          +CT^{1/[2(1+\alpha)]}\one\{0<\alpha<1\}.
\end{align*}
\end{enumerate}
No uniform constant
over all distributions satisfying only the qualitative asymptotics is
asserted. Throughout this appendix, $(y_t,x_t)$ denotes the
reference trajectory. The finite-scenario algorithm is handled by
Appendix~\ref{app:stars-sampling}.
\end{lemma}

\subsection{Step 1: mean thresholds and their finite-horizon correction}

We use the random-walk representation in Section~\ref{app:stars-scalar}
to quantify the inventory premium and its finite-horizon correction.

Let $D_{\tau,y}$ denote the uncapped uniform-price hindsight marginal with
$\tau$ future observations and inventory index $y$. Write
$\lambda_\tau(y)=\E F^{-1}(D_{\tau,y})$ for the uncapped dollar threshold
and $\lambda_\tau^B(y)=\E\Delta V_\tau^B(y)$ for its capped counterpart.
Its purchase-probability coordinate is
$\gamma_\tau^B(y)=F(\lambda_\tau^B(y))$.
An index is real only when $1\le y\le(\tau\wedge B)$. As in the algorithm,
an unavailable action is omitted; later a missing threshold is assigned
probability zero only to write the drift formula compactly.
Here $\tau$ is the horizon of the continuation problem. At decision epoch
$t$, there are $\tau_t=T-t+1$ decisions left, so the thresholds used for
the current action have continuation horizon $\tau_t-1$.

\paragraph{The premium with a long future.}
The hitting representation in Lemma~\ref{lem:stars-hitting} gives the limit
$D_{\tau,y}\Rightarrow1/2+Z_y$, where $\Pp(Z_y>z)=R(z)^y$ for
$0<z<1/2$ and $R(z)=(1/2-z)/(1/2+z)$.
Set $A_y=\E g_+(Z_y)$, $\lambda_\infty(y)=a^++A_y$, and
$\gamma_\infty(y)=F(\lambda_\infty(y))$. We use these limiting
quantities to isolate the effect of a finite deadline on the
thresholds. Stochastic monotonicity of the hitting probabilities and the
capacity comparison imply
\begin{align*}
 \lambda_\tau^B(y)\le \lambda_\infty(y),\qquad
 \tfrac12<\gamma_\infty(y)<1,\qquad
 \gamma_\infty(y+1)\le\gamma_\infty(y).
\end{align*}
To justify convergence of the means despite possible jumps of $F^{-1}$,
couple the marginal laws by inverse distribution functions. Their limit
has a continuous law and $F^{-1}$ has only countably many jumps; bounded
convergence applies. The strict bounds on $\gamma_\infty(y)$ follow because the limiting price marginal has positive mass
arbitrarily close to $a^+$ and strictly above it, and is not concentrated
at the largest support price.

The size of the premium follows by one tail integral. For every fixed
$h>0$,
\begin{align*}
 \E Z_y^h
 &\stackrel{(a)}{=}\int_0^{1/2}h z^{h-1}R(z)^y\,dz
 \stackrel{(b)}{=}\frac{\Gamma(1+h)}{(4y)^h}
                   \bigl(1+O_h(y^{-2})\bigr).
\end{align*}
Here $(a)$ is the tail-integral identity for a nonnegative variable,
and $(b)$ follows by substituting $w=4yz$ and controlling the
exponential approximation as follows. On $0\le w\le y$, the logarithmic series gives
$y\log R(w/(4y))=-w-E_y(w)$ with
$0\le E_y(w)\le Cw^3/y^2$. The integrand's difference from $e^{-w}$
is therefore at most $Cy^{-2}w^3e^{-w}$. Beyond $y$, both the
remaining integral and the comparison tail are exponentially small,
since $R(w/(4y))^y\le e^{-w}$. Integration proves the stated error.

Applying \eqref{eq:stars-sc-regularity} to this integral and then inverting
the quantile expansion gives
\begin{align}
 A_y&\sim c_+\Gamma(1+\alpha)(4y)^{-\alpha},&
 \gamma_\infty(y)&=\frac12+\frac{\kappa_\alpha}{4y}+o(y^{-1}),
 &&0<\alpha<1,\notag\\
 A_y&=\frac{c_+}{4y}+O(y^{-1-\nu}),&
 \gamma_\infty(y)&=\frac12+\frac1{4y}+O(y^{-1-\nu}),
 &&\alpha=1,
 \label{eq:stars-sc-stationary-expansion}
\end{align}
where $\nu=(\eta\wedge2)>0$ in the case $\alpha=1$, and
$\kappa_\alpha=\Gamma(1+\alpha)^{1/\alpha}$.
For the first application, bracket $g_+(q)/(c_+q^\alpha)$ by
$1\pm\eta_0$ near zero; the complement has exponentially small
probability. Then let $\eta_0\downarrow0$. For the second, the
quantitative remainder contributes $O(y^{-1-\eta})$.
Inversion uses $F(F^{-1}(u))=u$ and the same two monotone quantile bounds;
it remains valid across other support gaps. In particular,
\begin{align}
 c(y+d)^{-\alpha}\le A_y\le Cy^{-\alpha},\qquad
 \gamma_\infty(y)\le\tfrac12+C/y
 \label{eq:stars-sc-stationary-bounds}
\end{align}
for fixed $c,C,d>0$.

For $0<\alpha<1$ the important coefficient satisfies
$\alpha<\kappa_\alpha<1$. To check both strict inequalities, let
$X$ have Gamma density $x^{\alpha-1}e^{-x}/\Gamma(\alpha)$ and let
$E$ be mean-one exponential. Strict Jensen gives
$1/\Gamma(\alpha)=\E X^{1-\alpha}<\alpha^{1-\alpha}$ and
$\Gamma(1+\alpha)=\E E^\alpha<1$. Taking $\alpha$th roots proves
the claim. For uniform prices, the coefficient is exactly the critical
one: $\gamma_\infty(y)=1/2+1/(4y)+O(y^{-3})$.

\paragraph{The deadline correction.}
Let $\sigma_y$ be the first passage to $y$ of the walk in
Lemma~\ref{lem:stars-hitting}, and put
$S_{\tau,y}(z)=\Pp_{1/2+z}(\sigma_y>\tau)$.
Splitting the quantile at its median interval, and using the first-hit
likelihood ratio from that lemma, gives the exact identity
\begin{align}
 \lambda_\infty(y)-\lambda_\tau(y)
 =\Delta S_{\tau,y}(0)
 +\int_{(0,1/2]}R(z)^yS_{\tau,y}(z)\,dg_+(z)
 +\int_{(0,1/2]}S_{\tau,y}(z)\,dg_-(z).
 \label{eq:stars-sc-bias-exact}
\end{align}
Indeed the finite mean is
$a^-+\Delta[1-S_{\tau,y}(0)]+
\E g_+((D_{\tau,y}-1/2)^+)-\E g_-((1/2-D_{\tau,y})^+)$.
The missing upper-tail probability is $R(z)^yS_{\tau,y}(z)$,
while the lower-tail probability is $S_{\tau,y}(z)$; integrate these
tails. Quantile jumps cause no ambiguity because a uniform marginal,
being one of finitely many continuous input prices, has no fixed atom.

The main estimate is
\begin{align}
 S_{\tau,y}(z)\le C\frac{y}{\sqrt \tau}e^{-\tau z^2/4},
 \qquad 1\le y\le\sqrt \tau/4,\quad 0\le z<1/2.
 \label{eq:stars-sc-survival-tilt}
\end{align}
We prove this by changing measure from the fair walk. Under that walk,
the likelihood ratio for upward probability $1/2+z$ depends only on
the terminal sum $S_\tau$: it is
$L_z(S_\tau)=(1-4z^2)^{\tau/2}((1+2z)/(1-2z))^{S_\tau/2}$.
On survival, $S_\tau<y$, and this likelihood is increasing in $S_\tau$.
Consequently $S_{\tau,y}(z)\le L_z(y)S_{\tau,y}(0)$.
Writing $\ell(z)=\log L_z(y)$, we have $\ell(0)=0$,
$\ell'(0)=2y$, and
$\ell''(z)=-2(\tau+y)/(1+2z)^2-2(\tau-y)/(1-2z)^2\le-\tau$,
since $y\le \tau/2$. Consequently,
\begin{align*}
 \ell(z)
 &\stackrel{(a)}{\le}2yz-\frac{\tau z^2}{2}
 \stackrel{(b)}{\le}\frac{4y^2}{\tau}-\frac{\tau z^2}{4}
 \stackrel{(c)}{\le}\frac14-\frac{\tau z^2}{4},
\end{align*}
where $(a)$ integrates the bound on $\ell''$ twice using the
initial values, $(b)$ uses $2yz\le4y^2/\tau+\tau z^2/4$, and $(c)$ uses
$y\le\sqrt \tau/4$. Lemma~\ref{lem:stars-fair} bounds the fair survival
probability by $Cy/\sqrt \tau$, proving \eqref{eq:stars-sc-survival-tilt}.
For $\tau<16$, the indicated range contains no admissible index.

Put $r=1+\alpha$. Since $R^y\le1$, the two integrals in
\eqref{eq:stars-sc-bias-exact} are at most
$\E[g_+(Z^-)+g_-(Z^-)]\le C\E(Z^-)^\alpha$, where
$Z^-=(1/2-D_{\tau,y})^+$. Integrating \eqref{eq:stars-sc-survival-tilt}
against $\alpha z^{\alpha-1}dz$ and substituting $w=\sqrt \tau z$
therefore proves
\begin{align}
 0\le \lambda_\infty(y)-\lambda_\tau^B(y)
 \le C\left(\frac{\Delta y}{\sqrt \tau}+\frac{y}{\tau^{r/2}}\right)
       +\epsT,\qquad y\le\bigl((\sqrt \tau/4)\wedge(B/2+1)\bigr).
 \label{eq:stars-sc-bias}
\end{align}
The substitution has a finite integral because $\alpha>0$ at zero
and the tail is Gaussian. Capacity contributes at most $\epsT$ by
Lemma~\ref{lem:stars-capacity}.
When $\alpha=1$, the same bound holds for every real
$y\le B/2+1$: for $y>\sqrt \tau/4$, Lemma~\ref{lem:stars-uniformmoments}
gives $\E Z^-\le C(y^{-1}+y/\tau+\tau^{-1/2})\le C'y/\tau$; also
$S_{\tau,y}(0)\le Cy/\sqrt \tau$. In particular,
\begin{align}
 0\le \lambda_\infty(y)-\lambda_\tau^B(y)
 \le C\left(\frac{\Delta y}{\sqrt \tau}+\frac{y}{\tau}\right)+\epsT
 \quad(\alpha=1).
 \label{eq:stars-sc-bias-linear}
\end{align}

\subsection{Step 2: reduce compensation to small and large inventory costs}

There are no prices inside a gap, so threshold distance alone can be
misleading. For $H_F(z)=\E(z-p)^+$, subtract its affine part across
the median interval and set
$R_F(z)=H_F(z)-H_F(a^-)-\tfrac12(z-a^-)$.
This nonnegative function is zero on $[a^-,a^+]$. At the upper
quantile, its area is
$\int_{a^+}^{F^{-1}(1/2+q)}(F(u)-1/2)du\le qg_+(q)$;
integrating $1/2-F(u)$ proves the corresponding lower-side bound.
Hence $R_F(F^{-1}(1/2\pm q))\le Cq^r$.
For an existing hindsight marginal $W=\Delta V_\tau^B(y)$, let
$D_{\tau,y}^B$ be its uniform-price counterpart. Then
\begin{align*}
 J_F(W)
 &\stackrel{(a)}{=}\E R_F(W)-R_F(\E W)
 \stackrel{(b)}{\le}\E R_F(W)
 \stackrel{(c)}{\le}C\E|D_{\tau,y}^B-1/2|^r,
\end{align*}
where $(a)$ uses cancellation of affine terms in the Jensen-gap
identity from Section~\ref{app:stars-compensation}, $(b)$ uses
$R_F\ge0$, and $(c)$ uses quantile equivariance and the area bound
$R_F(F^{-1}(1/2\pm q))\le Cq^r$.

Apply quantile equivariance and Lemma~\ref{lem:stars-uniformmoments} to
the existing marginals at $y$ and $y+1$. With $\tau$ continuation prices,
the current decision and its continuation together have $\tau+1$
time steps. This gives
\begin{align}
 \partial\bar R_\tau^B(y)\le C\left\{(y+1)^{-r}
            +\left(\frac{y+1}{\tau+1}\right)^r
            +(\tau+1)^{-r/2}+\epsT\right\}
            +2\one\{y>B/2\}.
 \label{eq:stars-sc-compensation}
\end{align}
At $y=0$ only index one exists; at $y\ge1$ use
$y^{-r}\le2^r(y+1)^{-r}$. Unavailable actions contribute no channel.
Each channel has cost at most one, giving the displayed wall bound.
A compulsory action has compensation zero, including the final epoch.
At decision epoch $t$, we use this bound with $\tau=\tau_t-1$.

To sum \eqref{eq:stars-sc-compensation}, we next control large
inventories and then the weighted visits near zero.

\subsection{Step 3: control the physical and terminal boundaries}

\paragraph{The physical ceiling.}
From \eqref{eq:stars-sc-stationary-bounds}, the actual drift obeys
$\E[x_t\mid\mathcal F_{t-1}]\le C/(y_{t-1}+1)$.
At positive inventory this follows from
$\gamma_{\tau_t-1}^B(y+1)+\gamma_{\tau_t-1}^B(y)-1$ with
$y=y_{t-1}$; absent purchases and forced
sales can only decrease the drift. At zero enlarge $C$.
For any integer $m\ge1$, the binomial expansion and $|x_t|\le1$
give $\E[y_t^{2m}-y_{t-1}^{2m}\mid\mathcal F_{t-1}]
\le C_m(1+y_{t-1}^{2m-2})$.
The linear term uses the drift bound; all higher powers have degree
at most $2m-2$. The same inequality holds after stopping at the first
visit above $B/2$, since every later increment is zero. Induction on
$m$, summing over decision epochs, proves
$\E y_{T\wedge\sigma}^{2m}\le C_m(T+1)^m$ for that stopping time
$\sigma$. On $\{\sigma\le T\}$ the stopped inventory exceeds $B/2$.
Thus
\begin{align}
 \Pp(\max_{t\le T}y_t>B/2)
 &\stackrel{(a)}{\le}
   \frac{\E y_{T\wedge\sigma}^{2m}}{(B/2)^{2m}}
 \stackrel{(b)}{\le}C_{m,\rho}T^{-m},
 \label{eq:stars-sc-physical-wall}
\end{align}
where $(a)$ is Markov's inequality on the crossing event and $(b)$
uses the stopped-moment bound and $B\ge\rho T$.
Thus reaching the upper half of the store is sufficiently unlikely
along the reference path to contribute only a bounded regret term.

\paragraph{The shrinking terminal boundary.}
The terminal constraint requires a stronger estimate: inventory must
decrease as the remaining horizon shrinks. Set
$L_\tau=4+\log(1+\Delta\tau^{\alpha/2})$. The conclusion will be
\begin{align}
 \E y_{t-1}^2\le C\tau_t L_{\tau_t},\qquad \tau_t=T-t+1.
 \label{eq:stars-sc-terminal-moment}
\end{align}
It is $C\tau_t$ when there is no gap, and at most
$C\tau_t(1+\log\tau_t)$ otherwise, including $\tau_t=1$.

Fix a real index $y\le \tau$ and put $v=y/\tau$.
Use $\theta_0=5/12$ and $\theta_1=1/3$.
At uniform test level $1/2-\theta_0v$, the hypothetical walk has
mean increment $2\theta_0v$. Hitting $y$ by time $\tau$ requires its
centered maximum to exceed $(1-2\theta_0)y=y/6$.
Lemma~\ref{lem:stars-fair} therefore gives
$\Pp(D_{\tau,y}>1/2-\theta_0v)\le e^{-y^2/(72\tau)}$ and, at the
median test level, $\Pp(D_{\tau,y}>1/2)\le e^{-y^2/(2\tau)}$.
Retaining the median geometry when bounding the mean yields
\begin{align}
 \lambda_\tau(y)-a^-
 \le A_y+\Delta e^{-y^2/(2\tau)}
        -g_-(\theta_0v)(1-e^{-y^2/(72\tau)}).
 \label{eq:stars-sc-terminal-mean}
\end{align}
The median gap enters this estimate through the probability that the
marginal remains above the upper gap edge; the other terms are the
same as in the no-gap case.

For all $0<v\le1$, the assumptions imply
\begin{align}
 g_-(\theta_0v)-g_-(\theta_1v)\ge c v^\alpha,
 \qquad g_-(\theta_0v)\le Cv^\alpha.
 \label{eq:stars-sc-quantile-margin}
\end{align}
For small $v$, choose the asymptotic error so that
$(1-\eta_0)\theta_0^\alpha>(1+\eta_0)\theta_1^\alpha$.
For $v\ge v_0>0$, Lipschitz continuity of $F$ gives the lower
bound $(\theta_0-\theta_1)v/M\ge
(\theta_0-\theta_1)v_0^{1-\alpha}v^\alpha/M$.
The upper bound follows from the global edge estimate
$g_+(q)+g_-(q)\le Cq^\alpha$.

Divide the three positive errors in \eqref{eq:stars-sc-terminal-mean}
by $v^\alpha$, and write $w=y/\sqrt \tau$. Their bounds are, respectively,
$Cw^{-2\alpha}$, $\Delta \tau^{\alpha/2}w^{-\alpha}e^{-w^2/2}$, and
$Ce^{-w^2/72}$.
All are smaller than the margin in \eqref{eq:stars-sc-quantile-margin}
when $y\ge K\sqrt{\tau L_\tau}$ for a fixed sufficiently large $K$:
the middle expression uses
$\Delta \tau^{\alpha/2}\le e^{L_\tau}$, and the others decrease with
$w$. Hence $\lambda_\tau(y)\le F^{-1}(1/2-y/(3\tau))$ in this range.
For the complementary range use
$\gamma_\tau^B(y)\le\gamma_\infty(y)\le1/2+C/y$.
Combining the two bounds, and increasing the constant if needed, gives
$\gamma_\tau^B(y)\le1/2+CL_{\tau+1}/y-y/[3(\tau+1)]$.
Capacity only lowers a real marginal. Missing thresholds set to zero
also satisfy this bound: each relevant artificial index has
$y\le \tau+2\le3(\tau+1)/2$ for $\tau\ge1$, so the right side is
nonnegative.

Condition on the predecision history and write $y=y_{t-1}>0$.
Use the preceding bound with continuation horizon $\tau_t-1$.
Independence of the new price gives
\begin{align*}
 \E[y_t^2\mid\mathcal F_{t-1}]
 &\stackrel{(a)}{=}
 y^2+2y\{\gamma_{\tau_t-1}^B(y+1)+\gamma_{\tau_t-1}^B(y)-1\}
       +\E[x_t^2\mid\mathcal F_{t-1}]\\
 &\stackrel{(b)}{\le}
 \left(1-\frac4{3\tau_t}\right)y^2+CL_{\tau_t},
\end{align*}
where $(a)$ expands $(y+x_t)^2$ and uses the threshold formula for
the conditional drift, and $(b)$ substitutes the preceding bound
on each threshold and uses $x_t^2\le1$. At zero the next square
is at most one, so the same upper bound holds after enlarging $C$.
The coefficient is nonnegative for $\tau_t\ge2$.
Moreover, $uL'_u\le\alpha/2\le1/2$, so for $\tau\ge2$,
\begin{align*}
 (\tau-1)L_{\tau-1}-(\tau-4/3)L_\tau
 &\stackrel{(a)}{\ge}\tfrac13L_\tau-\tfrac12
 \stackrel{(b)}{\ge}\tfrac16L_\tau,
\end{align*}
where $(a)$ integrates $L'_u\le1/(2u)$ and uses
$(\tau-1)\log(\tau/(\tau-1))\le1$, and $(b)$ uses $L_\tau\ge4$.
A sufficiently large constant $A$ therefore makes the envelope
$A\tau L_\tau$ satisfy the preceding conditional second-moment recursion.
Starting from $y_0=0$, induction and the tower property prove
\eqref{eq:stars-sc-terminal-moment} down to the penultimate inventory;
the final compulsory action gives $y_T=0$.

\subsection{Step 4: bound compensation-weighted visits near zero}

For $f(y)=(y+1)^{-r}$, we implement the potential argument outlined
in Section~\ref{sec:proof-ideas}. The following finite-horizon form
of the comparison theorem of \citet[Theorem~2.1]{glynn1996liapounov}
retains the terminal potential and applies to the time-varying
inventory process.

\begin{lemma}
\label{lem:stars-drift-comparison}
Let $(y_t)_{t=0}^{\ell}$ be an adapted inventory process taking
values in $\{0,\ldots,B\}$, and let
$\Psi:\{0,\ldots,B\}\to[0,\infty)$ be finite-valued.
Suppose $c_t,e_t$ are nonnegative integrable quantities determined by
$\mathcal F_{t-1}$ and, for $1\le t\le\ell$,
$\E[\Psi(y_t)-\Psi(y_{t-1})\mid\mathcal F_{t-1}]
\le-c_t+e_t$ almost surely. Then
\begin{align*}
 \E\sum_{t=1}^{\ell}c_t
 \le \E\Psi(y_0)-\E\Psi(y_\ell)+\E\sum_{t=1}^{\ell}e_t
 \le \E\Psi(y_0)+\E\sum_{t=1}^{\ell}e_t.
\end{align*}
\end{lemma}
\Needspace{6\baselineskip}
\begin{proof}
The tower property gives
$\E c_t\le\E\Psi(y_{t-1})-\E\Psi(y_t)+\E e_t$.
Summing over $t$ telescopes the potential and proves the first
inequality. Its nonnegativity gives the second.
The argument also holds conditional on the history at the start
of a fixed block of decision epochs.
\end{proof}

Our construction uses an increasing $\Phi$, whose positive expected
change accounts for the visit cost. On a bounded inventory range,
the function $\Psi(y)=\max_{0\le z\le B}\Phi(z)-\Phi(y)$
is nonnegative and has the opposite drift. Thus the preceding lemma
applies to $c_t=f(y_{t-1})$ whenever
$\E[\Phi(y_t)-\Phi(y_{t-1})\mid\mathcal F_{t-1}]
\ge f(y_{t-1})-e_t$.
The next lemma constructs such a potential explicitly from lower
bounds on the buying probabilities.

\Needspace{20\baselineskip}
\begin{lemma}
\label{lem:stars-sc-accounting}
Consider a contiguous block of at most $\bar\tau$ decision epochs. Suppose every
threshold on edges $1,\ldots,J$ is real throughout the block and
$\gamma_{\tau_t-1}^B(y)\ge b_y\ge1/2$ at each epoch $t$ in the
block, where $b_y<1$ are fixed numbers.
Fix a nonnegative decreasing $f$, set $\pi_0=1$ and
$\pi_j=\prod_{i=1}^j b_i/(1-b_i)$, and define its weighted mean by
$\bar f=\sum_{j=0}^J\pi_j f(j)/\sum_{j=0}^J\pi_j$.
Set $G_J=0$ and, for $0\le j<J$, set
\begin{align}
 G_j=\frac{\sum_{i=0}^j\pi_i(f(i)-\bar f)}{\pi_j b_{j+1}},
 \qquad
 \Phi(y)=\sum_{j=0}^{(y\wedge J)-1}G_j.
 \label{eq:stars-sc-accounting-function}
\end{align}
Conditionally on any history at the start of the block,
\begin{align*}
 \E\sum_{t\text{ in block}}f(y_{t-1})\le \bar\tau\bar f+\Phi(J).
\end{align*}
\end{lemma}

\begin{proof}
Every $G_j$ is nonnegative: the weighted average of a decreasing
$f$ over an initial segment is at least its average $\bar f$
over the whole segment. The identity
$\pi_jb_{j+1}=\pi_{j+1}(1-b_{j+1})$ makes
\eqref{eq:stars-sc-accounting-function} telescope to
\begin{align}
 b_{y+1}G_y-(1-b_y)G_{y-1}=f(y)-\bar f,\qquad 0\le y\le J.
 \label{eq:stars-sc-accounting-equation}
\end{align}
Omit the downward term at zero and the upward term at $J$; the
latter equation uses $\sum_{i=0}^J\pi_i(f(i)-\bar f)=0$.
Thus the weighted partial sums in the construction give precisely
the expected-increment identity needed to account for $f(y)$.

At a reference state $y=y_{t-1}<J$, we have
\begin{align*}
 \E[\Phi(y_t)-\Phi(y_{t-1})\mid\mathcal F_{t-1}]
 &\stackrel{(a)}{=}\gamma_{\tau_t-1}^B(y+1)G_y-
                           (1-\gamma_{\tau_t-1}^B(y))G_{y-1}\\
 &\stackrel{(b)}{\ge}b_{y+1}G_y-(1-b_y)G_{y-1}
 \stackrel{(c)}{=}f(y)-\bar f,
\end{align*}
where $(a)$ conditions on the independent current price, $(b)$ uses
the assumed lower bounds on the two thresholds and $G_j\ge0$, and $(c)$ is
\eqref{eq:stars-sc-accounting-equation}. At zero the downward
terms in this chain are omitted.
At $y=J$ the purchase term is zero, so no bound on edge $J+1$
is needed. Above $J$, $\Phi$ is flat and $f(y)\le f(J)\le\bar f$,
so the same inequality holds, even at a constrained boundary.
Sum it and use the tower property. The potential telescopes on
the reference trajectory; $0\le\Phi\le\Phi(J)$ proves the claimed accounting bound.
\end{proof}

The centering by $\bar f$ enforces both boundary conditions: multiplying
\eqref{eq:stars-sc-accounting-equation} by $\pi_y$ and summing over
$0\le y\le J$ makes the left side zero, so the right side must also
sum to zero. The resulting potential $\Psi(y)=\Phi(J)-\Phi(y)$
satisfies the comparison lemma with $c_t=f(y_{t-1})$ and
$e_t=\bar f$.

\paragraph{Which inventories receive an upward drift?}
Group decision epochs according to their remaining continuation horizon:
in a block with time scale $\bar\tau$, we have
$\bar\tau\le \tau_t-1<2\bar\tau$. There are at most $\bar\tau$
epochs in this block. In the threshold estimates below, write
$\tau=\tau_t-1$, so $\bar\tau\le\tau<2\bar\tau$.
For all sufficiently large $\bar\tau$ and $T$, its indices up to
$J=\lfloor h\sqrt{\bar\tau}\rfloor$ are real and below $B/2+1$,
where $0<h<1/4$ is a fixed small constant. The mean split used in
\eqref{eq:stars-sc-bias-exact} gives
$\lambda_\tau(y)-a^-=\Delta(1-S_{\tau,y}(0))+A_y-E_{\tau,y}$, where
$0\le E_{\tau,y}\le Cy\bar\tau^{-r/2}$.
Choose $h$ so small that $S_{\tau,y}(0)\le1/4$ and the last error
is at most $A_y/4$ for $y\le J$. This is possible because
$A_y\ge c(y+d)^{-\alpha}$ and
$y(y+d)^\alpha \bar\tau^{-r/2}\le C h^r$ once $\bar\tau$ is large.
The capacity error is also at most $A_y/4$ for sufficiently large
$T$. Hence $\lambda_\tau^B(y)>a^-$ and
\begin{align*}
 \gamma_\tau^B(y)\ge1/2,\qquad 1\le y\le J.
\end{align*}
All fixed excluded remaining-time blocks contain only a bounded
number of decision epochs and will contribute a constant.

There is stronger repulsion on a smaller band of width
$H=\lfloor\ell\bigl(\sqrt{\bar\tau}\wedge(\sqrt{\bar\tau}/\Delta)^{1/r}\bigr)\rfloor$,
where the gap-dependent entry is infinite at $\Delta=0$ and
$\ell>0$ is a small fixed constant with $H\le J/4$.
Increase the fixed block-size cutoff so that $H\ge(2\vee d)$
and taking integer parts changes these scales by only fixed factors.

If $0<\alpha<1$, choose $\alpha<\kappa<\kappa_\alpha$.
The strict coefficient gap in \eqref{eq:stars-sc-stationary-expansion}
implies $A_y-g_+(\kappa/[4(y+d)])\ge c_*(y+d)^{-\alpha}$ for
all $y\ge1$, with fixed $d,c_*>0$.
For large $y$ this follows by comparing leading coefficients;
increase $d$ to cover the remaining finite set using $A_y>0$.
By \eqref{eq:stars-sc-bias}, choosing $\ell$ smaller if needed makes
the deadline error use less than half this slack, since
$\Delta H^r/\sqrt{\bar\tau}$ and $H^r\bar\tau^{-r/2}$ are bounded by fixed
multiples of $\ell^r$. Therefore we may use the reference thresholds
$b_y=1/2+\kappa/[4(y+d)]$ for $y\le H$ and $b_y=1/2$ for
$H<y\le J$.
Their weights satisfy
\begin{align}
 \pi_j\asymp(j+1)^\kappa\quad(0\le j\le H),\qquad
 \pi_j\asymp H^\kappa\quad(H<j\le J).
 \label{eq:stars-sc-power-weights}
\end{align}
Indeed, the logarithm of $b_y/(1-b_y)$ equals
$\kappa/(y+d)+O((y+d)^{-3})$. The error series is summable,
and the harmonic sum differs from $\log(j+1)$ by a bounded
amount. Constants absorb the fixed shift and small indices.

For $\alpha=1$, put $e_y=\gamma_\infty(y)-1/2$ and use
$b_y=1/2+(e_y-A[\Delta y/\sqrt{\bar\tau}+y/\bar\tau])^+$ for $y\le H$;
again take $b_y=1/2$ for $H<y\le J$.
These are lower bounds by \eqref{eq:stars-sc-bias-linear}, the Lipschitz
property of $F$, and the protected-band lower bound
$\gamma_\tau^B(y)\ge1/2$; enlarge $A$ to
absorb $\epsT$. The expansion of $\gamma_\infty$ gives
$\log[\gamma_\infty(y)/(1-\gamma_\infty(y))]
=1/y+O(y^{-1-\nu})+O(y^{-3})$.
The log-odds function has bounded derivative on
$[1/2,\gamma_\infty(1)]$, while the total perturbation up to $H$
is at most $C(\Delta H^2/\sqrt{\bar\tau}+H^2/\bar\tau)\le C$.
Thus these weights satisfy
\begin{align}
 \pi_j\asymp j+1\quad(0\le j\le H),\qquad
 \pi_j\asymp H\quad(H<j\le J).
 \label{eq:stars-sc-linear-weights}
\end{align}

\paragraph{From threshold protection to an occupation bound.}
Apply Lemma~\ref{lem:stars-sc-accounting} to $f(j)=(j+1)^{-r}$.
For $\alpha<1$, \eqref{eq:stars-sc-power-weights} gives
\begin{align*}
 \sum_{j=0}^J\pi_j&\ge cJH^\kappa,&
 \sum_{j=0}^J\pi_jf(j)&\le CH^{\kappa-\alpha},\\
 \bar f&\le CH^{-\alpha}/J,&
 G_j&\le
 \begin{cases}C(j+1)^{-\alpha},&j<H,\\ CH^{-\alpha},&H\le j<J.
 \end{cases}
\end{align*}
For the numerator, sum $(j+1)^{\kappa-1-\alpha}$ below $H$
and bound $H^\kappa\sum_{j\ge H}(j+1)^{-1-\alpha}$ above it.
The condition $\kappa>\alpha$ is exactly what makes the first
sum of order $H^{\kappa-\alpha}$. For $G_j$, drop the negative
$-\bar f$ term in its numerator and use $b_{j+1}\ge1/2$.
Summing the displayed increments gives
$\Phi(J)\le CJH^{-\alpha}$, since $H\le J$.
Consequently,
\begin{align*}
 \E\sum_{t\text{ in block}}(y_{t-1}+1)^{-r}
 &\stackrel{(a)}{\le}\bar\tau\bar f+\Phi(J)
 \stackrel{(b)}{\le}CJH^{-\alpha}\\
 &\stackrel{(c)}{\le}
 C\left\{\bar\tau^{(1-\alpha)/2}
          +\Delta^{\alpha/r}\bar\tau^{1/(2r)}\right\},\qquad \alpha<1,
\end{align*}
where $(a)$ is Lemma~\ref{lem:stars-sc-accounting}, $(b)$ uses
the bounds on $\bar f$ and $\Phi(J)$ together with $\bar\tau\asymp J^2$,
and $(c)$ substitutes the two scales in the definition of $H$.
All comparisons are conditional on an arbitrary starting history.
The sums just used follow by adding the first term to the appropriate
integral on $[1,H]$ or $[H,\infty)$; no integral starts at zero.

For $\alpha=1$, the same computation with
\eqref{eq:stars-sc-linear-weights} gives
$\bar f\le C\log H/(JH)$ and
\begin{align*}
 G_j\le
 \begin{cases}
 C\log(j+2)/(j+1),&j<H,\\
 C\log H/H,&H\le j<J.
 \end{cases}
\end{align*}
Consequently the same accounting bound is
\begin{align*}
 \E\sum_{t\text{ in block}}(y_{t-1}+1)^{-2}
 \le C\left\{\log^2 H+\frac{\sqrt{\bar\tau}}{H}\log H\right\}.
\end{align*}
The logarithmic square is simply the sum of
$\log(j+2)/(j+1)$ up to $H$.
For a fixed positive gap, $H\asymp\Delta^{-1/2}\bar\tau^{1/4}$
for all sufficiently large $\bar\tau$, and
$\log H\le C\sqrt{\bar\tau}/H$. Thus
the preceding block estimate is $O_F(\bar\tau^{1/4}\log \bar\tau)$.
The finitely many earlier blocks contribute a distribution-dependent
constant, giving the stated block rate for a positive gap.

Partition $\tau\ge1$ into dyadic blocks $[2^l,2^{l+1})$.
The first chronological block may be truncated; the accounting lemma
allows that. Take expectations and sum the geometric series to obtain
\begin{align}
 \E\sum_{t=1}^T(y_{t-1}+1)^{-r}
 &\le C\{T^{(1-\alpha)/2}
              +\Delta^{\alpha/r}T^{1/(2r)}\},&&\alpha<1,\notag\\*
 \E\sum_{t=1}^T(y_{t-1}+1)^{-2}
 &\le C T^{1/4}\log T,&&\alpha=1,\quad\Delta>0.
 \label{eq:stars-sc-block-occupation}
\end{align}
Both power exponents in the first line are positive. For the second
line, bound the logarithm by $\log T$ and sum the geometric series
$\sum_{2^l\le T}2^{l/4}\le CT^{1/4}$.
The last epoch adds at most one.

\paragraph{No gap at the critical exponent: use the accounting rule globally.}
For $\alpha=1$, $\Delta=0$, restarting the accounting function
on every block would charge repeatedly for its logarithmic sum.
The finite-horizon error is now only $Cy/\tau$, so the same difference
equation can instead be used once on the whole horizon.

Take $b_y=\gamma_\infty(y)$, $\bar f=0$, and no cutoff $J$.
The construction in Lemma~\ref{lem:stars-sc-accounting} now uses $\pi_j=\prod_{i=1}^j\gamma_\infty(i)/(1-\gamma_\infty(i))$
and $G_j=(\sum_{i=0}^j\pi_i(i+1)^{-2})/
(\pi_j\gamma_\infty(j+1))$. Define $\Phi(y)=\sum_{j=0}^{y-1}G_j$.
Exactly the same summable log-odds expansion gives
$\pi_j\asymp j+1$ and $0\le G_j\le C\log(j+2)/(j+1)$.
Equation~\eqref{eq:stars-sc-accounting-equation} still holds, now with
$f(y)=(y+1)^{-2}$ and $\bar f=0$.

The actual thresholds are slightly smaller than these reference
coefficients. By \eqref{eq:stars-sc-bias-linear} and the Lipschitz property,
$0\le\gamma_\infty(y)-\gamma_\tau^B(y)\le Cy/(\tau+1)+C\epsT$
for real indices $y\le B/2+1$. The same bound covers missing
terminal indices $y=\tau+1,\tau+2$, assigned threshold zero, because
$y/(\tau+1)\ge1$. Below the physical half-capacity, a purchase is
not excluded by the physical ceiling; bounded small $T$ are absorbed
in the constants.
Each threshold error multiplies a corresponding increment $G_y$
or $G_{y-1}$ in the expected change of $\Phi$. At epoch $t$, use
$y=y_{t-1}$ and continuation horizon $\tau_t-1$ to obtain
\begin{align*}
 (y+1)^{-2}\le
 \E[\Phi(y_t)-\Phi(y_{t-1})\mid\mathcal F_{t-1}]
 +C\frac{\log(y+2)}{\tau_t}+C\epsT+C\one\{y>B/2\}.
\end{align*}
At zero omit the downward increment. Above $B/2$, boundedness
of $f$ and all $G_j$ justifies the added constant wall charge.
Sum up to $t=T-1$. The potential telescopes on the actual path,
and $y_{T-1}\le1$ makes its final value at most $G_0$.
Using \eqref{eq:stars-sc-physical-wall} with $m=3$ yields
\begin{align}
 \E\sum_{t=1}^T(y_{t-1}+1)^{-2}
 &\stackrel{(a)}{\le}
 C+C\log(B+2)\sum_{\tau=2}^T\tau^{-1}
       +CT\epsT+CT\Pp(\max_t y_t>B/2)\notag\\
 &\stackrel{(b)}{\le}C\log^2 T,
 \label{eq:stars-sc-global-occupation}
\end{align}
where $(a)$ sums the conditional drift inequality by the tower
property, telescopes $\Phi$, and adds at most one for the final epoch;
$(b)$ uses the harmonic-sum bound, $B=\lceil\rho T\rceil$,
$T\epsT=O(1)$, and \eqref{eq:stars-sc-physical-wall} with $m=3$.
For the critical uniform case, applying the accounting identity over
the whole horizon therefore gives the required logarithmic-square
bound.

\subsection{Step 5: sum the compensation and finish}

We now sum the terms in \eqref{eq:stars-sc-compensation} using the
regret decomposition from Section~\ref{app:stars-compensation}.
The physical-wall error is bounded by \eqref{eq:stars-sc-physical-wall},
and $T\epsT=O(1)$. The inverse-inventory term is controlled by
\eqref{eq:stars-sc-block-occupation} or \eqref{eq:stars-sc-global-occupation}.

For the terminal-inventory term,
\begin{align}
 \E\left(\frac{y_{t-1}+1}{\tau_t}\right)^r
 &\stackrel{(a)}{\le}
 \frac{[\E(y_{t-1}+1)^2]^{r/2}}{\tau_t^r}
 \stackrel{(b)}{\le}C\tau_t^{-r/2}L_{\tau_t}^{r/2},
 \label{eq:stars-sc-final-moment}
\end{align}
where $(a)$ is Jensen's inequality with $r/2\le1$, and $(b)$
uses $(y+1)^2\le2y^2+2$ and \eqref{eq:stars-sc-terminal-moment}.
With no gap, $L_\tau=4$, so its sum and the separate fluctuation
sum are $O(T^{(1-\alpha)/2})$ for $\alpha<1$, and
$O(\log T)$ for $\alpha=1$. These follow from
$1+\int_1^T x^{-r/2}dx$, treating $r=2$ separately.
Combining with the corresponding occupation bounds gives the two
no-gap claims in Lemma~\ref{thm:stars-continuous}.

For a positive gap, $L_\tau\le C\log\tau$ for $\tau\ge2$,
while the term at $\tau=1$ is bounded by a constant.
When $\alpha<1$, the sum in \eqref{eq:stars-sc-final-moment} is at most
$CT^{(1-\alpha)/2}(\log T)^{r/2}$, which is dominated by
$CT^{1/(2r)}$ because
$1/[2(1+\alpha)]-(1-\alpha)/2=\alpha^2/[2(1+\alpha)]>0$.
For $\alpha=1$, it is $O(\log^2 T)$, dominated by
$T^{1/4}\log T$. Adding \eqref{eq:stars-sc-block-occupation} proves
the two positive-gap claims and completes Lemma~\ref{thm:stars-continuous},
which supplies the reference bounds. Appendix~\ref{app:stars-sampling}
adds at most $M\sum_{\tau=2}^T\tau^{-1}\le M\log T$ for the
specified finite sample counts, completing all continuous cases of
Theorem~\ref{thm:stars-upper}.

\section{Finite atomic prices: a logarithmic buffer emerges automatically}
\label{app:stars-atomic}

For atomic prices, moving a threshold through an empty price interval
does not change the trading decision. We therefore work directly with
the compensation formula, which charges only for disagreements at
support prices. For a hindsight marginal $W$, this formula is
\begin{align}
 J_F(W)=\sum_{i=1}^d w_i\bigl(\E(W-z_i)^+\wedge\E(z_i-W)^+\bigr),
 \label{eq:stars-at-hinges}
\end{align}
where the fixed support and masses satisfy
$z_1<\cdots<z_d$, $w_i>0$, and $\sum_iw_i=1$.
Indeed, at $z_i$ the identity follows from
$A-(A-D)^+=A\wedge D$, with $A=\E(W-z_i)^+$ and $D=\E(z_i-W)^+$.
Only uncertainty that straddles an actual support price contributes.

This identity gives the reference charge at a specified inventory.
Lemma~\ref{lem:sampled-compensation} bounds the sampled charge by
twice this quantity plus $C/K_t$. Our task here is to determine how
the reference charge varies with inventory. Appendix~\ref{app:finite-atomic}
then controls its accumulation along the sampled inventory path,
completing the atomic case of Theorem~\ref{thm:stars-upper}.
Here $\tau$ denotes the number of prices in a continuation problem;
at decision time $t$, this is $\tau_t-1$, where $\tau_t=T-t+1$ includes
the current decision.

\subsection{Compensation at a given inventory when there is a median gap}

Suppose $F(z_j)=1/2$ for some $j<d$. Write $a^-=z_j$,
$a^+=z_{j+1}$, $\Delta=a^+-a^->0$, $\mu^-=w_j$, and
$\mu^+=w_{j+1}$. Assume first that $\mu^+<1/2$, and set
$p_*=1/2+\mu^+$, $q_*=1/2-\mu^+$, and $r=q_*/p_*\in(0,1)$.
The ordinary-walk representation in Lemma~\ref{lem:stars-hitting} explains the special role of $a^+$. For the tail above this atom, the ordinary walk has upward probability $q_*<1/2$, so its probability of reaching level $y$ is at most $r^y$. Immediately below the atom, the corresponding ordinary walk is fair, and its probability of not reaching $y$ by time $\tau$ is of order $(y/\sqrt \tau)\wedge1$. The lemma identifies these probabilities with one-time tails of the uncapped count process; Lemma~\ref{lem:stars-capacity} controls the effect of the physical cap. The smaller of these two tail bounds controls compensation at $a^+$.

\begin{lemma}
\label{lem:stars-at-local}
Fix the finite atomic law $F$ with positive median gap and
$\rho\in(0,1)$ as above. There are constants $c,C>0$, depending
only on $F$ and $\rho$, such that the following hold for every
feasible $y\le B/2$ and $\tau\ge1$. Write
$\epsT=CT^2e^{-c_\rho T}$ for the error from
Lemma~\ref{lem:stars-capacity}.
\begin{enumerate}[label=$(\alph*)$]
\item If $\mu^+<1/2$, there is a fixed $s\in(0,r)$, depending on
$F$, such that
\begin{align}
 \partial\bar R_{\tau}^B(y)\le C\left[
 \left(r^y\wedge\frac{y+1}{\sqrt \tau}\right)
 +s^y+e^{-c\tau}+\epsT\right]. \label{eq:stars-at-local}
\end{align}
\item If $\mu^+=1/2$, then
$\partial\bar R_{\tau}^B(y)\le C(e^{-c\tau}+\epsT)$.
\end{enumerate}
\end{lemma}
\begin{proof}
Apply \eqref{eq:stars-at-hinges} to an existing marginal
$W=\Delta V_{\tau}^B(\ell)$, where $\ell\in\{y,y+1\}$ and
$\ell\le B/2+1$. For $z_i\le a^-$, put $d_i=1-2F(z_i)\ge0$.
The walk for $W>z_i$ has drift $d_i$; the walk whose failure determines
$W<z_i$ has drift $d_i+2w_i$. If $\ell\ge(d_i+w_i)\tau$, a hit by any
time $v\le\tau$ in the first walk requires a centered upward deviation
at least $w_i\tau$. If $\ell<(d_i+w_i)\tau$, failure to hit in the second
requires a terminal downward deviation exceeding $w_i\tau$.
The maximal and endpoint Hoeffding bounds, together with
Lemma~\ref{lem:stars-capacity}, give
\begin{align*}
 \E(W-z_i)^+\wedge\E(z_i-W)^+
 \le e^{-w_i^2\tau/2}+\epsT,\qquad z_i\le a^-.
\end{align*}
We used that a positive part is at most its positivity probability, since prices lie in $[0,1]$. At the lowest support price, the downward positive part is zero.

At $a^+$, Lemma~\ref{lem:stars-hitting} gives
$\E(W-a^+)^+\le r^{\ell}$, while
\begin{align*}
 \E(a^+-W)^+
 &\stackrel{(a)}{\le}\Pp(W\le a^-)
 \stackrel{(b)}{\le}\Pp_{1/2}(\sigma_{\ell}>\tau)+\epsT
 \stackrel{(c)}{\le}C\ell/\sqrt \tau+\epsT.
\end{align*}
where (a) uses the absence of support in $(a^-,a^+)$ and the
unit price range, (b) uses the fair queue at $a^-$ and
Lemma~\ref{lem:stars-capacity}, and (c) applies
Lemma~\ref{lem:stars-fair}.
For every $z_i>a^+$, the upward hinge is at most $r_i^{\ell}$, where
$r_i=(1-F(z_i))/F(z_i)<r$. Choose $s\in(0,r)$ larger than these finitely
many ratios. Adding the hinges for the existing marginals at $y$ and
$y+1$ proves \eqref{eq:stars-at-local}. At $y=0$ only the marginal at
$1$ exists; compulsory actions contribute zero. If $\mu^+=1/2$, no price
exceeds $a^+$, whose hinge is then zero, so only the bounds for atoms at
or below $a^-$ remain.
\end{proof}

\subsection{The mean thresholds reveal a logarithmic reserve}

We next identify the inventory levels at which the mean continuation
value makes buying at $a^+$ preferable to holding. This shows how a
logarithmic reserve arises from the scenario optimization. We also
quantify the buying advantage, which will allow us to control the
effect of sampling on these decisions.

\Needspace{11\baselineskip}
\begin{lemma}
\label{lem:stars-at-thresholds}
Fix $F$ and $\rho$ as above, and suppose $\mu^+<1/2$.
Choose a sufficiently large fixed $A$. For a sufficiently large
remaining-time scale $\bar\tau\ge2$, define
\begin{align*}
 h_{\bar\tau}=\left\lfloor\frac{\log(\sqrt{\bar\tau}/[A\log \bar\tau])}{|\log r|}\right\rfloor-1.
\end{align*}
The constants below and the cutoff on $\bar\tau$ may depend on $F$, $\rho$,
and $A$, but not on $\bar\tau$, $\tau$, or $T$.
\begin{enumerate}[label=$(\alph*)$]
\item There is a constant $c>0$ such that, uniformly for
$\bar\tau\le \tau<2\bar\tau$ and existing $y\le h_{\bar\tau}+1$,
$\lambda_{\tau}^B(y)-a^+\ge c\log \bar\tau/\sqrt{\bar\tau}$.
\item $h_{\bar\tau}\asymp\log \bar\tau$ and
$r^{h_{\bar\tau}}\asymp\log \bar\tau/\sqrt{\bar\tau}$.
\end{enumerate}
\end{lemma}
\begin{proof}
Let $\gamma=z_{j+2}-a^+>0$. At price level $a^++\gamma/2$,
likelihood reversal gives an upward-tail probability at least
$r^y(1-e^{-c\tau})$ when $y=O(\log \tau)$: the reversed walk has fixed
positive drift, so its failure probability is exponentially small.
Fair survival bounds the negative part below $a^+$. Including the
capacity error, therefore,
\begin{align*}
 \lambda_{\tau}^B(y)-a^+\ge c_1r^y-C_2y/\sqrt \tau-C_3\epsT,
                  \qquad y=O(\log \tau).
\end{align*}
At $y=h_{\bar\tau}+1$, the definition gives
$r^y\ge A\log \bar\tau/\sqrt{\bar\tau}$ and $y\le C\log \bar\tau$.
Choose $A$ large, then a fixed cutoff on $\bar\tau$ large enough to absorb
the exponential error. This gives the stated margin; concavity
extends it to smaller indices. The two scale estimates follow
directly from $h_{\bar\tau}$'s definition. Increasing the cutoff ensures
$h_{\bar\tau}\ge1$ and $h_{\bar\tau}+1<(\bar\tau\wedge B)$, since $B\ge\rho \bar\tau$.
\end{proof}

\paragraph{Seeing the buffer in the four-price example.}
Give each price in $\{0.2,0.4,0.6,0.8\}$ probability $1/4$, as in
Example~\ref{ex:atomic}. Here $a^+=0.6$ and $r=1/3$.
Without the capacity bound, $W=\Delta V_\tau(y)$ takes only these
four prices whenever $1\le y\le\tau$. Therefore
$\E W-0.6=0.2\{\Pp(W>0.6)-\Pp(W\le0.4)-\Pp(W\le0.2)\}$.
Let $\sigma_y$ be the first time a sign walk starting at zero reaches
$y$, and let the subscript on $\Pp$ give its upward probability.
Lemma~\ref{lem:stars-hitting} identifies the three probabilities
above with hitting at upward probability $1/4$, survival at $1/2$,
and survival at $3/4$, respectively. Reversing the likelihood of the
first walk gives
\begin{align}
 \lambda_\tau(y)-0.6
 =0.2\left[3^{-y}\Pp_{3/4}(\sigma_y\le\tau)
 -\Pp_{1/2}(\sigma_y>\tau)-\Pp_{3/4}(\sigma_y>\tau)\right].
 \label{eq:stars-at-example}
\end{align}
For $y=O(\log\tau)$, Hoeffding's inequality gives
$\Pp_{3/4}(\sigma_y>\tau)=O(e^{-c\tau})$, since failure to hit
requires the final walk position to be below $y$.
Lemma~\ref{lem:stars-fair} bounds fair survival above and below by
constant multiples of $y/\sqrt{\tau}$.
The sign changes where $3^{-y}\asymp y/\sqrt{\tau}$, hence at
$y=(\log \tau)/(2\log3)-(\log\log \tau)/\log3+O(1)$.
To obtain this expansion, take logarithms in the equation
$3^{-y}=y/\sqrt \tau$ for its unique positive real solution. This first
gives $y\asymp\log \tau$; substituting back gives the stated expression.
Shifting the solution by a sufficiently large fixed constant in either
direction makes $3^{-y}/(y/\sqrt{\tau})$ respectively larger or smaller
than the fixed survival constants; the exponentially small terms
in \eqref{eq:stars-at-example} do not change these signs.
Monotonicity of $\lambda_{\tau}(y)$ then locates
its crossing within $O(1)$. The capacity error $\epsT$ is negligible
relative to $\log \tau/\sqrt \tau$ at these indices.

At price $0.6$, the mean-threshold rule buys if
$\lambda_\tau(y+1)>0.6$, sells if $\lambda_\tau(y)<0.6$, and holds
when its two adjacent marginals straddle that price, including ties.
The continuation values therefore generate the buffer endogenously.
Appendix~\ref{app:finite-atomic} shows that the reserve remains
effective with $K_t\ge(T-t+1)^2$, although sampling can change the
cutoff on a particular realization.

\subsection{A crossing identity that applies to the sampled policy}

To turn the local compensation bound into a regret bound, we need to
relate visits at adjacent inventory levels. A purchase from $y$ to
$y+1$ and a sale from $y+1$ to $y$ are opposite crossings of the same
pair of levels. Over any consecutive block of decisions, their counts
differ by at most one: every purchase must be followed by a sale before
another purchase can cross that same pair of levels, and conversely.
The difference is determined by the block's initial and final
inventories. We use this observation to express occupation in terms of
the policy's corrections to the actions prescribed by the static median-sign rule.

Use the fair signs $\xi_t=\one\{p_t\le a^-\}-\one\{p_t\ge a^+\}$
and upward corrections $u_t=(x_t-\xi_t)^+$ of any feasible online
policy, including sampled \textsf{STARS}. On a consecutive block
$I=\{a,\ldots,c\}$, define expected visits
$\nu_j=\E\sum_{t\in I}\one\{y_{t-1}=j\}$ and corrections
$A_j=\E\sum_{t\in I}u_t\one\{y_{t-1}=j\}$.
If $C_j^+,C_j^-$ count expected upcrossings and downcrossings of
$j\leftrightarrow j+1$, cancellation gives
\begin{align*}
 C_j^+-C_j^-=
 \E[\one\{y_c\ge j+1\}-\one\{y_{a-1}\ge j+1\}]
 \in[-1,1]. 
\end{align*}
Pointwise, $\one\{x_t=1\}\le\one\{\xi_t=1\}+u_t$ and
$\one\{x_t=-1\}\ge\one\{\xi_t=-1\}-u_t$.
The current sign is fair independently of the past, hence
$C_j^+\le\nu_j/2+A_j$ and
$C_j^-\ge\nu_{j+1}/2-A_{j+1}$. It follows that
\begin{align}
 \nu_{j+1}&\stackrel{(a)}{\le}\nu_j+2+2(A_j+A_{j+1}),\notag\\
 \nu_j&\stackrel{(b)}{\le}\nu_h+2(j-h)+4\sum_{i=h}^j A_i,
       \qquad j>h.\label{eq:stars-at-correctioncross}
\end{align}
Here (a) combines the preceding inequalities with crossing cancellation,
and (b) sums over successive edges. The pointwise inequalities remain
valid at compulsory sales and unavailable purchases.

\subsection{Compensation when the median lies strictly inside an atom}
\label{sec:stars-at-medianatom}

The only remaining possibility is an atom $a$ with
$F(a^-)<1/2<F(a)$,
where $F(a^-)=\Pp(p<a)$ denotes a left limit. Buying or selling at
this atom produces a fixed drift. We first show that compensation
decays exponentially with the distance from the storage boundaries,
up to an exponentially small remaining-horizon term.
We then use the drift away from a costly boundary to bound its
accumulated compensation under sampled \textsf{STARS}.

\paragraph{A count estimate with both storage boundaries.}
Let $C_{t+1}=\clip_{[0,B]}(C_t+\xi_{t+1})$, $C_0=0$, with upward probability $p>1/2$. If $p<1$, put $\delta=2p-1$ and $r_p=(1-p)/p$. For $1\le y\le B$ and $y<\delta \tau$,
\begin{align}
 \Pp(C_{\tau}<y)
 \le e^{-(\delta \tau-y)^2/(2\tau)}+r_p^{B-y}.
 \label{eq:stars-at-cappeddrift}
\end{align}
To prove this, write $S_t=\sum_{i=1}^t\xi_i$. The upper-reflected process is
\begin{align*}
 Z_{\tau}=S_{\tau}-\max_{0\le j\le \tau}(S_j-B)^+
 =S_{\tau}\wedge\left(B+\min_{0\le j\le \tau}(S_{\tau}-S_j)\right).
\end{align*}
Substitution verifies $Z_{t+1}=B\wedge(Z_t+\xi_{t+1})$, $Z_0=0$. Adding lower reflection can only increase the process, so $C_{\tau}\ge Z_{\tau}$. Thus $C_{\tau}<y$ requires either $S_{\tau}<y$ or a suffix sum below $-(B-y)$. Hoeffding bounds the first event; reversal of the signs and the negative-drift hitting formula bound the second by $r_p^{B-y}$. The weaker exponent $B-y$ covers the strict integer inequality. If $p=1$, $C_{\tau}=\tau\wedge B$ deterministically, and the failure probability is zero at every existing index.

\paragraph{Local compensation is concentrated near the walls.}
For constants $c,C>0$ and every feasible $y$, $\tau\ge1$,
\begin{align}
 \partial\bar R_{\tau}^B(y)\le C\{e^{-cy}+e^{-c(B-y)}+e^{-c\tau}\}.
 \label{eq:stars-at-strictlocal}
\end{align}
The lower-wall term is unnecessary if no support price exceeds $a$; the upper-wall term is unnecessary if no price is below $a$.
We prove the bound for each existing marginal
$W=\Delta V_{\tau}^B(\ell)$, $\ell\in\{y,y+1\}$, by treating support
prices below, at, and above the median atom separately.
For $z_i<a$, put $d_i=1-2F(z_i)>0$.
If $\ell\ge(d_i+w_i)\tau$, maximal Hoeffding makes the upward hinge at
most $e^{-w_i^2\tau/2}$. If $\ell<(d_i+w_i)\tau$, apply
\eqref{eq:stars-at-cappeddrift} to the walk immediately below $z_i$,
whose drift is $d_i+2w_i$, to make the downward hinge at most
$e^{-w_i^2\tau/2}+e^{-c_i(B-\ell)}$. If that walk has upward probability
one, its downward hinge is instead zero.

At $a$, put $\delta_-=1-2F(a^-)>0$. Its upward hinge is at most
$e^{-c\ell}$, because $1-F(a)<1/2$. When $\ell\ge\delta_-\tau/2$,
this is at most $e^{-c'\tau}$; otherwise \eqref{eq:stars-at-cappeddrift}
bounds the downward hinge by $e^{-c'\tau}+e^{-c(B-\ell)}$.
If $a$ is an extreme support price, its hinge minimum is zero.
At every atom above $a$, the upward hinge is at most $e^{-c_i\ell}$
by negative-drift hitting. Sum these finitely many bounds in
\eqref{eq:stars-at-hinges} and add the existing marginals at $y$ and $y+1$
to obtain \eqref{eq:stars-at-strictlocal}. Index shifts change only
constants; compulsory actions cost zero. This also verifies the stated
omissions of wall terms.

It remains to bound how often sampled \textsf{STARS} visits these costly
boundary regions. Appendix~\ref{app:finite-atomic} proves the required
occupation bounds directly for the sampled trajectory, accounting
for empirical ties and terminal restrictions.

\section{The effect of finitely many scenarios}
\label{app:stars-sampling}

We now account for the finite scenarios used by \textsf{STARS}. Sampling can
change both the current action and all subsequent inventories, so we
track its effect along the inventory path.
Write $\pi^K$ for \textsf{STARS} with $K_t$ fresh shared scenarios at epoch $t$.
As before, $\pi^\circ$ denotes the population-threshold reference rule:
it replaces each scenario average by its expectation, uses thresholds
$\lambda_{\tau_t-1}^B(y)$ at epoch $t$, where $\tau_t=T-t+1$,
and holds at ties. For continuous prices, we compare
the online rewards of the two policies. We then prove a second bound,
valid for arbitrary price laws, on compensation at the inventories
visited by the sampled policy.

\begin{proposition}
\label{prop:stars-continuous-sampling}
Fix an $M$-Lipschitz distribution $F$. Let $K_t\ge1$ be chosen
from the past before observing $p_t$, and use $K_t$ shared fresh i.i.d.\
futures independent of the current price and actual future. Then
\begin{align*}
 \Reg_T(\pi^K)
 \le\Reg_T(\pi^\circ)+M\E\sum_{t=1}^{T-1} K_t^{-1/2}.
\end{align*}
In particular $K_t\ge(T-t+1)^2$ adds at most $M\log T$,
preserving all the continuous-class rates.
\end{proposition}

\begin{proof}
\emph{First bound the value of one additional unit under the reference policy.}
Let $J_\tau(y)$ be the expected reward of the fixed reference rule
$\pi^\circ$ starting at $y$ with $\tau$ decision epochs left,
still ending at zero.
The continuation value $J_\tau$ is defined by following this fixed
online policy.
Run this policy from adjacent feasible inventories on the same prices.
Decreasing thresholds imply that the higher process's action is no
larger than the lower one's. A crossing would require a purchase at
$z$ and a sale at $z+1$, hence a price both below and above the same
threshold $\lambda_{\tau-1}^B(z+1)$ when $\tau$ decisions remain.
Holding at ties prevents this too. Walls
remove outward actions; if the upper state is on the terminal
boundary, it must sell and the lower state cannot buy. The distance
therefore stays one or becomes zero. Until coalescence the actions
agree; at coalescence the upper action is one less. Both end at zero,
so coalescence occurs, and their reward difference is exactly the
price at that epoch. Consequently
\begin{align}
 0\le J_\tau(y+1)-J_\tau(y)\le1,\qquad
 |\mathcal Q_\tau(y,p,x)-\mathcal Q_\tau(y,p,z)|\le|x-z|,
 \label{eq:stars-online-lipschitz}
\end{align}
where $\mathcal Q_\tau(y,p,x)=-px+J_{\tau-1}(y+x)$. The second
inequality follows by summing the first between feasible residual
inventories; the reward difference lies between $-p(x-z)$ and
$(1-p)(x-z)$ when $x>z$. We use $\mathcal Q$ to evaluate a chosen
action followed by the fixed reference rule.

\emph{Next compare the actions at the sampled policy's current inventory.}
Condition on its pre-price history, including $y$ and $K=K_t$, and
write $\tau=\tau_t$ for the number of remaining decisions.
Let $X$ be the reference action at that state and $\widehat X$ the
\textsf{STARS} action. Shared scenarios keep the two empirical thresholds
ordered. Each is an average of $K$ variables in $[0,1]$, so its
variance is at most $1/(4K)$ and its mean absolute error at most
$1/(2\sqrt K)$ by Cauchy--Schwarz. In the next display,
$\lambda_b=\lambda_{\tau-1}^B(y+1)$ and
$\lambda_s=\lambda_{\tau-1}^B(y)$ are the existing population thresholds,
and their hats denote the corresponding shared-scenario sample means.
Independence of the current price
and the $M$-Lipschitz CDF give
\begin{align}
 \E[|X-\widehat X|\mid\text{past}]
 &\stackrel{(a)}{\le}\E[|F(\lambda_b)-F(\widehat \lambda_b)|
                  +|F(\lambda_s)-F(\widehat \lambda_s)|\mid\text{past}]\notag\\
 &\stackrel{(b)}{\le}M\E[|\lambda_b-\widehat\lambda_b|
                    +|\lambda_s-\widehat\lambda_s|\mid\text{past}]
 \stackrel{(c)}{\le}M/\sqrt K.\label{eq:stars-sample-actions}
\end{align}
where (a) integrates the buy and sell disagreement intervals using
independence of the current price and scenarios, (b) applies the
$M$-Lipschitz property of $F$, and (c) uses the two mean absolute
error bounds $1/(2\sqrt K)$ proved above.
Missing channels have zero error, forced actions agree, and ties are
null under the continuous price law.

\emph{Finally sum the reward differences along the sampled path.}
Using the fixed-policy equation $J_\tau(y)=\E_p\mathcal Q_\tau(y,p,X)$, we obtain
\begin{align*}
 &\left|J_\tau(y)-
 \E[-p_t\widehat X+J_{\tau-1}(y+\widehat X)\mid\text{past}]\right|\\
 &\qquad\stackrel{(a)}{\le}
 \E[|\mathcal Q_\tau(y,p_t,X)-\mathcal Q_\tau(y,p_t,\widehat X)|
       \mid\text{past}]\\
 &\qquad\stackrel{(b)}{\le}\E[|X-\widehat X|\mid\text{past}]
 \stackrel{(c)}{\le}M K_t^{-1/2}.
\end{align*}
where (a) uses the fixed-policy equation and
$|\E Z|\le\E|Z|$, (b) applies \eqref{eq:stars-online-lipschitz},
and (c) applies \eqref{eq:stars-sample-actions} at the empirical state.
At the last epoch, the terminal target forces the same action for
both policies, so the terminal summand is zero.
Take expectations and sum along the empirical path. The $J$ terms
telescope to $J_T(0)$ since $J_0(0)=0$. This bounds the difference
of expected online rewards; subtracting from the common hindsight
benchmark proves the proposition. The asserted schedule uses
$\sum_{\tau=2}^T1/\tau\le\log T$.
\end{proof}

The Lipschitz assumption bounds the probability that sampling changes
an action. At an atom, this probability can remain large even when the
threshold estimate is accurate. A change in action can nevertheless
have a small cost if its mean advantage is close to zero. We use this
observation to bound compensation directly, without a continuity
assumption.

\begin{lemma}
\label{lem:sampled-compensation}
Let $F$ be any law on $[0,1]$. Conditional on the pre-price history,
fix a feasible inventory $y$, a continuation length $\tau=\tau_t-1=T-t$, and an integer
$K=K_t\ge1$ chosen from that history. The $K$ shared scenario futures
are i.i.d., independent of the current price and actual future.
Let $\widehat{\partial\bar R}_{\tau,K}^B(y)$ be the expected marginal
compensation of the resulting \textup{\textsf{STARS}} action, averaging over all three
sources of randomness. Let $\partial\bar R_\tau^B(y)$ be the local
compensation of the population-threshold reference rule, as in
Proposition~\ref{prop:stars-compensation}. There is an absolute
constant $C<\infty$ such that
$\widehat{\partial\bar R}_{\tau,K}^B(y)
\le2\partial\bar R_\tau^B(y)+C/K$.
The assertion includes atoms and empirical or population ties.
Forced actions have zero compensation.
\end{lemma}

\begin{proof}
\emph{First separate the sampling error from the reference compensation.}
Condition also on the current price $p$. For an existing buy channel,
put $Z=W-p$, where $W=\Delta V_\tau^B(y+1)$ is the actual future
marginal. Its $K$ scenario copies are $Z_1,\ldots,Z_K$; write
$\bar Z=K^{-1}\sum_{j=1}^K Z_j$,
$A=\E Z^+$, $D=\E(-Z)^+$, and $m=\E Z=A-D$.
These expectations are conditional on the fixed history and price.
Holding at ties makes the buy indicator $\one\{\bar Z>0\}$.
The actual future is independent of the scenarios, so its expected
compensation in this channel is
\begin{align*}
 L_K
 &\stackrel{(a)}{=}
   \E[Z^+-Z\one\{\bar Z>0\}]
 \stackrel{(b)}{=} A-m\Pp(\bar Z>0).
\end{align*}
Here (a) subtracts the chosen channel's gain from hindsight's positive
gain, and (b) uses independence of the actual future and scenarios.
If $m>0$, this equals $D+m\Pp(\bar Z\le0)$; if $m<0$, it equals
$A+(-m)\Pp(\bar Z>0)$; if $m=0$, it equals $A=D$.
The reference channel cost is $A\wedge D$. Thus the additional
compensation is the absolute mean advantage multiplied by the
probability that the scenarios give it the wrong sign.
In particular, a population tie adds no expected compensation,
regardless of the sampled action.

\emph{Next bound the cost of a wrong-sign estimate.}
Suppose first that $m>0$, and put $v=\Var(Z)$.
Since $Z\in[-p,1-p]\subset[-1,1]$, we have
$|Z-m|\le1$ and
$v\le\E Z^2\le\E|Z|=m+2D$.
If $m\le D$, the extra term above is at most $D$.
If $m>D$, then $v\le3m$, and the one-sided Bernstein inequality gives
\begin{align*}
 m\Pp(\bar Z\le0)
 &\stackrel{(a)}{\le}
 m\exp\!\left(-\frac{Km^2}{2(v+m/3)}\right)
 \stackrel{(b)}{\le}m e^{-3Km/20}
 \stackrel{(c)}{\le}\frac{20}{3eK}.
\end{align*}
Step (a) applies Bernstein to the independent centered variables
$Z_j-m$, whose absolute values are at most one and whose variances
equal $v$. Step (b) substitutes $v\le3m$. Step (c) maximizes
$u e^{-3Ku/20}$ over $u\ge0$.
For $v=0$ the wrong-sign probability is zero, so the same bound holds.
If $m<0$, apply this argument to $-Z$: its negative-part expectation
is $A$, and its nonpositive sample-mean event contains
$\{\bar Z>0\}$. Thus in all cases
$L_K\le2(A\wedge D)+20/(3eK)$.

For a sell channel take $Z=p-\Delta V_\tau^B(y)$ instead. Selling again
occurs exactly when its sample mean is positive, and the same proof
applies. Shared scenarios preserve the ordering of the two marginal
thresholds, so the mutually exclusive channel indicators describe a
feasible maximizing action.
Integrating the fixed-price inequalities over $F$, their reference
costs sum to $\partial\bar R_\tau^B(y)$ by
Proposition~\ref{prop:stars-compensation}. Summing the at most two
existing channels gives the lemma, for example with $C=40/(3e)$.
Unavailable channels are omitted. At a compulsory-action state both
policies and hindsight take that action, so compensation is zero;
this includes the terminal epoch $\tau=0$.
\end{proof}

Along the actual \textsf{STARS} trajectory, the tower property and the pathwise
compensation identity therefore give
\begin{align*}
 \Reg_T(\pi^K)
 \le 2\E\sum_{t=1}^T
       \partial\bar R_{\tau_t-1}^B(y_{t-1}^{\pi^K})
       +C\E\sum_{t=1}^T K_t^{-1}.
\end{align*}
For $K_t\ge(T-t+1)^2$, the last sum is at most
$1+\int_1^\infty u^{-2}\,du=2$.
The remaining task is to bound the first sum along the sampled
inventory path. Appendix~\ref{app:finite-atomic} provides this
occupation bound for finite atomic prices.

\section{Finite scenarios for atomic prices}
\label{app:finite-atomic}

We now combine the local compensation bounds with control of the
inventory path to prove the atomic case of
Theorem~\ref{thm:stars-upper}. Throughout this appendix, $y_t,x_t$ are the
\emph{sampled} \textsf{STARS} trajectory, $\tau_t=T-t+1$ is the number
of remaining decisions including the current one, and
$K_t\ge\tau_t^2$ is measurable before the current price arrives.
Thus the continuation after the current action contains $\tau_t-1$
prices. We use $\tau$ for a generic continuation length and substitute
$\tau=\tau_t-1$ when applying a continuation bound at time $t$.
All constants may depend on the fixed finite law and $\rho$, but not
on the sampling schedule or $T$.
We take $T\ge2$ below; at $T=1$, the only feasible action is to hold
and the regret is zero.

Lemma~\ref{lem:sampled-compensation} bounds the loss from one decision:
the conditional charge is at most
$2\partial\bar R_{\tau_t-1}^B(y_{t-1})+C/K_t$. Since
$\sum_tK_t^{-1}\le2$, we need only bound the reference charge
\emph{at these sampled states}. We first handle a median gap, then a
median strictly inside an atom.

\subsection{A median gap: inventory protection under sampled thresholds}

Use $a^-,a^+,\mu^+,r,s$ from Appendix~\ref{app:stars-atomic}.
When $\mu^+<1/2$, the reference local bound there has profile
$f_{\bar\tau}(j)=\bigl(r^j\wedge((j+1)/\sqrt{\bar\tau})\bigr)+s^j$ on a block
$\bar\tau\le\tau_t-1<2\bar\tau$. Thus $\bar\tau$ specifies the
remaining-time scale of the block. We show that the same block occupation bound holds
for sampled \textsf{STARS}. To do so, we identify inventory regions where the
empirical thresholds preserve the direction of the reference
decisions with high probability. We will establish the strict
inequality above $a^-$ on the region defined below, which is large
enough for the occupation bound.

Fix $c_*=7/2$ and, for sufficiently large $\tau$, put
$H_{\tau}=\lfloor\sqrt{c_*\tau\log \tau}\rfloor$ and
$J_{\tau}=\lceil A\log \tau\rceil$, with a fixed sufficiently large $A$.
Take a fixed cutoff $\tau_0\ge2$ large enough that
$H_{\tau}+2<\bigl(\tau\wedge(B/2)\bigr)$ for $\tau\ge \tau_0$; this is possible uniformly in
$T\ge \tau$, since $B\ge\rho \tau$. Constants below cover smaller $\tau$.

\Needspace{15\baselineskip}
\begin{lemma}
\label{lem:sampled-atomic-regions}
Fix the finite atomic law $F$ with positive median gap and
$\rho\in(0,1)$ as above. There are constants $A,\tau_0,c,C>0$,
depending only on $F$ and $\rho$, such that, with
$\delta_{\tau}=C\exp(-c \tau^{1/5})$, the following bounds hold uniformly
over existing indices, $K\ge(\tau+1)^2$, and $\tau\ge \tau_0$.
They also hold conditionally at the current inventory, because the
scenarios are fresh and independent of the past.
\begin{enumerate}[label=$(\alph*)$]
\item If $y\ge J_{\tau}$, then
$\Pp(\widehat\lambda_{\tau}^B(y)\ge a^+)\le\delta_{\tau}$.
\item If $y\le H_{\tau}$, then
$\Pp(\widehat\lambda_{\tau}^B(y)\le a^-)\le\delta_{\tau}$.
\item If $\mu^+<1/2$, the logarithmic buffer $h_{\bar\tau}$ in
Lemma~\ref{lem:stars-at-thresholds} can be chosen so that, uniformly
for $\bar\tau\le \tau<2\bar\tau$ and $y\le h_{\bar\tau}+1$,
$\Pp(\widehat\lambda_{\tau}^B(y)\le a^+)\le C e^{-c\bar\tau}$.
\end{enumerate}
\end{lemma}

\begin{proof}
\emph{Above the buffer, sampling rarely prevents a sale.}
For $W=\Delta V_{\tau}^B(y)$, the gap and fair-walk representation give
\begin{align*}
 a^+-\E W
 &\ge\Delta\Pp(W\le a^-)-\E(W-a^+)^+ \ge c\bigl((y/\sqrt \tau)\wedge1\bigr)-r^y.
\end{align*}
The second line uses Lemma~\ref{lem:stars-fair}; capping the queue
only increases its failure probability. If $\mu^+=1/2$, the upward
hinge is zero and the $r^y$ term is omitted. Choose $A$ so that,
at $y=J_{\tau}$, the right side is at least $c\log \tau/\sqrt \tau$ for
large $\tau$. Monotonicity in $y$ gives the same lower bound for larger
indices. Hoeffding for $K$ independent variables in $[0,1]$ then
bounds the first probability by $\exp(-c \tau\log^2 \tau)$.

\emph{Below $H_{\tau}$, the scenarios detect the value of buying.}
The fair walk has the moderate-deviation bound
\begin{align*}
 \Pp_{1/2}(\sigma_y\le \tau)
 \ge c\,\tau^{-c_*/2}/\sqrt{\log \tau},\qquad y\le H_{\tau}.
\end{align*}
We prove this estimate directly, keeping track of its polynomial
exponent because it determines the required sample size.
For even $\tau=2m$ and parity-admissible $z=2\ell$, the ratio of the
binomial mass at $z$ to its central mass is
$\prod_{i=1}^\ell(m-i+1)/(m+i)$. Using
$\log(1-u)\ge-u-u^2$ for $0\le u\le1/2$ and
$\log(1+u)\le u$ gives a log ratio at least
$-\ell^2/m-C\ell^3/m^2$.
For odd $\tau=2m+1$ and $z=2\ell+1$, the corresponding product is
$\prod_{i=1}^\ell(m-i+1)/(m+i+1)$, with log ratio at least
$-\ell(\ell+1)/m-C(\ell+1)^3/m^2$.
Thus in either parity, uniformly for $z=O(\sqrt{\tau\log \tau})$,
the mass at $z$ is at least
$c \tau^{-1/2}\exp(-z^2/(2\tau)-o(1))$; the central mass is bounded
below by $c/\sqrt \tau$ by the binomial estimate in
Lemma~\ref{lem:stars-fair}.
There are at least $c\sqrt{\tau/\log \tau}$ admissible positions between
$H_{\tau}$ and $H_{\tau}+\lfloor\sqrt{\tau/\log \tau}\rfloor+2$.
For each, $z^2/(2\tau)\le(c_*/2)\log \tau+C$. Summing their masses
and using $\{S_{\tau}\ge y\}\subset\{\sigma_y\le \tau\}$ proves the bound.

For $y\le H_{\tau}$, the capacity comparison and the positive drift of
the count immediately below $a^-$ now give
$q:=\Pp(W>a^-)\ge c \tau^{-7/4}/\sqrt{\log \tau}$ and
$d:=\Pp(W<a^-)\le C e^{-c\tau}$.
The capacity error $CT^2e^{-c_\rho T}$ is absorbed into
$Ce^{-c'\tau}$ uniformly in $T\ge \tau$.
Put $X=W-a^-$. Its positive values are at least $\Delta$, so
$m:=\E X\ge\Delta q-d\ge\Delta q/2$ beyond a fixed cutoff.
Also $\Var(X)\le\E X^2\le q+d\le2q$ and $m\le q$.
Bernstein, with centered increments bounded by one, yields
\begin{align*}
 \Pp(\widehat\lambda_{\tau}^B(y)\le a^-)
 &\le\exp\!\left(-\frac{Km^2}{2(\Var(X)+m/3)}\right)
 \le e^{-cKq}
 \le\exp\!\left(-\frac{c \tau^{1/4}}{\sqrt{\log \tau}}\right).
\end{align*}
This is at most $Ce^{-c'\tau^{1/5}}$ after changing constants. Since
the preceding estimate decreases with $K$, the bound holds for every
$K\ge(\tau+1)^2$.

\emph{Below the logarithmic buffer, the upper median atom induces a
purchase.}
The proof of Lemma~\ref{lem:stars-at-thresholds} gives
$\lambda_{\tau}^B(y)-a^+\ge c_1r^y-C_2y/\sqrt \tau-C_3\varepsilon_T$
for $y=O(\log \tau)$. Choose the fixed constant in $h_{\bar\tau}$ large enough
that at $y=h_{\bar\tau}+1$ this is at least $c\log \bar\tau/\sqrt{\bar\tau}$.
Concavity extends the margin to smaller indices. Hoeffding and
$K\ge \bar\tau^2$ give the last claim.
\end{proof}

\paragraph{Use the sampled decision probabilities to bound visits.}
Suppose $\mu^+<1/2$. On a block $I$ with $\bar\tau\le\tau_t-1<2\bar\tau$, put
$h=h_{\bar\tau}$ and $d_{\bar\tau}=Ce^{-c\bar\tau}$. Below $h$, sampled purchases have
conditional probability at least $p_*-d_{\bar\tau}$, and sales from the next
state have probability at most $q_*+d_{\bar\tau}$. This follows by conditioning
on the fresh scenarios in Lemma~\ref{lem:sampled-atomic-regions} and
then on the independent current price. At $y=h$, the conditional
expected upward correction is at least $2\mu^+(1-d_{\bar\tau})$.

Let $\nu_j,A_j$ be the visits and upward corrections defined in
Appendix~\ref{app:stars-atomic}, now for the sampled trajectory.
Crossing cancellation gives
$\nu_j\le r_{\bar\tau}\nu_{j+1}+1/(p_*-d_{\bar\tau})$ below $h$, where
$r_{\bar\tau}=(q_*+d_{\bar\tau})/(p_*-d_{\bar\tau})=r+O(d_{\bar\tau})$.
Since $h=O(\log \bar\tau)$, $r_{\bar\tau}^\ell\le C r^\ell$ for $0\le\ell\le h$,
and $r_{\bar\tau}$ is bounded away from one. Thus
$\nu_j\le C r^{h-j}\nu_h+C$.
Above $h$, the pointwise correction inequalities giving
\eqref{eq:stars-at-correctioncross} are unchanged by sampling.
In particular, $\nu_j\le\nu_h+2(j-h)+4\sum_{i=h}^jA_i$.

Multiply these inequalities by $f_{\bar\tau}(j)$ and sum. Below $h$, use
$f_{\bar\tau}(j)\le(j+1)/\sqrt{\bar\tau}+s^j$ and the geometric sums
$\sum_{j\le h}r^{h-j}(j+1)\le(h+1)/(1-r)$ and
$\sum_{j\le h}r^{h-j}s^j\le r^h/(1-s/r)$.
Above $h$, use $f_{\bar\tau}(j)\le2r^j$ and interchange the finite sums.
These steps give
\begin{align*}
 \sum_{t\in I}\E f_{\bar\tau}(y_{t-1})
 &\le C+C\left((h+1)/\sqrt{\bar\tau}+r^h\right)\nu_h
       +Cr^h\sum_iA_i\\
 &\stackrel{(a)}{\le}
 C+C\frac{\log \bar\tau}{\sqrt{\bar\tau}}\E\sum_{t\in I}u_t.
\end{align*}
Step (a) uses $2\mu^+(1-d_{\bar\tau})\nu_h\le\sum_iA_i$ and the buffer
estimates $h\asymp\log \bar\tau$, $r^h\asymp\log \bar\tau/\sqrt{\bar\tau}$.
The constant term includes the uniformly bounded
$(h+1)^2/\sqrt{\bar\tau}$. This is the required sampled block bound.

\subsection{Bounding the accumulated corrections}

The block estimate reduces our task to controlling the upward
corrections to the static median-sign rule, whose proposed trades form
a symmetric random walk. We separate the
corrections that occur in the protected inventory regions from the
rare exceptions. This will allow us to compare their accumulated
number with the fair walk's running minimum.

Write $\xi_t=\one\{p_t\le a^-\}-\one\{p_t\ge a^+\}$,
$S_t=\sum_{i\le t}\xi_i$, and
$u_t=(x_t-\xi_t)^+$, $\ell_t=(\xi_t-x_t)^+$.
The signs remain i.i.d.\ fair and independent of the pre-price history.
With $U_t=\sum_{i\le t}u_i$, $L_t=\sum_{i\le t}\ell_i$, the exact
identity is $y_t=S_t+U_t-L_t$.

For $\tau_t-1\ge \tau_0$, call an upward correction \emph{exceptional} if
$y_{t-1}>J_{\tau_t-1}$, and a downward correction exceptional if
$y_{t-1}+1\le H_{\tau_t-1}$.
Declare all corrections with $\tau_t-1<\tau_0$ exceptional. For $\tau\ge \tau_0$,
Lemma~\ref{lem:sampled-atomic-regions} bounds each exceptional event's
conditional probability by $\delta_{\tau}$. Indeed, an upper-half price
must be sold when the sampled sell threshold is below $a^+$, and a
lower-half price must be bought when the sampled buy threshold is
above $a^-$ and $y+1\le H_{\tau}$; all those purchases are feasible.
Set $\delta_{\tau}=1$ below the cutoff. Consequently
$\sum_{\tau\ge0}(\tau+1)^q\delta_{\tau}<\infty$ for every fixed $q$.

Let $U_t^{\rm bad}$ count exceptional upward corrections with their
magnitudes. Then $\E U_T^{\rm bad}\le2\sum_{\tau}\delta_{\tau}\le C$.
For each $t$, after the most recent nonexceptional upward correction
before or at $t$, all upward corrections through $t$ are included in
$U^{\rm bad}$. The inventory identity and
monotonicity of $L$ therefore give the pathwise bound
\begin{align*}
 \max_{t\le T}y_t\le J_T+2+2\max_{t\le T}|S_t|+U_T^{\rm bad}.
\end{align*}
If there is no such correction, the bound follows directly from
$y_t\le S_t+U_t^{\rm bad}$. Since $J_T=O(\log T)$ and $B\ge\rho T$,
maximal Hoeffding and Markov's inequality imply
\begin{align*}
 T\Pp(\max_{t\le T}y_t>B/2)
 \le CT e^{-cT}+C(T/B)\E U_T^{\rm bad}\le C.
\end{align*}
Thus the contribution of visits above half capacity is bounded by a
constant when summed over the horizon.

\paragraph{Separate runs by downward corrections.}
Put $M_t=\max_{j\le t}(-S_j)$ and $D_t=U_t-M_t$.
Nonnegative inventory and nondecreasing $U$ imply $D_t\ge0$.
We show $\E D_t\le C\log T$ for every deterministic $t\le T$.
Before the first downward correction, the last nonexceptional upward
correction and the inventory identity bound the total upward
corrections by $M_t+J_T+2+U_t^{\rm bad}$.

Every subsequent interval begins just after a downward-correction
epoch $j$ and ends at the next such epoch (or the horizon). Within
that interval there are no further downward corrections. If $i$ is
its last nonexceptional upward-correction epoch before $t$, then
$y_{i-1}\le J_{\tau_i-1}\le J_{\tau_j-1}$. Put
$Z_j=\max_{0\le r\le\tau_j-1}(-S_{j+r}+S_j)$.
The inventory identity then gives
\begin{align*}
 \text{nonexceptional upward corrections in this interval}
 \le\left(J_{\tau_j-1}+2-y_j+Z_j\right)^+.
\end{align*}
Exceptional upward corrections before $i$ enter with a negative sign
in this upper bound and may be dropped; later ones are charged to
$U^{\rm bad}$. If no such $i$ exists, the interval contributes zero.

For a nonexceptional downward correction, $y_j\ge H_{\tau_j-1}-2$.
Conditional on $\mathcal F_j$, the future fair signs give
$\Pp(Z_j>z\mid\mathcal F_j)\le e^{-z^2/(2(\tau_j-1))}$ for $z\ge0$
by the exponential-martingale maximal inequality.
Take $a_{\tau}=H_{\tau}-J_{\tau}-4$. For large $\tau$, $a_{\tau}\ge\sqrt \tau$ and
$a_{\tau}^2/(2\tau)\ge(13/8)\log \tau$, since
$H_{\tau}^2/(2\tau)\sim(7/4)\log \tau$ and $J_{\tau}=O(\log \tau)$.
Tail integration therefore gives
\begin{align*}
 \E[(Z_j-a_{\tau_j-1})^+\mid\mathcal F_j]
 &\le\int_{a_{\tau_j-1}}^\infty e^{-z^2/(2(\tau_j-1))}\,dz
 \le\frac{\tau_j-1}{a_{\tau_j-1}}e^{-a_{\tau_j-1}^2/(2(\tau_j-1))}
 \le C (\tau_j-1)^{-9/8}.
\end{align*}
The middle inequality follows by replacing $1$ in the integrand by
$z/a_{\tau_j-1}$ and integrating the derivative of the exponential.
The estimate bounds the upward corrections over the entire run, and
its sum over $j$ is finite. A run attached to an exceptional
downward epoch instead costs at most $2(\tau_j-1)$; its expected contribution
is at most $2(\tau_j-1)\delta_{\tau_j-1}$. Summing the initial interval, all later
intervals, and $U^{\rm bad}$ proves the claimed bound on $\E D_t$.
For each deterministic $j$, we condition on $\mathcal F_j$ and bound
the cost using the entire future fair-walk maximum. Summing these
bounds covers all possible downward-correction epochs.

We now incorporate the compensation weights from the block estimate.
For deterministic nonnegative weights nondecreasing in chronological
time, substitute $U_t=M_t+D_t$ and sum by parts:
\begin{align*}
 \sum_{t=1}^{t_*}w_tu_t
 &=\sum_{t=1}^{t_*}w_t(M_t-M_{t-1})+w_{t_*}D_{t_*}
       -\sum_{t<t_*}(w_{t+1}-w_t)D_t.
\end{align*}
Since $D_t\ge0$, taking expectations gives
\begin{align*}
 \E\sum_{t=1}^{t_*}w_tu_t
 &\le\sum_{t=1}^{t_*}w_t\E(M_t-M_{t-1})
       +w_{t_*}\E D_{t_*}.
\end{align*}
The fair-walk record estimate is
\begin{align*}
 \E(M_t-M_{t-1})
 &\stackrel{(a)}{=}\tfrac12\Pp(\max_{j\le t-1}S_j=0)
 \stackrel{(b)}{=}\tfrac12\,2^{-(t-1)}
       \binom{t-1}{\lfloor(t-1)/2\rfloor}
 \stackrel{(c)}{\le}C/\sqrt t.
\end{align*}
Here (a) requires an independent downward sign at the previous
minimum and reverses the preceding signs, (b) is fair reflection,
and (c) is the central-binomial bound in Lemma~\ref{lem:stars-fair}.
If $\tau_t-1\ge \tau_0$ and $w_t\le C\log(\tau_t-1)/\sqrt{\tau_t-1}$,
this estimate and $\sum_{t=1}^{T-1}[t(T-t)]^{-1/2}\le4$ bound the first term by
$C\log T$. The weights are uniformly bounded, so the second term
has the same order. This gives the required weighted-correction
bound, with only a fixed number of final epochs excluded.

\paragraph{Finish the median-gap case.}
For $\mu^+<1/2$, partition $\tau_t-1\ge \tau_0$ into dyadic blocks and use
Lemma~\ref{lem:sampled-compensation}, the local bound
\eqref{eq:stars-at-local}, and the sampled block bound above.
There are $O(\log T)$ blocks. Their correction weights
$\log \bar\tau/\sqrt{\bar\tau}$ are nondecreasing in chronological time beyond
a fixed cutoff and bounded by $C\log(\tau_t-1)/\sqrt{\tau_t-1}$.
The weighted estimate bounds their total cost by $C\log T$.
The terms $\sum_{\tau} e^{-c\tau}$, $T\varepsilon_T$,
$T\Pp(\max y_t>B/2)$, $\sum_tK_t^{-1}$, and the final $\tau_0$ epochs
each cost a constant. Hence $\Reg_T(\pi^K)\le C\log T$.
If $\mu^+=1/2$, the stronger local bound in
Lemma~\ref{lem:stars-at-local} omits $f_{\bar\tau}$ altogether, so the same
ceiling and sampling estimates even give $\Reg_T(\pi^K)\le C$.

\subsection{A strict median atom: protecting logarithmic boundary regions}

When the median lies strictly inside an atom, local compensation
decays exponentially away from either boundary, up to an exponentially
small remaining-horizon term. It therefore suffices
to control occupation in logarithmic-width boundary regions. We show
that sampled \textsf{STARS} has a fixed drift away from each costly boundary
throughout the corresponding region, and then apply crossing balance.

Suppose $F(a^-)<1/2<F(a)$. If prices above $a$ exist, put
$r_+=(1-F(a))/F(a)\in(0,1)$; otherwise omit every lower-wall term.
If prices below $a$ exist, put
$r_-=F(a^-)/(1-F(a^-))\in(0,1)$; otherwise omit every upper-wall term.
The support-sensitive version of
\eqref{eq:stars-at-strictlocal} is
\begin{align*}
 \partial\bar R_{\tau}^B(y)
 \le C\{r_+^y+r_-^{B-y}+e^{-c\tau}\}.
\end{align*}
To verify the exponents, at prices above $a$ the upward hinge has
geometric ratio $(1-F(z_i))/F(z_i)\le r_+$.
At prices below $a$, the two Hoeffding cases in the proof of
\eqref{eq:stars-at-strictlocal} give an exponential-in-$\tau$ term and
an upper-reflection term with ratio
$F(z_i^-)/(1-F(z_i^-))\le r_-$.
For an existing marginal at $\ell\in\{y,y+1\}$, if
$\ell\ge[1-2F(a^-)]\tau/2$, the upward hinge at $a$ is exponentially
small in $\tau$; otherwise its downward hinge is at most
$Ce^{-c\tau}+r_-^{B-\ell}$ by \eqref{eq:stars-at-cappeddrift}.
If $a$ is an extreme price, the hinge minimum there is zero.
Summing the finitely many hinges for the existing marginals at $y$ and $y+1$ proves
the display. Shifting an index changes only constants.

Fix a block $\bar\tau\le\tau_t-1<2\bar\tau$, of length at most $\bar\tau$.
In the following threshold bounds, take $\tau=\tau_t-1$ and
$W=\Delta V_{\tau}^B(y)$ for an existing marginal.
For a required lower-wall term choose
$h_+=\lfloor3\log \bar\tau/(2|\log r_+|)\rfloor$.
For $y\le h_++1$, likelihood reversal for the count above $a$ gives
$\Pp(W>a)\ge c r_+^y$; its reversed walk has positive drift and,
because $y=O(\log \bar\tau)$, hits $y$ by $\tau$ except with exponentially
small probability. The capacity error is also exponentially small.
The count immediately below $a$ has positive drift, so
$\Pp(W<a)\le Ce^{-c\bar\tau}$ at these indices by
\eqref{eq:stars-at-cappeddrift}, using $B\ge\rho \bar\tau$.
The nearest upper support spacing is positive. Consequently the
mean of $W-a$ is at least $c r_+^y$ and its variance at most
$C r_+^y$. Bernstein and $K\ge \bar\tau^2$ show
\begin{align*}
 \Pp(\widehat\lambda_{\tau}^B(y)\le a)
 \le e^{-cK r_+^y}\le Ce^{-c\sqrt{\bar\tau}},\qquad y\le h_++1.
\end{align*}

For the upper wall choose
$h_-=\lfloor3\log \bar\tau/(2|\log r_-|)\rfloor$.
At $B-y\le h_-+1$, the negative hinge has a useful lower bound even
when $\tau$ is close to $y$. The capped count immediately below $a$
has stationary probabilities
$\pi_i=(1-r_-)r_-^{B-i}/(1-r_-^{B+1})$,
$0\le i\le B$. Detailed balance verifies this formula, and the
clipped-update coupling is monotone. Starting the count at zero
therefore makes it stochastically no larger than a stationary copy.
It follows that
$q:=\Pp(W<a)\ge\pi_{y-1}\ge(1-r_-)r_-^{B-y+1}$.
Meanwhile $\Pp(W>a)\le r_+^y\le Ce^{-c\bar\tau}$ if upper prices exist,
and is zero otherwise. The nearest lower support spacing gives
$\E(a-W)\ge c q$ and $\Var(a-W)\le Cq$ for large $\bar\tau$.
Bernstein thus gives
\begin{align*}
 \Pp(\widehat\lambda_{\tau}^B(y)\ge a)
 \le e^{-cKq}\le Ce^{-c\sqrt{\bar\tau}},\qquad B-y\le h_-+1,
\end{align*}
at existing indices. Here stationarity is used to bound the auxiliary
count that represents a hindsight threshold; occupation of the
time-varying trading inventory is controlled next.

\paragraph{Bound boundary occupation by crossing balance.}
At lower-wall edges, the outward purchase probability is at least
$p_+=F(a)-d_{\bar\tau}$ and the inward sale probability is at most
$q_+=1-F(a)+d_{\bar\tau}$, where $d_{\bar\tau}=Ce^{-c\sqrt{\bar\tau}}$.
Their ratio is $r_++O(d_{\bar\tau})$; its powers through $h_+$ are at most
$C r_+^j$. Crossing cancellation therefore gives
$\nu_j\le C r_+^{h_+-j}\bar\tau+C$ below $h_+$.
Multiplication by $r_+^j$ and summation give
\begin{align*}
 \sum_{t\in I}\E r_+^{y_{t-1}}
 &\le C+C\bar\tau(h_++1)r_+^{h_+}+\bar\tau r_+^{h_++1}\le C,
\end{align*}
since $r_+^{h_+}\asymp \bar\tau^{-3/2}$.
At the upper wall, use distance $B-y$, outward sale probability
at least $1-F(a^-)-d_{\bar\tau}$, and inward purchase probability at most
$F(a^-)+d_{\bar\tau}$. The same calculation gives
$\sum_{t\in I}\E r_-^{B-y_{t-1}}\le C$.
Terminal restrictions only strengthen these inequalities: if the
upper state of an edge is $\tau+1$, its sale is compulsory and the
purchase from $\tau$ is unavailable; higher states are unreachable.
All lower-wall actions used above are feasible for large $\bar\tau$ because
$h_+=O(\log \bar\tau)<(\bar\tau\wedge B)$.
The reference local charge therefore costs at most a constant per
dyadic block on the sampled path. Lemma~\ref{lem:sampled-compensation}
adds only $C\sum_tK_t^{-1}\le2C$. There are $O(\log T)$ blocks;
the fixed number of remaining epochs costs a constant. A singleton
support has zero regret by terminal balance. This proves the
finite-scenario atomic guarantee in Theorem~\ref{thm:stars-upper}.

\end{document}